\documentclass[aps,prl,twocolumn,superscriptaddress]{revtex4-2}
\usepackage[utf8]{inputenc}
\usepackage{graphics,graphicx,array,afterpage}
\usepackage{grffile}
\usepackage[export]{adjustbox}
\usepackage{lineno}
\usepackage{xcolor}
\usepackage{longtable}
\usepackage{svg}
\usepackage{braket}
\usepackage{multirow}
\usepackage{enumerate}
\usepackage[colorlinks, citecolor=blue, urlcolor=blue]{hyperref}
\usepackage{amsmath}
\usepackage{amssymb}
\usepackage{amsthm}
\usepackage{mathrsfs}
\usepackage{pifont}
\usepackage{hyperref}

\newtheorem{theorem}{Theorem}

\newtheorem{lemma}{Lemma}

\begin{document}

\title{Gaussian Universality and Its Breakdown in Tensor-Network Machine Learning}
\author{Shi-Tuan Wang}
\affiliation{Theoretical Physics Division, Chern Institute of Mathematics and LPMC, Nankai University, Tianjin 300071, China}
\author{Zidu Liu}
\email{zidu.liu@mpq.mpg.de}
\affiliation{Max Planck Institute of Quantum Optics, 85748 Garching, Germany}
\author{Li-Wei Yu}
\email{yulw@nankai.edu.cn}
\affiliation{Theoretical Physics Division, Chern Institute of Mathematics and LPMC, Nankai University, Tianjin 300071, China}

\begin{abstract}

Gaussian-process limits are powerful in describing overparameterized machine learning models, yet their validity in structured tensor-network architectures remains unclear. Here we analytically present a moment-based approach that identifies precise conditions for the emergence and breakdown of Gaussian universality in tensor-network learning models, with a focus on matrix product states. We prove that in the large bond dimension limit, the learning models with both local and global observables converge to Gaussian processes, with explicit finite-size bounds on higher-order moment deviations. Whereas in the large physical dimension limit, the Gaussian universality no longer persists: while the models with local observables retain Gaussian-process behavior, those global cases exhibit persistent non-Gaussian corrections. Our results reveal that Gaussian-process behavior in tensor-network learning is controlled not only by parameter number, but also by architectural scaling, observable locality, and the spectral properties.
\end{abstract}

\maketitle

Gaussian process (GP) plays a pivotal role in characterizing the statistical behavior of high-dimensional machine learning systems~\cite{Neal1996Bayesian,Rasmussen2006Gaussian,Matthews2018Gaussian}. The discovery that infinitely wide neural networks admit GP limits~\cite{Neal1996Bayesian,Matthews2018Gaussian} has revealed a remarkable simplification in machine learning. Under suitable conditions, numerous microscopic architectural details become irrelevant, and the collective behavior of the model is effectively governed by a covariance kernel. This provides a tractable analytical framework for studying a variety of highly complex learning models~\cite{Novak2019Bayesian,Yang2019Wide,Hron2020Infinite}, and has led to a deeper understanding of their function-space statistics. More recently, analogous GP limits have also been established in quantum machine learning, where certain Haar-random quantum neural networks exhibit Gaussian behavior in the large-Hilbert-space limit~\cite{Garcia2025Quantumneural}. These developments suggest that Gaussian descriptions may arise in a much broader class of learning architectures than previously recognized. Yet this also raises a basic question: whether such Gaussian descriptions remain valid in architectures whose parameters are not weakly correlated, but constrained by geometry, locality, and tensor contractions.
\begin{figure}
\centering
\includegraphics[width=\linewidth]{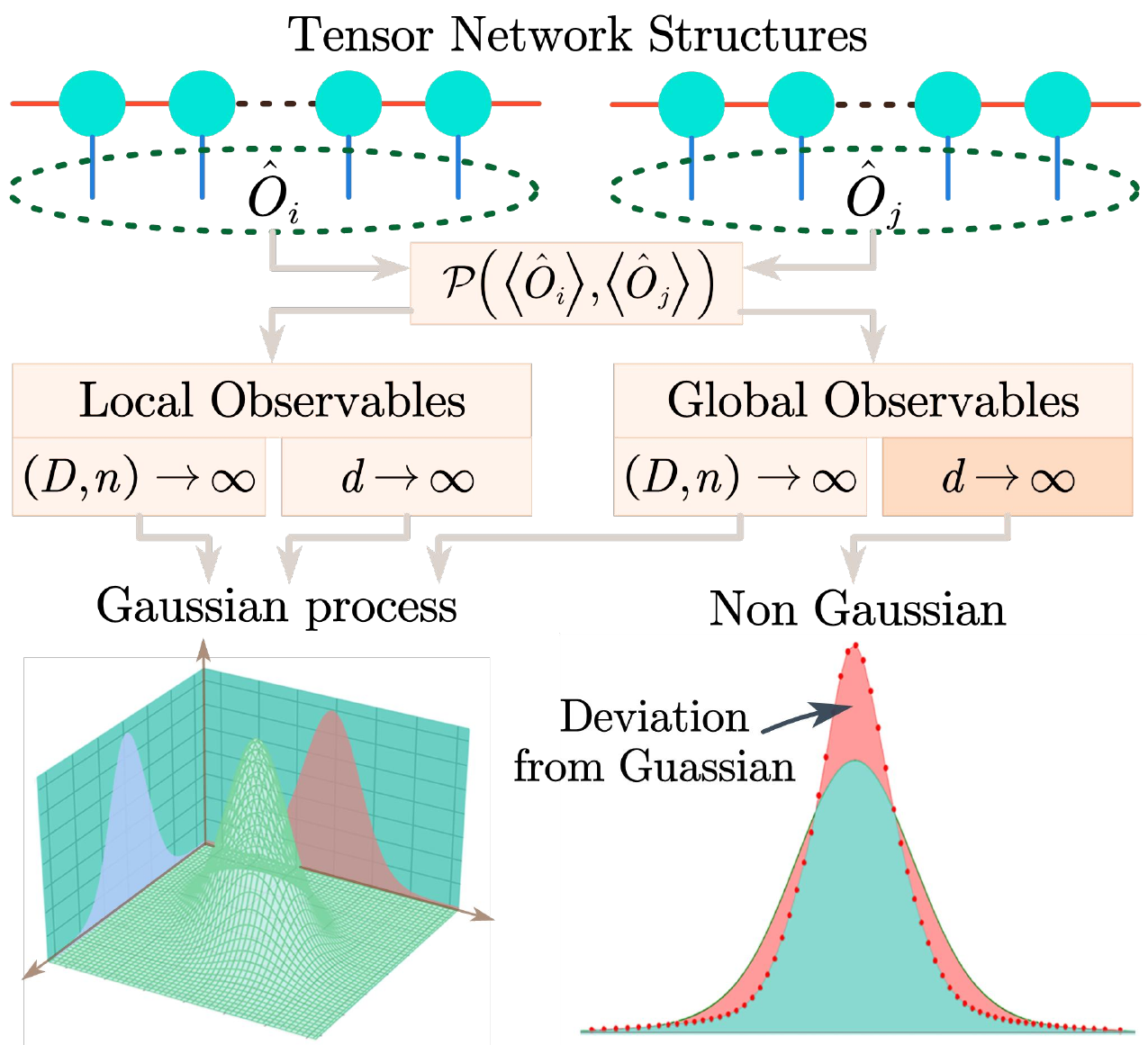}
\caption{Schematic summary of statistical behaviors arising in different tensor-network scaling regimes.  For the local observable outcome $C(\hat{O}_l)=\langle \psi |\hat{O}_l|\psi \rangle$, as either $D$ or $d$ increases, the outcome distribution approaches a Gaussian distribution. For the global observable outcome $C(\hat{O}_g)=\langle \psi |\hat{O}_g|\psi \rangle$, as $D$ increases, the outcome distribution gradually approaches a Gaussian distribution, whereas for increasing $d$, the outcome distribution is no longer Gaussian. For a finite sequence $\{\hat{O}_{ai}\}_{i=1}^{p}$, one has $\mathcal{P}( C(\hat{O}_{a1}),\cdots ,C(\hat{O}_{ap}))\sim \mathcal{N} \left( 0,\varSigma \right)$, $a\in\{l,g\}$.}
\label{fig:mps-gp}
\end{figure}

Tensor-network (TN) machine learning model offers a natural setting for addressing this question. Originally developed as efficient representations of quantum many-body states with area-law entanglement structures~\cite{Eisert2010ColloquiumArea,Roman2014APractical,Orus2019Tensor, Biamonte2019Lectures, Cirac2020Matrix, Banuls2023Tensora,Rieser2023Tensor, Wang2023Tensor}, TNs have been adopted in machine learning~\cite{Ran2023Tensor} as structured parametrizations controlled by bond dimension and contraction geometry, with applications ranging from supervised learning and generative modeling to model compression~\cite{Novikov2015Tensorizing,Cichocki2016Tensor,Cichocki2016Tensor2,stoudenmire_2016_supervised,Stoudenmire2016Advances,Cui2018Smeared,Chen2018Parallelized,Levine2018Deep,Stoudenmire2018Learning,Chen2018Equivalence,Gao2018Quantum,Gao2021Enhancing,Han2018Unsupervised,Cheng2019Tree,Roberts2019Tensornetwork,Liu2019Machine,Hayashi2019Exploring,Levine2019Quantum,Bhatia2019Matrix,Huggins2019Towards,Glasser2020From,Su2020Convolutional,Sun2020Tangent,Sun2020Generative,Sun2020Finite,Sun2020Model,Gao2020Compressing,Wall2021Generative,Chen2021Anend,Cheng2021Supervised,Kardashin2021Quantum,Felser2021Quantuminspired,Wang2021Experimental,Hawkins2021Bayesian,Chen2022Kernelized,Shi2022Clustering,Convy2022Mutual,Metz2023Selfcorrecting,Liao2023Decohering,Ran2023Tensor,Wu2023Tensornetwork,Lopez-Piqueres2023Symmetric,Liujing2023Tensor,Meng2023Residual,Tu2024Generating,Shin2024Dequantizing,Singh2024Federated,Tomut2024CompactifAI,Bermejo2024Quantum,Casagrande2024Tensor,Chen2024Machine,Su2024Language,PozasKerstjens2024privacypreserving,Teng2025Learning,Malz2025Computational,Upreti2026Logic,Fan2026Disentangling}. Unlike conventional wide neural networks, TN architectures possess structured connectivity, locality constraints, and strongly correlated parameters inherited from their underlying tensor geometry. These features would fundamentally modify the statistical mechanisms that underlie GP limits and make the emergence of Gaussian behavior far from obvious. Although recent works have explored GP descriptions of TN models using central-limit-type arguments~\cite{wesel2025tensor,guo2021neuraltangentkernelmatrix}, a systematic understanding of their higher-order statistical structure is still lacking. In particular, it remains largely unexplored whether GP persists across different TN scaling regimes, and how non-Gaussian corrections emerge from the underlying network architecture.

\begin{table*}[t]\label{tab:main_results}
\renewcommand{\arraystretch}{1.3}
\caption{Summary of the statistical regimes and asymptotic behaviors of output vector $\mathscr{C}$.}
\label{tab:main_results}
\begin{ruledtabular}
\begin{tabular}{cccc}
Observable Type & Asymptotic Regime & Gaussian? & Deviation Term \\
\hline
                   & $(D,n) \to \infty, \, D^2 \ll d^n$ & $\checkmark$ & $\mathcal{O} ( k^2/D^2d )+\mathcal{O}(\mathrm{poly}\left( k \right) D^2/d^n)$ \\
\cline{2-4}
Local $\hat{O}_l$  & $(D,n) \to \infty, \, D^2 \gg d^n$ & $\checkmark$ & $\mathcal{O} ( k^2 d^{n/2} / D )+\mathcal{O}(k^2/d^n)  $ \\
\hline
Local $\hat{O}_l$  & $d \to \infty$ & $\checkmark$ & $\mathcal{O} ( k^2D^2/d ) $ \\
\hline
Global $\hat{O}_g$ & $(D, n) \to \infty$ & $\checkmark$ & {$\mathcal{O} (nk!m^{k/2}/D)+\mathcal{O} ( k^2/d^n) $} \\
\hline
Global $\hat{O}_g$ & $d \to \infty$ & $\times$ &  Non-vanishing \footnote {For explicit analytical expressions, see Sec.\ IV D of the Supplemental Material~\cite{GPTNML_Supplementary}}

\end{tabular}
\end{ruledtabular}
\end{table*}

Here we develop a moment-based framework for characterizing the Gaussian universality and its breakdown in TN learning. We consider Haar-random TNs and treat observable expectation values as random variables induced by the tensor ensemble. Through computing higher-order moments, rather than relying solely on central-limit theorems, we obtain explicit deviations from Gaussian statistics at finite bond and physical dimensions~\cite{akhiezer1965classical,kendall1977advanced}. More precisely, we focus on matrix product states (MPS), the canonical one-dimensional tensor-network architecture. We rigorously prove two theorems: (1) In the large-bond-dimension regime, learning models with either local or global observables converge to GPs, with finite-size corrections to the moments governed by explicit powers of the bond dimension and system size; (2) in the large-physical-dimension limit, Gaussian universality becomes observable-dependent: learning models retain GP statistics for local observables, whereas exhibiting persistent non-Gaussian corrections for global observables (Fig.~\ref{fig:mps-gp}). This separation shows that Gaussianity in TN learning is not a generic consequence of large parameter dimension, but is controlled jointly by architectural scaling and observable locality. We further discuss how the observable spectrum influence the Gaussianity. We carry out numerical simulations to show that the predicted Gaussian and non-Gaussian regimes are already apparent in MPS models of modest size.

\emph{The framework---} We focus on the periodic MPS $|\psi\rangle=\sum_{j_1,\ldots,j_n=1}^{d}{\rm tr}[A_{j_1}^{(1)}A_{j_2}^{(2)}\cdots A_{j_n}^{(n)}]|j_1,j_2,\ldots,j_n\rangle$~\cite{Schollwock2011Densitymatrix}, where \(|j_s\rangle\) denotes the local basis state of the \(s\)-th physical site, and $A_{j_s}^{(s)}\in\mathbb C^{D\times D}$ are local tensors with bond dimension $D$.  We use the unitary-embedding representation of MPS, where each local tensor is expressed as  $(A_j^{(s)})_{\alpha\beta}=\langle \alpha,j|U_s|\beta,0\rangle$, $U_s\in U(Dd)$, and \(|0\rangle\) is a fixed auxiliary physical state~\cite{Perez2007Matrix,Gross2010Quantum,Silvano2010Statistical,Silvano2010Typicality,collins2013matrix,Haferkamp2021Emergent,GPTNML_Supplementary}. The unitary embedding MPS $|\psi\rangle$ is not exactly normalized. However, the values of $ \langle\psi|\psi\rangle$ are exponentially concentrated around one \cite{Haferkamp2021Emergent}. Such an exponential concentration allows us to prove our results without regarding the normalization factor as the system size is sufficiently large. We consider a quantum-inspired MPS learning model in which inputs and variational parameters  are encoded into the local tensors of the MPS. The learning output $C(\hat{O})$ is defined by the expectation values of observable $\hat{O}$ on the normalized MPS ansatz, i.e., $C(\hat{O})=\langle\psi|\hat{O}|\psi\rangle/\langle\psi|\psi\rangle$, where  $\hat{O}$ can be either  local ($\hat{O}_l$) or global ($\hat{O}_g$). 

To study the Gaussian behavior of our learning models, we consider the randomly initialized case, where the local tensors of MPS are sampled from independently Haar-random embedding unitaries. For the convenience of presenting our results and without of generality, we suppose the observable $\hat{O}$ be the linear combinations of a finite set of Pauli string operators $\{ \hat{Q}_i\} _{i=1}^{m}$, where $\hat{Q}_i$'s are traceless and obey the Clifford algebra $\{\hat{Q}_i, \hat{Q}_j\} = 2\delta_{ij}I$. 

\emph{GP limit for local observables---}
Local observables provide a natural class of readouts for quantum learning models. They appear in a wide range of settings, including classical-shadow protocols~\cite{Huang2020Predicting}, quantum classifiers~\cite{Lu2020Quantum}, and variational approaches to many-body Hamiltonians~\cite{Carleo2017Solving}. In addition, local observables can avoid the barren plateau problem in shallow or locality-preserving quantum models~\cite{Cerezo2021Cost,Liu2022Presence}. 
We first study the Gaussian behavior of learning models with local observables, which act nontrivially on single physical site. 
\begin{theorem} \label{Theorem_Local}
Define the learning output with local observable $\hat{O}_{li}$ by $C(\hat{O}_{li})$. Let the output vector $\mathscr{C} _l=\{C(\hat{O}_{li})\}_{i=1}^{p}$ be the stochastic process generated by the Haar-random MPS. Then $\mathscr{C}_l$ converges to the GP in the following two cases:  
\begin{itemize}
\item[(1)]For the bond dimension $D\rightarrow\infty$, the physical site $n\rightarrow\infty$, we have 
\begin{equation}
\mathscr {C}_l \sim \mathcal{GP}(0,\varSigma),
\, \varSigma_{ij}=
\begin{cases}
\dfrac{\operatorname{tr}(\hat O_{li}\hat O_{lj})}{D^2d^{n+1}},
 D^2\ll d^n,\\[1ex]
\dfrac{\operatorname{tr}(\hat O_{li}\hat O_{lj})}{d^{2n}},
 D^2\gg d^n.
\end{cases}
\end{equation}

\item[(2)]For the physical dimension $d\rightarrow\infty$, we have
\begin{equation}
\mathscr{C} _l\sim \mathcal{G} \mathcal{P} \left( 0,\varSigma \right), \,\qquad \varSigma _{ij}=\frac{\mathrm{tr(}\hat{O}_{li}\hat{O}_{lj})}{D^2d^{n+1}}.
\end{equation}
\end{itemize}
\end{theorem}

Theorem~\ref{Theorem_Local} shows that learning models with local observables converge to GPs under different large-TN scalings, while their covariances obey distinct asymptotic scaling laws. In the $(D,n)\rightarrow\infty$ limit at fixed $d$, two distinct covariance scalings arise. When $D^2\ll d^n$, the accessible state space is limited by the bond dimension. Since $\operatorname{tr}(\hat O_{li}\hat O_{lj})$ is typically of order $d^n$, the covariance scales polynomially as $\mathcal{O}(1/dD^2)$. By contrast, when $D^2\gg d^n$, the covariance scales as $\mathcal{O}(1/d^n)$, leading to exponential concentration with increasing system size. A different scaling arises for $d\rightarrow\infty$ at fixed $D$, where the covariance scales as $\mathcal{O}( 1/D^2d ) $. Thus, although the Gaussianity persists for learning models with local observables, the covariance of the resulting GP obeys distinct scaling laws across different asymptotic regimes. Beyond these limits, we further derive explicit finite-$(D,d,n)$ corrections to all $k$th-order moments, as summarized in Table~\ref{tab:main_results} and detailed in the Supplemental Material~\cite{GPTNML_Supplementary}.

\emph{Gaussianity and its breakdown for global observables---} Global observables arise naturally in a variety of TN learning tasks, including generative modeling~\cite{Han2018Unsupervised,Liu2023Tensor,Akamatsu2026Plastic}, probabilistic sequence learning~\cite{Miller2021Tensor}, and hybrid classification--generation schemes~\cite{Mossi2025Matrix}, where likelihoods or configuration probabilities are determined by contractions over the entire TN. This motivates us to examine the statistical properties of TN learning models with global observables. We obtain the following theorem.

\begin{theorem}\label{Theorem_Global}
Define the learning output with global observable $\hat{O}_{gi}$ by $C(\hat{O}_{gi})$. Let the output vector $\mathscr{C} _g=\{C(\hat{O}_{gi}\}_{i=1}^{p}$ be the stochastic process generated by the Haar-random MPS. Then $\mathscr{C}_g$ obey the following distributions: 
\begin{itemize}
\item[(1)]For the bond dimension $D\rightarrow\infty$, the physical site $n\rightarrow\infty$, we have 
\begin{equation}
\mathscr{C} _g\sim \mathcal{G} \mathcal{P} \left( 0,\varSigma \right) ,\qquad \varSigma _{ij}=\frac{\mathrm{tr(}\hat{O}_{gi}\hat{O}_{gj})}{d^{2n}}
\end{equation}
\item[(2)] In the regime $d\to\infty$ at fixed $(D,n\ge2)$, the distribution of $C(\hat{Q})$ remains non-Gaussian, where $\hat{Q}$ denotes a Pauli string acting nontrivially on all physical sites. For $nk^2\ll D^2$, its $k$th-order moment deviates from the corresponding Gaussian moment by $\mathcal{O}(nk^2/D^2)$.
\end{itemize}
\end{theorem}

Theorem~\ref{Theorem_Global} characterizes the asymptotic statistics of global observables under different architectural scalings. In the limit $(D,n)\rightarrow\infty$ at fixed $d$, global observables converge to a GP with an exponentially concentrated covariance. By contrast, when $d\rightarrow\infty$ at fixed $D$, the output distribution remains non-Gaussian. This non-Gaussian behavior suggests a natural extension beyond a single Pauli string $\hat{Q}$. For a more general observable of the form $\hat O=\sum_i c_i\hat Q_i$, the corresponding output becomes $C(\hat O_g)=\sum_i c_i C(\hat Q_{gi})$, suggesting that non-Gaussian features may persist for broader classes of global observables. This provides evidence for a breakdown of Gaussian universality for global observables under this scaling.

\emph{Proof of Theorems 1 and 2---}
We sketch the main ideas here and leave the technical details to the Supplemental Material~\cite{GPTNML_Supplementary}. Our proof combines the method of moments with the Cram\'er-Wold device to establish convergence to a GP. We first evaluate the $k$th-order moments $\mu_k=\mathbb{E}_{|\psi\rangle}[(C(\hat O))^k]$, by averaging over the random unitary-embedded MPS $|\psi\rangle$, which is mapped to calculating the partition function of 1D spin model. In two cases of asymptotic limits, $D\to\infty$ (with both local and global outputs), and $d\to\infty$ (with the local output), the leading-order contributions to the term $\mu_k$ coincide with the moments of a standard Gaussian distribution. Thus one can uniquely determine the distribution of the output $C(\hat{O})$ based on the Carleman condition~\cite{akhiezer1965classical}. Then, by showing that an arbitrary finite linear combination of $C(\hat O_i)\in\mathscr C$ converges to a Gaussian distribution, one can establish the convergence of the random sequence $\mathscr C_a$, $a\in\{l,g\}$, via the Cram\'er-Wold device~\cite{Billingsley1995Probability}.  Moreover, we quantify finite-size corrections to the GP limit by bounding the deviations of all $k$th-order moments from their Gaussian counterparts.

The two theorems establish when Gaussianity emerges or breaks down in TN learning models under different architectural scalings and observable localities. In the large-$D$ regime, the enhanced expressive capacity of the MPS allows both local and global outputs to converge to Gaussian distributions. By contrast, when $d$ becomes large at fixed $D$, the limited entanglement capacity of the MPS drives it toward a weakly correlated, nearly product-state structure. For local outputs, the nontrivial contribution remains confined to the observed site, while the remaining sites contribute only trivial factors, so Gaussianity is preserved. For global outputs, however, the contraction involves contributions from all sites and approaches a product of approximately independent random variables, leading generically to non-Gaussian statistics with heavy tails. The method of moments captures this behavior quantitatively: the global distribution interpolates between a standard Gaussian and an $n$th-order Gaussian chaos $Z_n = \prod_{i=1}^n g_i$ with $g_i \sim \mathcal{N}(0, \sigma^2)$. In particular, Gaussianity is recovered asymptotically as $D\to\infty$, whereas for $D=1$ the distribution reduces exactly to $Z_n$.

\emph{Observables and Gaussianity---}
The operator restrictions in Theorems~\ref{Theorem_Local} and~\ref{Theorem_Global} can be understood through spectral delocalization. Consider the Haar-random output $C_{\mathrm{Haar}}(\hat O)=\langle0|U^\dagger\hat O U|0\rangle$, where $U$ is drawn from the Haar measure on $U(N)$. We take $\hat O$ to be traceless, since its identity component only shifts the mean. A sufficient condition for its variance-normalized fluctuations to become Gaussian as $N\to\infty$ is
\begin{equation}
\frac{\|\hat O\|_\infty}{\sqrt{\operatorname{tr}(\hat O^2)}}\to0.
\label{eq:spectral_delocalization}
\end{equation}
This condition ensures that no single eigenvalue dominates the squared spectral weight~\cite{Meckes2020Random}. In contrast, for a rank-one projector, the output multiplied by $N$ converges to a unit-mean exponential distribution~\cite{Porter1956Fluctuations,Anschuetz2025Unified}.

\begin{figure}[htbp] 
    \centering
    \makebox[\linewidth][c]{%
\includegraphics[width=1.05\linewidth]{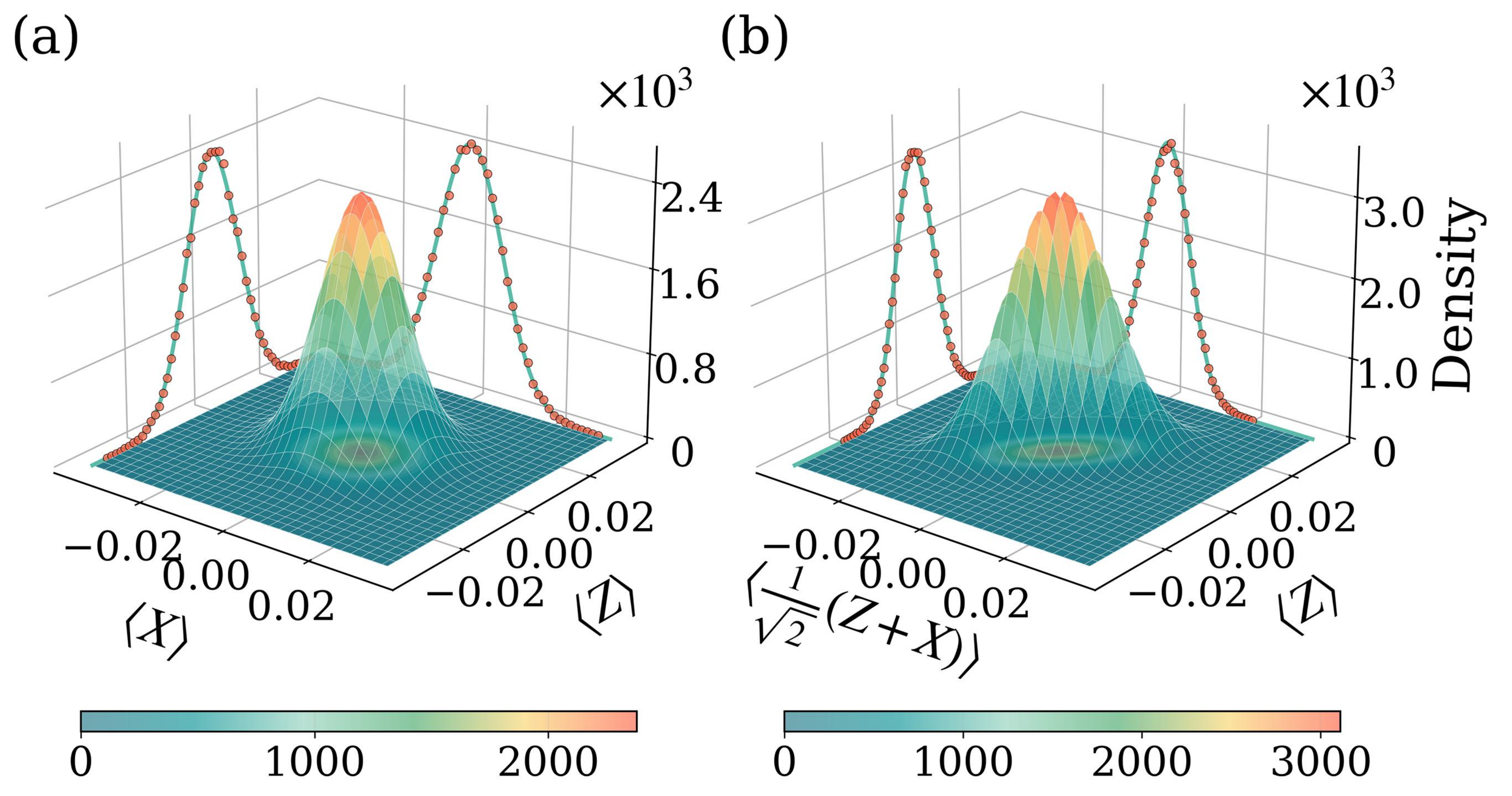}%
    }
\caption{Two-dimensional joint probability densities and their scaled marginal distributions obtained numerically with $n=15$ physical sites, bond dimension $D=128$, and $4\times10^6$ samples. The left panel shows the joint distribution of $\langle X\rangle$ and $\langle Z\rangle$, while the right panel presents the distribution for $\langle \frac{1}{\sqrt{2}}\left( X+Z \right) \rangle$ and $\langle Z\rangle$.}
\label{fig:placeholder}
\end{figure}

For joint Gaussianity, the Cram\'er--Wold criterion requires every fixed real linear combination of the rescaled outputs to converge to the corresponding Gaussian limit. Clifford generators $\{\hat Q_i\}_{i=1}^m$ provide a convenient setting: every nonzero real combination $\hat O=\sum_i c_i\hat Q_i$ satisfies $\hat O^2=(\sum_i c_i^2)I_N$, so the ratio in Eq.~\eqref{eq:spectral_delocalization} equals $N^{-1/2}$. The spectral condition therefore holds throughout their linear span, yielding a GP limit for Haar outputs when the rescaled covariances converge.

Local observables allow more general spectra. Taking $\hat O_l=\hat A\otimes I_R$ for instance, where $\hat{A}$ is a $d$-dimensional traceless observable and $R=d^{n-1}$. We have $\tfrac{\|\hat O_l\|_\infty}{\sqrt{\operatorname{tr}(\hat O_l^2)}}\le R^{-1/2}$. The identity $I_R$ gives each local eigenvalue an $R$-fold degeneracy. Thus Eq.~(\ref{eq:spectral_delocalization}) holds as $R\to\infty$ without further restrictions on the spectrum of $\hat A$.

For MPS outputs, however, spectral delocalization alone does not ensure Gaussianity, as illustrated by the non-Gaussian global outputs in the large-$d$ limit at fixed $(D,n\ge2)$. Extending the local result beyond the Clifford subspace therefore requires a separate analysis of the tensor-network contractions. For arbitrary fixed Hermitian single-site observables, the transfer-matrix analysis shows that the centered, rescaled moments converge to their Gaussian values within the large-$(D,n)$ regimes specified above. This establishes the local Gaussian limit without imposing a Clifford structure on the observable.
\begin{figure}[!t]    
\hspace*{-\parindent}
\includegraphics[width=1.01\columnwidth]{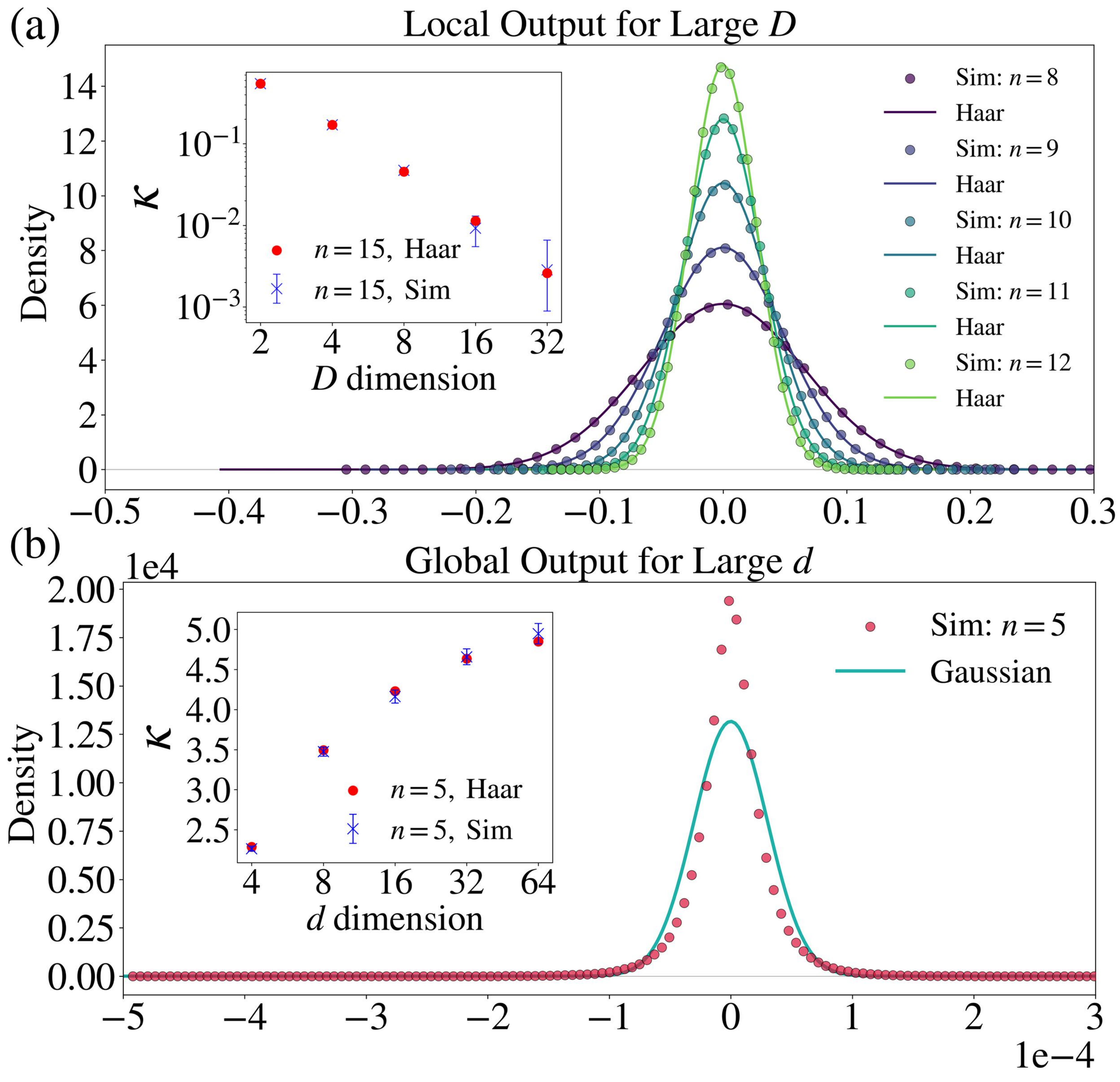}    
\caption{Numerical distributions and theoretical predictions for MPS observables, using approximately $10^7$ samples for each parameter set. Upper panel: local outputs $C(\hat{O}_l)$ for different system sizes $n$ at $D=32$, compared with the predicted Gaussian distribution. The inset shows the excess kurtosis $\kappa$ as a function of $D$. Lower panel: global outputs $C(\hat{O}_g)$ at $d=64$ and $D=2$, compared with a Gaussian distribution of equal variance. The inset shows $\kappa$ as a function of $d$ at $D=2$.}
\label{fig:placeholder2}
\end{figure}

\emph{Numerical results---} To test our theoretical predictions for the Gaussian limits and their breakdown, we numerically sample the local and global outputs, $C(\hat O_l)$ and $C(\hat O_g)$, using the same Haar-random unitary embeddings as in the analytical model. We perform the tensor contractions by using the Python package \texttt{quimb}~\cite{Gray2018quimb}. Fig.~\ref{fig:placeholder} illustrates the joint distribution of two outputs, which is consistent with Gaussian behavior even in the crossover regime $D^2 \approx d^n$. We further assess the predicted Gaussian and non-Gaussian behavior through the excess kurtosis $\kappa$. Both skewness and excess kurtosis vanish for a Gaussian distribution~\cite{Joanes1998,Jarque1980}. Since the output distributions considered here are symmetric about zero, their mean and skewness vanish. We therefore compare the sampled excess kurtosis with the theoretical predictions, testing its convergence to zero in the Gaussian regimes and to a nonzero value in the predicted non-Gaussian regime (Fig.~\ref{fig:placeholder2}).

A key implication of our results is that Gaussianity is not a property of the TN architecture alone, but also depends on the observables and the way the large-parameter limit is taken. Increasing the physical dimension enlarges the local Hilbert space while leaving the dimension of the virtual space unchanged. At fixed $(D,n\ge2)$, this limit makes the rescaled observed transfer matrices Gaussian, yet their products retain non-Gaussian higher-order moments. Single-site outputs instead reduce to linear Gaussian fluctuations as the unobserved transfer matrices approach depolarizing maps. Thus, the same MPS ensemble admits a Gaussian description for local outputs while retaining non-Gaussian global fluctuations. These fluctuations survive in the limiting distribution and cannot be captured by a Gaussian approximation with the correct covariance. Increasing the physical dimension therefore does not necessarily improve a GP approximation. The use of Gaussian priors in TN learning must be justified for the observables and parameter scaling relevant to the task.

\emph{Discussion---}Our findings motivate further investigation of Bayesian inference with TN models. One potential direction is to introduce explicit input encodings and derive the corresponding covariance kernels, connecting the observable-indexed processes studied here to prediction on input data. Incorporating observation noise and finite sampling would then allow the distinct covariance scalings of local and global outputs to be assessed in concrete inference tasks. In addition, extending the analysis to observables with intermediate support and to other TN geometries would clarify how connectivity and observable locality control the survival of non-Gaussian fluctuations. Such extensions could help identify the scope of Gaussian priors in more general TN learning models.

In summary, we have investigated GP limits and their breakdown for random MPS outputs under different parameter scalings. Within the specified large-$(D,n)$ regimes, we establish GP limits for local and global output families in the Clifford subspace, with less restrictive conditions on single-site observables for local outputs. In the large-$d$ limit at fixed $D$ and $n\ge2$, local outputs become Gaussian, whereas fully supported outputs retain non-Gaussian fluctuations. We explain this distinction through the transfer-matrix structure, quantify finite-parameter corrections through moment estimates, and provide numerical evidence supporting the theoretical predictions. Our results provide a theoretical basis for Gaussian priors in TN models while identifying regimes in which higher-order output statistics remain essential.

\emph{Acknowledgments---} This work was supported by the National Natural Science Foundation of China (Grant No. 12375060), and the Fundamental Research Funds for the Central Universities (Grants No. 63253097, No. 63263090).

\bibliography{Reference}

\clearpage
\onecolumngrid
\makeatletter
\setcounter{page}{1}
\setcounter{figure}{0}
\setcounter{equation}{0}
\setcounter{theorem}{0}

\setcounter{secnumdepth}{3}

\renewcommand{\thefigure}{S\@arabic\c@figure}
\renewcommand \theequation{S\@arabic\c@equation}
\renewcommand \thetable{S\@arabic\c@table}

\newpage

\begin{center} 
	{\large \bf Supplemental Material for: Gaussian Universality and Its Breakdown in Tensor-Network Machine Learning}
\end{center}

\renewcommand\thesection{S\Roman{section}}
{
  \hypersetup{linkcolor=blue}
  \tableofcontents
}

\section{Preliminaries and Proof Strategy}

\subsection{Weingarten Calculus for Random MPS }

The general formula for a matrix product state (MPS) with periodic boundary conditions is given by:
\begin{equation}
|\psi \rangle =\sum_{j_1,\cdots ,j_n}{\mathrm{tr}\left( A_{j_1}^{(1)}A_{j_2}^{(2)}\cdots A_{j_n}^{(n)} \right) |j_1,}\cdots ,j_n\rangle.
\end{equation}
where $|j_s\rangle $ represents the local state at the $s$-th site with physical dimension $d$. The $A_{j_s}^{\left( s \right)}$ are $D\times D$ matrices, where $D$ denotes the bond dimension. Following the unitary embedding scheme, we map each $D\times d\times D$ local tensor $A^{(s)}$ to a $Dd$-dimensional unitary matrix $U_s$~\cite{Perez-Garcia2007Matrix,Gross2010Quantum,Silvano2010Statistical, Silvano2010Typicality,collins2013matrix,Haferkamp2021Emergent} as illustrated in Fig.~\ref{fig:MPS} \cite{Perez-Garcia2007Matrix,Mele2024introduction,Haferkamp2021Emergent}.

\begin{figure}[h]
    \centering
    \includegraphics[width=0.5\linewidth]{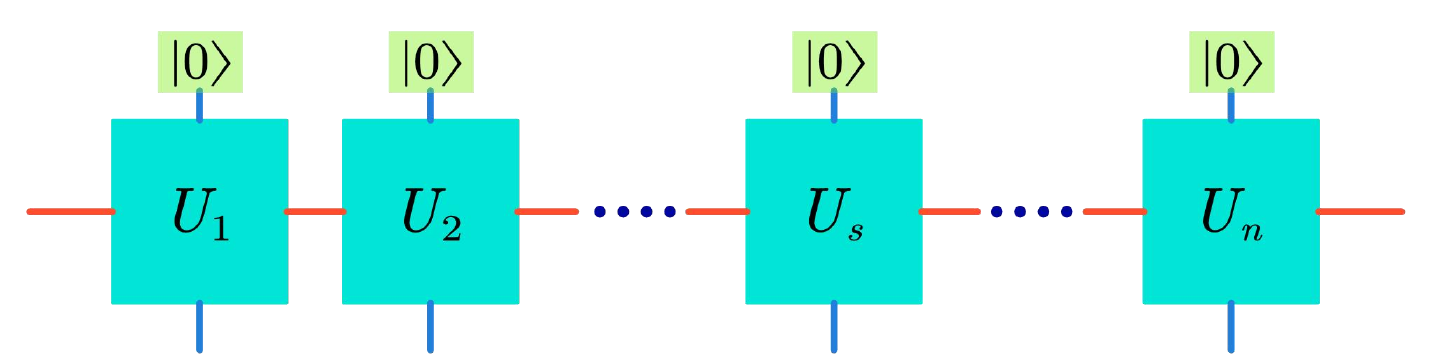}
    \caption{Matrix product state}
    \label{fig:MPS}
\end{figure}

Here, the unitary matrix \( U^{(s)} \in U(Dd) \) (\( s = 1, 2, \dots, n \)) can be expressed as \( C^d \otimes C^D \). Each connecting line in Fig.~\ref{fig:MPS} represents the contraction of a pair of subscripts. In our proof, the MPS adopts periodic boundary conditions, meaning the first and last matrices are also connected to form a ring structure. The dimension of the red line connecting the two \( U \) matrices is \( D \) , while the dimension of the blue line connecting \( U \) and \( | 0 \rangle \) is \( d \).

The expected value in the form of \( U_{i_1 j_1} \cdots U_{i_k j_k} \overline{U}_{i_1' j_1'} \cdots \overline{U}_{i_k' j_k'} \) under the Haar measure is given by \cite{Renes2004symmetric,Dankert2009exact,Harrow2009random,Collins2006Integration}:
\begin{equation}\label{Wgf}
\mathbb{E} \left( U_{i_1j_1}\cdots U_{i_kj_k}\overline{U}_{i_1\prime j_1\prime}\cdots \overline{U}_{i_k\prime j_k\prime} \right) =\sum_{\pi ,\sigma \in S_k}{\delta _{i_1,i_{\pi (1)}\prime}}\cdots \delta _{i_k,i_{\pi (k)}\prime}\delta _{j_1,j_{\sigma (1)}\prime}\cdots \delta _{j_k,j_{\sigma (k)}\prime}\,\mathrm{Wg}\left( \sigma \pi ^{-1},N \right),
\end{equation}
where \( \mathrm{Wg} \left( \sigma \pi^{-1}, N \right) \) is known as the Weingarten factor and $N$ is the dimension of $U$. This representation shows that the integration with respect to the Haar measure effectively connects the subscripts on each side of the matrices according to the permutations in the symmetric group $S_k$ and multiplies them by a Weingarten coefficient.

For the case where \( k = 2 \), the connection of two subscripts includes two possibilities: a symmetric connection \( S \) (i.e., \( U_1 \) connected to \( \overline{U}_1 \), \( U_2 \) connected to \( \overline{U}_2 \)) and an antisymmetric connection \( A \) (i.e., \( U_1 \) connected to \( \overline{U}_2 \), \( U_2 \) connected to \( \overline{U}_1 \)), which correspond to the identity permutation \( e \) and (12) in the second-order symmetric group~\cite{Mele2024introduction}.

\begin{figure}[htbp]
    \centering
    \includegraphics[width=0.4\linewidth]{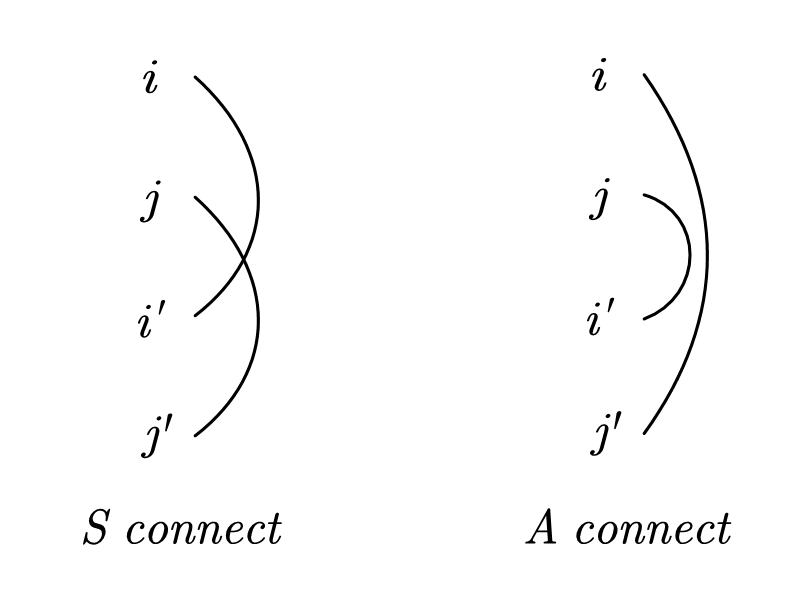}
    \caption{Subscript connection modes in the second-order case}
    \label{fig:121212}
\end{figure}

For the third-order case, the third-order symmetric group has six elements: the identity permutation \( e \), transposition (12), (13), (23), and 3-cycles (123) and (132), corresponding to six types of connections on each side of the matrix, as shown in Fig. \ref{fig:123}.

\begin{figure}[htbp]
    \centering
    \includegraphics[width=0.6\linewidth]{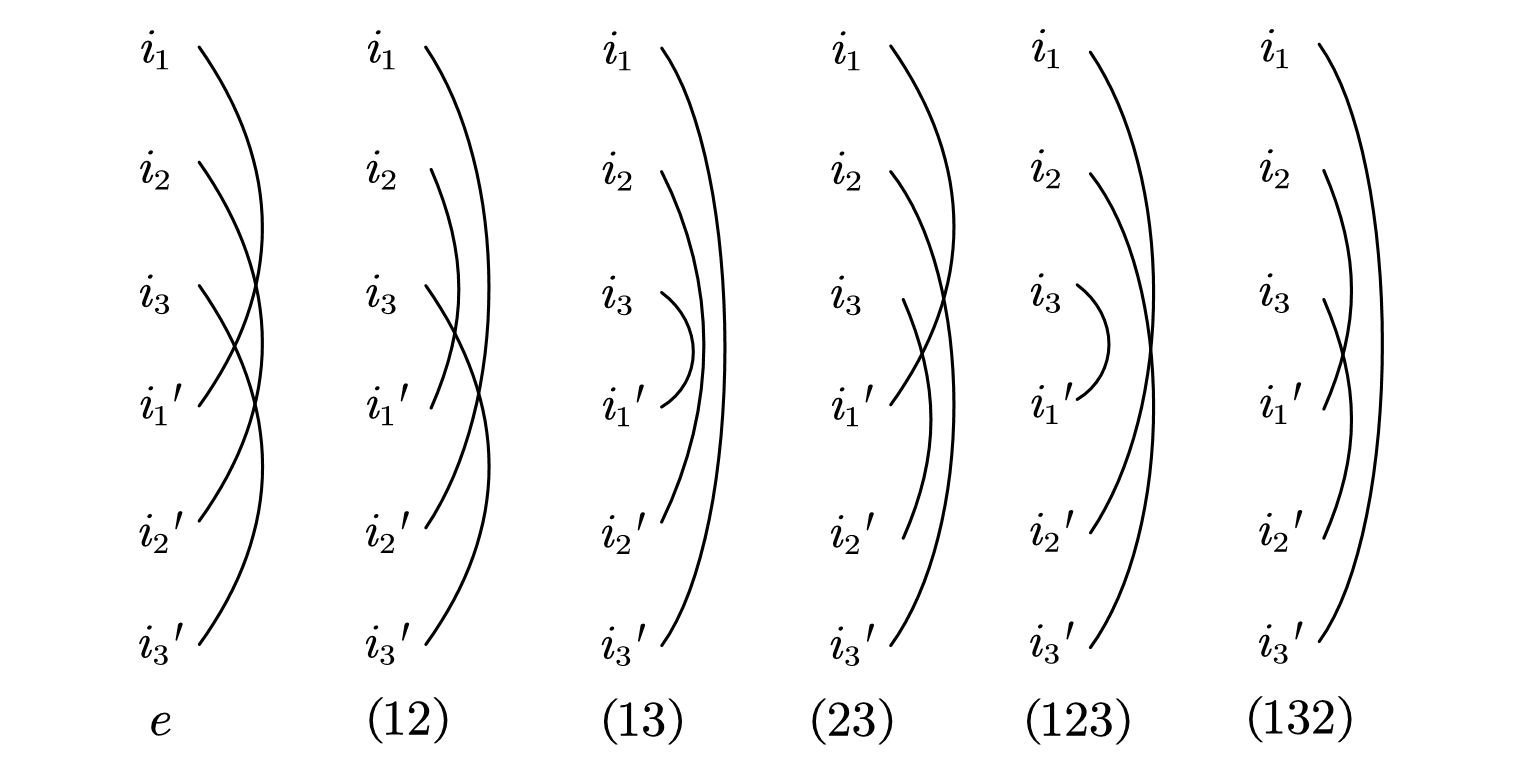}
    \caption{Subscript connection modes in the third-order case}
    \label{fig:123}
\end{figure}

The Weingarten matrix $\mathrm{Wg}$ is the inverse of a Gram matrix $G_N$ (with matrix elements $N^{\#(\pi \delta ^{-1})}$). Its matrix elements depend only on the conjugacy class corresponding to \( \pi \delta^{-1} \), which implies that \( \mathrm{Wg}(\delta \pi^{-1}, N) = \mathrm{Wg}(\pi \delta^{-1}, N) \). When the dimension $N$ tends to infinity, for a fixed order \( k \), the Weingarten factors asymptotically approach~\cite{Mele2024introduction,Collins2003moments}:
\begin{equation} \label{Wgj}
\mathrm{Wg} \left( \pi, N \right) = \mathrm{Moeb} \left( \pi \right) N^{-2k + \# \pi} + \mathcal{O} \left( N^{-2k + \# \pi - 2} \right),
\end{equation}
where $\#(\pi)$ represents the number of cycles in \( \pi \). For the identity permutation \( e \) in a \( k \)th-order symmetric group, \( \#e = k \), for a 2-cycle (12), \( \#(12) = k-1 \), and for a \( k \)-cycle \( (12\cdots k) \), \( \#(12\cdots k) = 1 \). According to the properties of symmetric groups, the inverse of a symmetric group element does not change its cycle number, so \( \#(\pi) = \#(\pi^{-1}) \). In fact, there is an exact bound for Weingarten the factor, which is given by the following lemma~\cite{Collins2017weingarten}.  

\begin{lemma} [Bound for Weingarten Factor]
For any $\sigma \in S_k$ and $N > \sqrt{6}k^{7/4}$, 
\begin{equation}\label{upperboundofwg}
\frac{1}{1-\frac{k-1}{N^2}}\le \frac{N^{2k-\#\sigma}\mathrm{Wg(}\sigma ,N)}{\mathrm{Moeb(}\sigma )}\le \frac{1}{1-\frac{6k^{7/2}}{N^2}}.
\end{equation}
In addition, the inequality on the left is valid for any $N\geqslant k$.
\end{lemma}
The Möbius coefficient \( \mathrm{Moeb}(\sigma) \) is determined by the cyclic structure of the symmetric group elements~\cite{Collins2006Integration}:
\begin{equation}
\mathrm{Moeb}(\sigma) = \prod_{C_i \in \sigma} \left( -1 \right)^{|C_i| - 1} c_{|C_i| - 1},\label{Moebcof}
\end{equation}
where \( C_1, C_2, \dots \) are the cycles of \( \sigma \), \( |C_i| \) is the length of the cycle, and \( c_n = \frac{(2n)!}{n!(n+1)!} \) is a Catalan number. Specifically, when \( \sigma \) is the identity permutation \( e \), the Möbius coefficient $\mathrm{Moeb}(e)=1$. Therefore, as $N$ tends to infinity, the Weingarten factor achieves its maximum value \( N^{-k} \) when \( \sigma = e \), i.e., \( \sigma \) is the identity permutation.

Additionally, since the Weingarten matrix is the inverse of the Gram matrix, we have the identity:
\begin{equation} \label{wginv}
\sum_{\delta \in S_k}{\mathrm{Wg}\left( \sigma \delta ^{-1},N \right) N^{\#(\delta \tau ^{-1})}}=\delta _{\sigma ,\tau}.
\end{equation}
\subsection{Proof Method}\label{ourProofMethod}
In this subsection, we introduce our proof strategy. We consider the output distribution with the local operator basis and the global operator basis as Eq.~(\ref{Cijl}) and Eq.~(\ref{Cijg}).
\begin{equation} \label{Cijl}
C\left( \hat{Q}_{l} \right) =\langle \psi |\hat{Q}_{}^{\left( 1 \right)}\otimes I^{\otimes (n-1)}|\psi \rangle  , 
\end{equation}
\begin{equation} \label{Cijg}
C\left( \hat{Q}_{g} \right) =\langle \psi |\hat{Q}_{}^{\left( 1 \right)}\otimes \hat{Q}_{}^{\left( 2 \right)}\otimes \cdots \otimes \hat{Q}_{}^{\left( n \right)}|\psi \rangle ,
\end{equation}
where $|\psi \rangle $ is a unitary embedding MPS, as shown Fig.~\ref{fig:MPS} with the unitaries $U$ sampled from the Haar measure. For the local output, we place the local operator at the first physical site. Due to the periodic boundary condition of the MPS, this choice does not affect our result. Here, $\hat{Q}^{(s)} = \bigotimes_{i=1}^t \sigma_i$ (with $\sigma_i \in \{X, Y, Z, I\}$) is a length-$t$ Pauli string acting on the $s$-th site of physical dimension $d=2^t$. In order to distinguish between local observables and global observables, we set $\mathrm{tr}(\hat{Q}_j^{(s)}) = 0$ to ensure that $\mathrm{tr}(\hat{Q}^{(s)})$ acts non-trivially on the corresponding site.

$C( \hat{Q}_{l})$ is represented diagrammatically in Fig.~\ref{fig:Cijl}.

\begin{figure}[h]
    \centering 
    \includegraphics[width=0.5\textwidth]{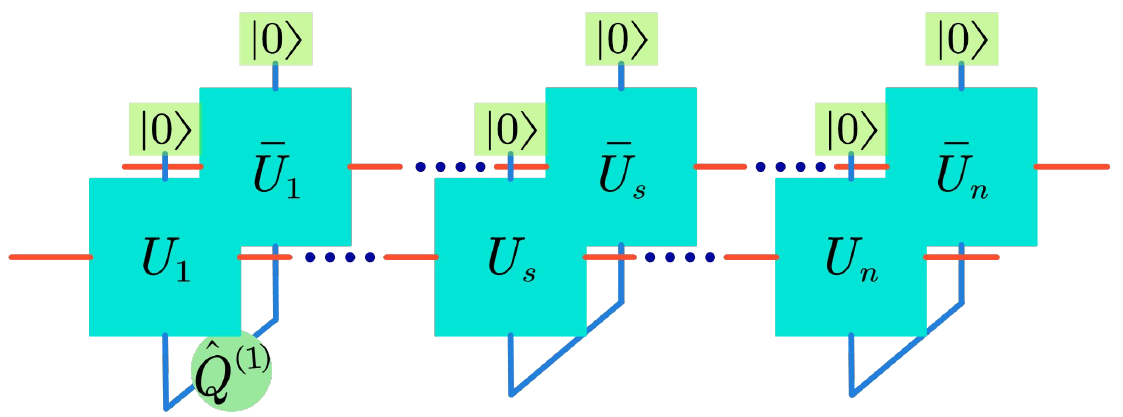}
    \caption{{Output distribution with local operator basis ($C(\hat{Q}_l)$)}}
    \label{fig:Cijl}
\end{figure} 
$C( \hat{Q}_{g})$ is represented diagrammatically in Fig.~\ref{fig:Cijg}.
\begin{figure}[h]
    \centering 
    \includegraphics[width=0.5\textwidth]{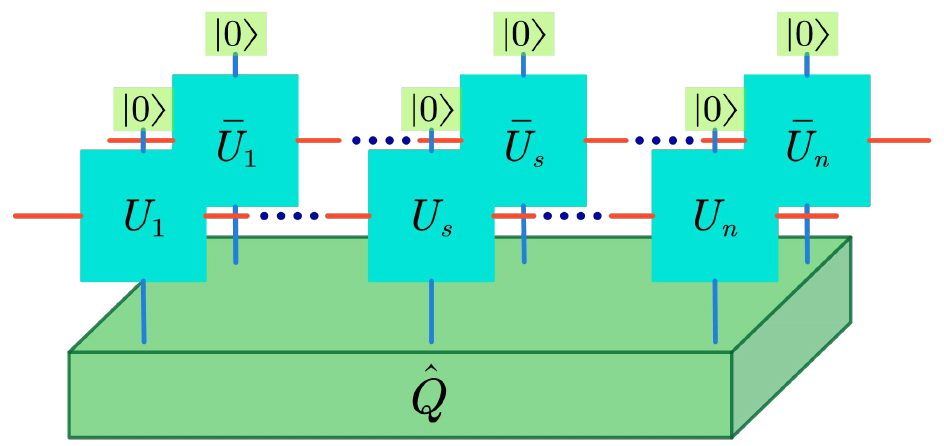}
    \caption{Output distribution with global operator basis ($C(\hat{Q}_g)$)}
    \label{fig:Cijg}
\end{figure}

First, we prove that $C( \hat{Q}_{l})$ approaches a Gaussian distribution as $d\rightarrow \infty $ and $D\rightarrow \infty $, respectively. The main method we use is to compare each order moment of $C( \hat{Q}_{l})$ with the corresponding moment of a Gaussian distribution, thereby determining the distribution of $C( \hat{Q}_{l})$. The \( k \)th-order moment of a distribution \( X \) is denoted as $\mathbb{E} (X^k)$.

The sufficiency of this method is guaranteed by our Lemma 2~\cite{kleiber2013Multivariate}.

\begin{lemma}[Carleman’s condition, Hamburger case] Let \( \gamma_k \) be the finite order moments of the distribution of a random variable \( X \) that can take values on the real line \( \mathbb{R} \). These moments uniquely determine the distribution of \( X \) if:
\begin{equation}
\sum_{k=1}^{\infty} \left( \gamma_{2k} \right)^{-1/2k} = \infty.
\end{equation}
\end{lemma}
For a standard Gaussian distribution, its $k$th-order moment is:
\begin{equation}
\mathbb{E}\bigl[ \mathcal{N}^k(0,\sigma^2) \bigr] =
\begin{cases}
0, & \text{if } k \text{ is odd}, \\[4pt]
\dfrac{k!}{(k/2)!\,2^{k/2}}\,\sigma^k, & \text{if } k \text{ is even}.
\end{cases}
\end{equation}
which can be readily verified to satisfy the condition of Lemma 2. 
 
Because the $k$th-order moment ($k=1,2,\dots,\infty$) completely describe a distribution, naturally, we can characterize the deviation of this distribution from a Gaussian distribution through the deviation of its moments from those of a Gaussian distribution. 
 
Next, we will prove that the distribution sequence with the following form constitutes a Gaussian process:
\begin{equation}
\mathscr{C} _l=\left( C\left( \hat{O}_{l1} \right) ,\cdots ,C_l\left( \hat{O}_{li} \right) ,\cdots ,C_l\left( \hat{O}_{lp} \right) \right),
\end{equation}
where
\begin{equation}
    \hat{O}_{li}=\frac{1}{c_i}\sum_{j=1}^m{{c_i}_j\hat{Q}_{lj}},
\end{equation}
is an operator spanned by the finite operator set $\{ \hat{Q}_{j} \}_{j=1}^m$. For theoretical simplicity, we usually require any two operator bases in set $\{ \hat{Q}_j \}_{j=1}^m$ to satisfy $\hat{Q}_i\hat{Q}_j+\hat{Q}_j\hat{Q}_i=2\delta _{ij}I$, i.e., the Clifford algebra. We will see in a later section that this setting has a certain necessity and not every operator will lead to a Gaussian distribution. In addition, we choose the constant $c_i$ such that $\mathrm{tr}( \hat{O}^2_i) =d^n$. In particular, for $\hat{O}_i$ is spanned by $\{ \hat{Q}_j \}_{j=1}^m$, satisfy $\hat{O}^2_i=I$.

The sequence $\mathscr{C}_l$ constitutes a Gaussian process if any finite subset is jointly Gaussian, a condition that, according to the Cramér-Wold device~\cite{Billingsley1995Probability}, is equivalent to requiring an arbitrary linear combination of these variables to follow a Gaussian distribution.
\begin{equation}
c_1C\left( \hat{O}_{l1} \right) +c_2C\left( \hat{O}_{l2} \right) +\cdots +c_pC\left( \hat{O}_{lp} \right) \sim \mathcal{N} \left( 0,\sigma ^2 \right).
\end{equation}
Due to the linearity of quantum mechanics, this is equivalent to proving:
\begin{equation}\label{CWdeviceGP}
C_l\left( c_1\hat{O}_{l1}+c_2\hat{O}_{l2}+\cdots +c_p\hat{O}_{lp} \right) \sim \mathcal{N} \left( 0,\sigma ^2 \right),
\end{equation}
which is equivalent to proving:
\begin{equation}
    C\left( \frac{1}{c}\sum_j^m{c_j\hat{Q}_{lj}} \right) \sim \mathcal{N} \left( 0,\sigma ^2 \right) ,
\end{equation}
for arbitrary coefficients $\{c_j\}_{j=1}^m$. Here, $c$ is the normalization coefficient.

Consequently, establishing that the sequence $\mathscr{C}_l$ constitutes a GP under certain parameter conditions is equivalent to proving the asymptotic Gaussianity corresponding to any operator within the span of $\{ \hat{Q}_{lj} \}_{j=1}^m$. By mapping this multivariate problem into an effective single-variable moment analysis, we can rigorously apply the method of moments in an identical fashion. Furthermore, this mapping ensures that the finite-parameter deviations of $\mathscr{C}_l$ from the GP limit are completely characterized by the $k$th-order moments of the corresponding univariate distributions.

Strictly speaking, the states $\vert\psi\rangle$ generated within the unitary embedding ensemble require normalization. However, as demonstrated in Ref.~\cite{Haferkamp2021Emergent}, the squared norm $\langle\psi\vert\psi\rangle$ is exponentially concentrated around $1$, with its variance governed by $\mathcal{O}(d^{-n})$. In the large-parameter limit considered in our framework, the fluctuations of this normalization factor vanish; moreover, within the precision of the finite corrections we are evaluating, the normalization induces only an exponentially suppressed relative error that can be safely neglected. In particular, for local observables in the large-$d$ regime, the normalization factor contributes a relative correction of at most $\mathcal{O}(d^{-2})$ (since $n\geqslant 2$), which can also be neglected at our working precision. Therefore, in the subsequent derivations, we will directly decouple the normalization factor and focus exclusively on the measurement $\langle \psi \vert\hat{O}\vert\psi \rangle$.

For the global output $C(\hat{Q}_{g})$ and the sequence $\mathscr{C} _g=\{ C( \hat{O}_{gi} ) \} _{i=1}^{p}$ (with $\hat{O}_i=\frac{1}{c_i}\sum_j^m{{c_i}_j\hat{Q}_{gj}}$ and for $\forall i,j$ $\hat{Q}_i\hat{Q}_j+\hat{Q}_j\hat{Q}_i=2\delta _{ij}I$), the proof process follows the same steps as above. However, when $d\rightarrow \infty $, the $k$th-order moment of $C(\hat{Q}_{g})$ will not converge to the $k$th-order moment of a Gaussian distribution, thus proving that for general operators, this parameter condition usually does not guarantee a Gaussian distribution. 

\section{Statement of Main Results}
In this section, we rigorously state our main results. Here, Theorem~\ref{Theorem_Local_d} and Theorem~\ref{Theorem_Local_D} correspond to Theorem 1 in the main text, while Theorem~\ref{Theorem_Global} corresponds to Theorem 2.

We first analyze the Gaussian behavior of learning models with local observables. 
We define the local operator basis as $\hat{Q}_{l} = \hat{Q}^{(1)} \otimes I \otimes \cdots \otimes I$, where $\hat{Q}^{(1)}$ is a length-$t$ traceless Pauli string ($\sigma \in \{X, Y, Z, I\}$) of physical dimension $d=2^t$. Consider a sequence of operators $\{ \hat{O}_i \}_{i=1}^{p}$ spanned by this basis $\{ \hat{Q}_{lj} \}_{j=1}^{m}$, subject to the Clifford algebra $\hat{Q}_i \hat{Q}_j + \hat{Q}_j \hat{Q}_i = 2\delta_{ij}I$. In the large-$d$ regime at fixed $D$, we establish the following theorem:
\begin{theorem} \label{Theorem_Local_d}
Define the learning output with local observable $\hat{O}_{li}$ by $C(\hat{O}_{li})$. Let the output vector $\mathscr{C} _l=\{C(\hat{O}_{li})\}_{i=1}^{p}$ be the stochastic process generated by the Haar-random MPS. Then for the physical dimension $d\rightarrow\infty$, we have
\begin{equation}
\mathscr{C} _l\sim \mathcal{G} \mathcal{P} \left( 0,\varSigma \right), \,\qquad \varSigma _{ij}=\frac{\mathrm{tr(}\hat{O}_{i}^{(1)}\hat{O}_{j}^{(1)})}{D^2d^2}.
\end{equation}
Furthermore, the deviation of the distribution from Gaussianity is completely characterized by a  finite correction term $\mathcal{O} ( k^2D^2/d ) $ depending on the moment order $k$, valid under the condition $k^2 D^2 / d \ll 1$.
\end{theorem}

Unlike the large-$d$ case, in the large-$(D,n)$ regime at fixed $d$, the operator restrictions for $C(\hat{O}_l)$ can be further relaxed as we mentioned in the main text. Consider a sequence of generic local observables $\{ \hat{O}_{li} \}_{i=1}^{p}$ taking the form $\hat{O}_{li} = \hat{O}_i^{(1)} \otimes I \otimes \cdots \otimes I$, subject to the traceless condition $\mathrm{tr}( \hat{O}_i^{(1)} ) = 0$. In the large-$(D,n)$ regime, we establish the following theorem:
\begin{theorem} \label{Theorem_Local_D}
Define the learning output with local observable $\hat{O}_{li}$ by $C(\hat{O}_{li})$. Let the output vector $\mathscr{C}_l = \{C(\hat{O}_{li})\}_{i=1}^{p}$ be the stochastic process generated by the Haar-random MPS. In the asymptotic limit $n \to \infty$ and $D \to \infty$, we have 
\begin{equation}
\mathscr{C}_l \sim \mathcal{GP}(0,\varSigma), \qquad
\varSigma_{ij}=
\begin{cases}
\dfrac{\operatorname{tr}(\hat{O}_{i}^{(1)}\hat{O}_{j}^{(1)})}{D^2d^2}, & D^2 \ll d^n, \\[2ex]
\dfrac{\operatorname{tr}(\hat{O}_{i}^{(1)}\hat{O}_{j}^{(1)})}{d^{n+1}}, & D^2 \gg d^n.
\end{cases}
\end{equation}
Furthermore, the finite-parameter deviation of the $k$th-order moment from asymptotic Gaussianity is bounded by:
\begin{itemize}
    \item[(i)] $\mathcal{O}\left(\frac{k^3}{D^2d}\right) + \mathcal{O}\left(\frac{\mathrm{poly}(k) D^2}{d^n}\right)$ for $D^2 \ll d^n$, valid under the conditions $nk^2/D \ll 1$ and $\mathrm{poly}(k)D^2/d^n \ll 1$;
    \item[(ii)] $\mathcal{O}\left(\frac{k^2 d^{n/2}}{D}\right) + \mathcal{O}\left(\frac{k^3}{d^n}\right)$ for $D^2 \gg d^n$, valid under the conditions $k^2d^{n/2}/D \ll 1$ and $k^3/d^n \ll 1$.
\end{itemize}
\end{theorem}

Notably, when the observables are restricted to the space of Clifford generators, these finite corrections reduce to:
\begin{itemize}
    \item $\mathcal{O}\left(\frac{k^2}{D^2 d}\right) + \mathcal{O}\left(\frac{\mathrm{poly}(k) D^2}{d^n}\right)$ for $D^2 \ll d^n$, valid under the conditions $nk^2/D \ll 1$ and $\mathrm{poly}(k)D^2/d^n \ll 1$;
    \item $\mathcal{O}\left(\frac{k^2 d^{n/2}}{D}\right) + \mathcal{O}\left(\frac{k^2}{d^n}\right)$ for $D^2 \gg d^n$, valid under the conditions $k^2d^{n/2}/D \ll 1$ and $k^2/d^n \ll 1$,
\end{itemize}
which are consistent with the results presented in our main text. 

Next, we extend our analysis to learning models with global observables, $\{ \hat{O}_{gi} \}_{i=1}^p$. In this case, each observable is spanned by a set of operators basis $\{ \hat{Q}_{gj} \}_{j=1}^m$ that satisfy the Clifford algebra $\hat{Q}_i \hat{Q}_j + \hat{Q}_j \hat{Q}_i = 2\delta_{ij}I$. The global basis operators $\hat{Q}_{g} = \bigotimes_{s=1}^n \hat{Q}^{(s)}$ act globally across all $n$ sites. Here, the local component at each site $s$, defined as $\hat{Q}^{(s)} = \bigotimes_{i=1}^t \sigma_i$ (with $\sigma_i \in \{X, Y, Z, I\}$), is a length-$t$ traceless Pauli string corresponding to the physical dimension $d=2^t$. We establish the following theorem:

\begin{theorem}\label{Theorem_Global} 
Define the learning output with global observable $\hat{O}_{gi}$ by $C(\hat{O}_{gi})$. Let the output vector $\mathscr{C} _g=\{C(\hat{O}_{gi}\}_{i=1}^{p}$ be the stochastic process generated by the Haar-random MPS. Then $\mathscr{C}_g$ follow the distributions below: 
\begin{itemize}
\item[(1)]In the asymptotic regime where $D/n \to \infty$ and $n \to \infty$, we have 
\begin{equation}
\mathscr{C} _g\sim \mathcal{G} \mathcal{P} \left( 0,\varSigma \right) ,\qquad \varSigma _{ij}=\frac{\mathrm{tr(}\hat{O}_{gi}\hat{O}_{gj})}{d^{2n}}
\end{equation}
Furthermore, the deviation of the distribution from Gaussianity is completely characterized by a finite correction term $\mathcal{O} (nk!m^{k/2}/D)+\mathcal{O} ( k^2/d^n) $ depending on the moment order $k$ and the size of the basis set $m$, valid under the condition $nk^2/D\ll1$ and $k^2/d^n\ll1$.

\item[(2)] In the regime where $d \to \infty$ at finite $D$, the distribution corresponding to the global operators basis $\hat{Q}_j$ is non-Gaussian, i.e., $\forall j$,  $C_g(\hat Q_j) \nsim \mathcal{N}(\mu, \sigma^2)$. For $nk^2\ll D^2$, the $k$th-order moment deviates from the exact Gaussian moment by $\mathcal{O} \Big( \frac{nk^2}{D^2} \Big)$. Specifically, the distribution exactly reduces to the $n$th-order Gaussian chaos $Z_n = \prod_{i=1}^n g_i$ with $g_i \sim \mathcal{N}(0, \sigma^2)$ when $D=1$.
\end{itemize}

\end{theorem}

\section{Exact Formulation of the $k$th-Order Moments}

In this section, we evaluate the $k$th-order moments $\mathbb{E}[ X^k ]$ by assuming that the unitaries $U$ in the MPS are drawn from the Haar measure. We define the random variable $X$ as either a local output $X_l = \langle \psi \vert{}\hat{O}^{(1)}\otimes I^{\otimes (n-1)}\vert{}\psi \rangle$ or a global output $X_g = \langle \psi \vert{}\hat{O}^{(1)}\otimes \hat{O}^{(2)}\otimes \cdots \otimes \hat{O}^{(n)}\vert{}\psi \rangle$, where $\hat{O}^{(s)}$ ($s=1,2,\dots,n$) denotes an arbitrary local traceless operator. The exact expressions derived here lay the theoretical foundation for all subsequent conclusions in our analysis. 

\subsection{ \( k \)-Order Moment of $X_l$}\label{Korder}

 We first consider $X_l=\langle \psi |\hat{O}^{\left( 1 \right)}\otimes I^{\otimes (n-1)}|\psi \rangle $. $X_l^k$ can be represented as shown in Fig.~\ref{fig:Cijk}.

\begin{figure}[h]
    \centering
    \includegraphics[width=0.5\linewidth]{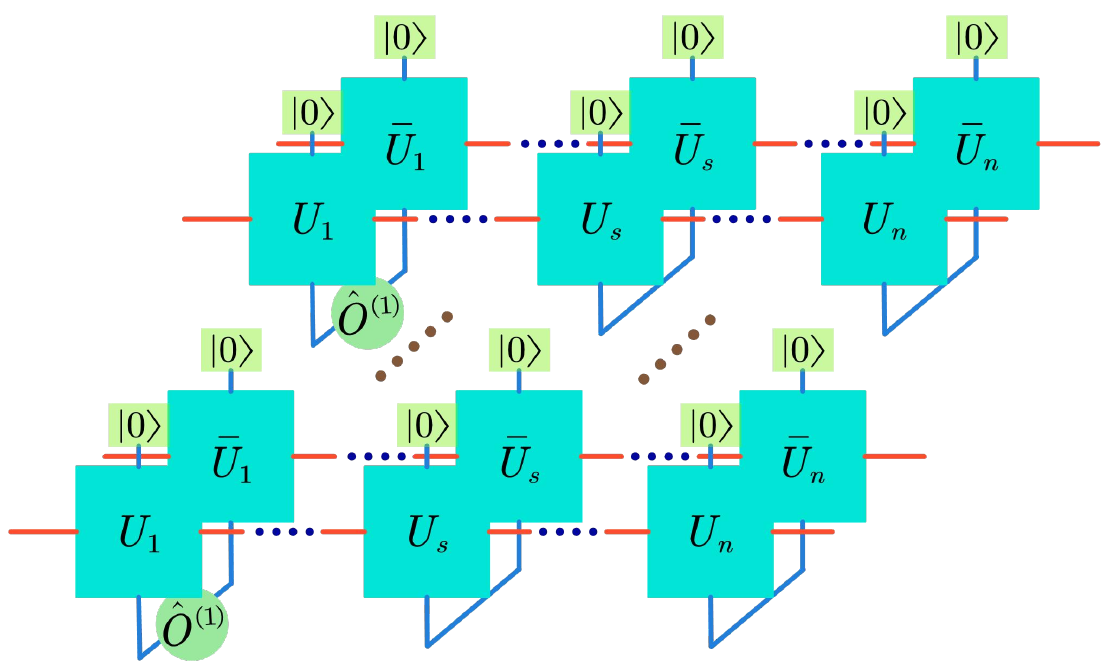}
    \caption{ Representation diagram of $X_l^k$}
    \label{fig:Cijk}
\end{figure}

We now consider each site separately. Under the condition that $U$ satisfies the Haar measure, the expected value of tensor contraction at each site is equivalent to alternative $U$ becoming  some connections (that is, the $\delta$ functions in Eq.~(\ref{Wgf})) and multiplying it by a Weingarten coefficient Wg. The resulting contractions are shown in Fig.~\ref{fig:Osite} for the observation site and in Fig.~\ref{fig:NOsite} for the non-observation sites.

\begin{figure}[htbp]
    \centering
    \includegraphics[width=0.77\linewidth]{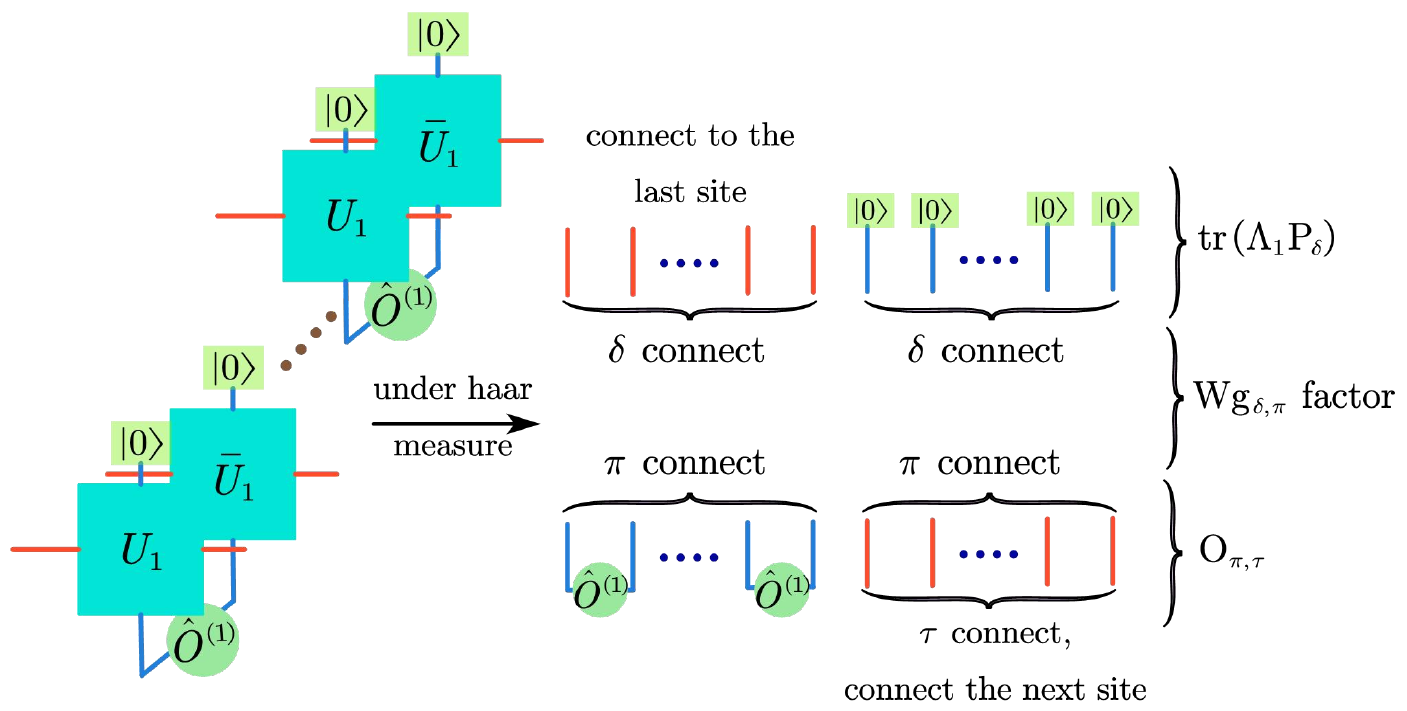}
    \caption{Representation diagram of the observation site}
    \label{fig:Osite}
\end{figure}

Each site can be represented by a matrix, and the contraction between sites can be represented by the product of these matrices. Because the MPS has a periodic boundary condition, the expectation $\mathbb{E}[ X_l^k ] $ can be expressed as the trace of the product of these matrices.

For the observation site, its corresponding matrix element is:
\begin{equation}
\mathrm{A}_{\delta ,\tau}=\sum_{\pi \in S_k}{\mathrm{tr}}\left( \Lambda _1\mathrm{P}_{\delta} \right) \mathrm{Wg}_{\delta ,\pi}\mathrm{O}_{\pi ,\tau}=\sum_{\pi \in S_k}{\mathrm{tr}}\left( \Lambda _1\mathrm{P}_{\delta} \right) \mathrm{Wg}\left( \delta \pi ^{-1},dD \right) \mathrm{O}_{\pi ,\tau},
\end{equation}
where \( \mathrm{tr} \left( \Lambda_1 \mathrm{P}_{\delta} \right) \) is the result of contraction of local state vector across \( k \) replicas, $\Lambda _1=|0\rangle \otimes |0\rangle \cdots \otimes |0\rangle $, \( \mathrm{P}_{\delta} \) is the permutation operator corresponding to \( \delta \), which is an element of the $k$th-order symmetric group. For any permutation operator \( \mathrm{P}_{\delta} \), $\mathrm{tr}\left( \Lambda _1\mathrm{P}_{\delta} \right) =1$. $\mathrm{O}_{\pi, \tau}$ is the tensor contraction corresponding to the relevant part in Fig.~\ref{fig:Osite} and it can be expressed by a matrix. Here, we use $k$th-order symmetric group elements to mark matrix subscripts. 

For a \( d \)-dimensional state vector, the permutation operator of \( \delta \in S_k \) can be expressed as:
\begin{equation}
\mathrm{P}_{\delta}^{\left( d \right)}=\sum_{i_1,\cdots ,i_k=0}^{d-1}{|}i_{\delta ^{-1}(1)},\cdots ,i_{\delta ^{-1}(k)}\rangle \langle i_1,\cdots ,i_k|,
\end{equation}  

The value of \( \mathrm{O}_{\pi, \tau} \) is determined by the ring structure formed by connecting the two ends:
\begin{align}
\mathrm{O}_{\pi ,\tau}=\mathrm{tr}\left( \mathrm{P}_{\pi}^{\left( D \right)}\left( \mathrm{P}_{\tau}^{\left( D \right)} \right) ^T \right) \mathrm{tr}\left( \mathrm{P}_{\pi}^{\left( d \right)}O \right),
\label{O}
\end{align}
where $\mathrm{tr}\left( \mathrm{P}_{\pi}^{(D)} \left( \mathrm{P}_{\tau}^{(D)} \right)^T \right) = \adjustbox{valign=c}{\includegraphics[width=0.13\linewidth]{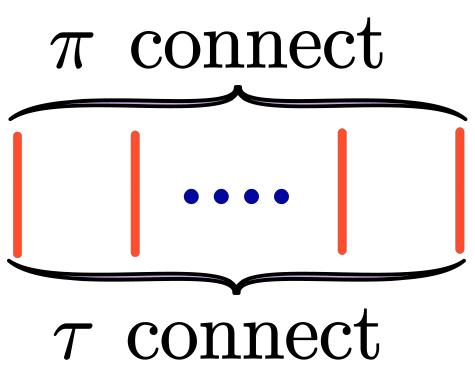}}$ and $\mathrm{tr}\left( \mathrm{P}_{\pi}^{\left( d \right)}O \right)=\adjustbox{valign=c}{\includegraphics[width=0.13\linewidth]{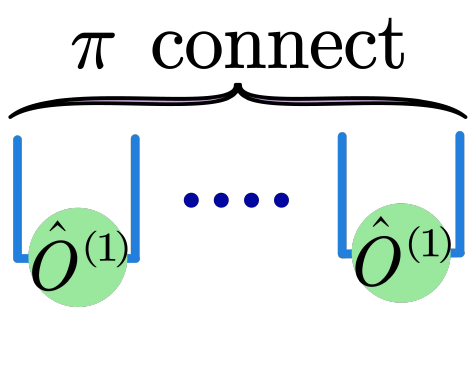}}$, ($O=\hat{O}^{(1)}\otimes \hat{O}^{(1)}\otimes \cdots \otimes \hat{O}^{(1)}$). 
$\mathrm{tr}\left( \mathrm{P}_{\pi}^{\left( D \right)}\left( \mathrm{P}_{\tau}^{\left( D \right)} \right) ^T \right) =G_{\pi ,\tau}$ is an element of the Gram matrix\cite{Collins2006Integration}, so $\mathrm{tr}\left( \mathrm{P}_{\pi}^{\left( D \right)}\left( \mathrm{P}_{\tau}^{\left( D \right)} \right) ^T \right) =D^{\#\left( \pi \tau ^{-1} \right)}$, $\#\left( \pi \tau ^{-1} \right) $ is the number of cycles in $\pi \tau ^{-1}$, which is exactly the number of ring structures formed by $\pi$ connection and $\tau$ connection.

  Take the 2-order case as an example, let $\pi =$(12) and $\tau=$(1)(2), then $\mathrm{tr}\left( \mathrm{P}_{\left( 12 \right)}^{\left( D \right)}\left( \mathrm{P}_{\left( 1 \right) \left( 2 \right)}^{\left( D \right)} \right) ^T \right)  = \adjustbox{valign=c}{\includegraphics[width=0.13\linewidth]{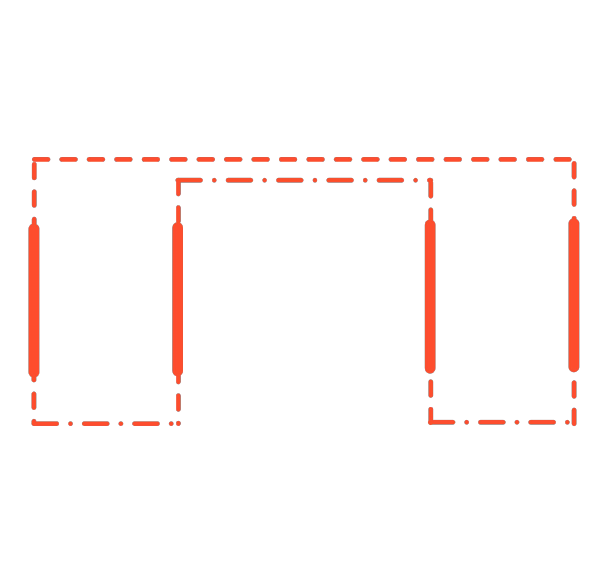}}=D^{\#\left( 12 \right)}=D^1$.  

For $\mathrm{tr}\left( \mathrm{P}_{\pi}^{\left( d \right)}O \right) $, we have $\mathrm{tr}\left( \mathrm{P}_{\pi}^{\left( d \right)}O \right) =\prod_{(t_1,t_2,\cdots ,t_m)\in \pi}{\mathrm{tr}}\left( \hat{O}^{(1)}_{t_1}\hat{O}^{(1)}_{t_2}\cdots \hat{O}^{(1)}_{t_m} \right) $, where  \( (t_1, t_2, \cdots, t_m) \) represents cycles in \( \pi \). Then,
\begin{equation}\label{trOgeneral}
\mathrm{O}_{\pi ,\tau}=D^{\#\left( \pi \tau ^{-1} \right)}\prod_{(t_1,t_2,\cdots ,t_m)\in \pi}{\mathrm{tr}}\left( \hat{O}^{(1)m} \right).
\end{equation}

 Given that the traceless Pauli-string operators $\hat{Q}^{(s)}$ satisfy $\mathrm{tr}(\hat{Q}^{(s)})=0$ and $(\hat{Q}^{(s)})^2=I$, Eq.~(\ref{trOgeneral}) can then be simplified as follows:
\begin{equation}\label{OPITAU}
\mathrm{O}_{\pi ,\tau}=d^{\#\pi}D^{\#\left( \pi \tau ^{-1} \right)},
\end{equation}
where $\pi$ is an element of the $k$th-order symmetric group with no odd cycles, we denote this set as $\mathcal{E} _k$. Specifically, when \( k \) is odd, \( \mathrm{O}_{\pi, \tau} \) is a zero matrix. 
\begin{figure}[htbp]
    \centering
    \includegraphics[width=0.77\linewidth]{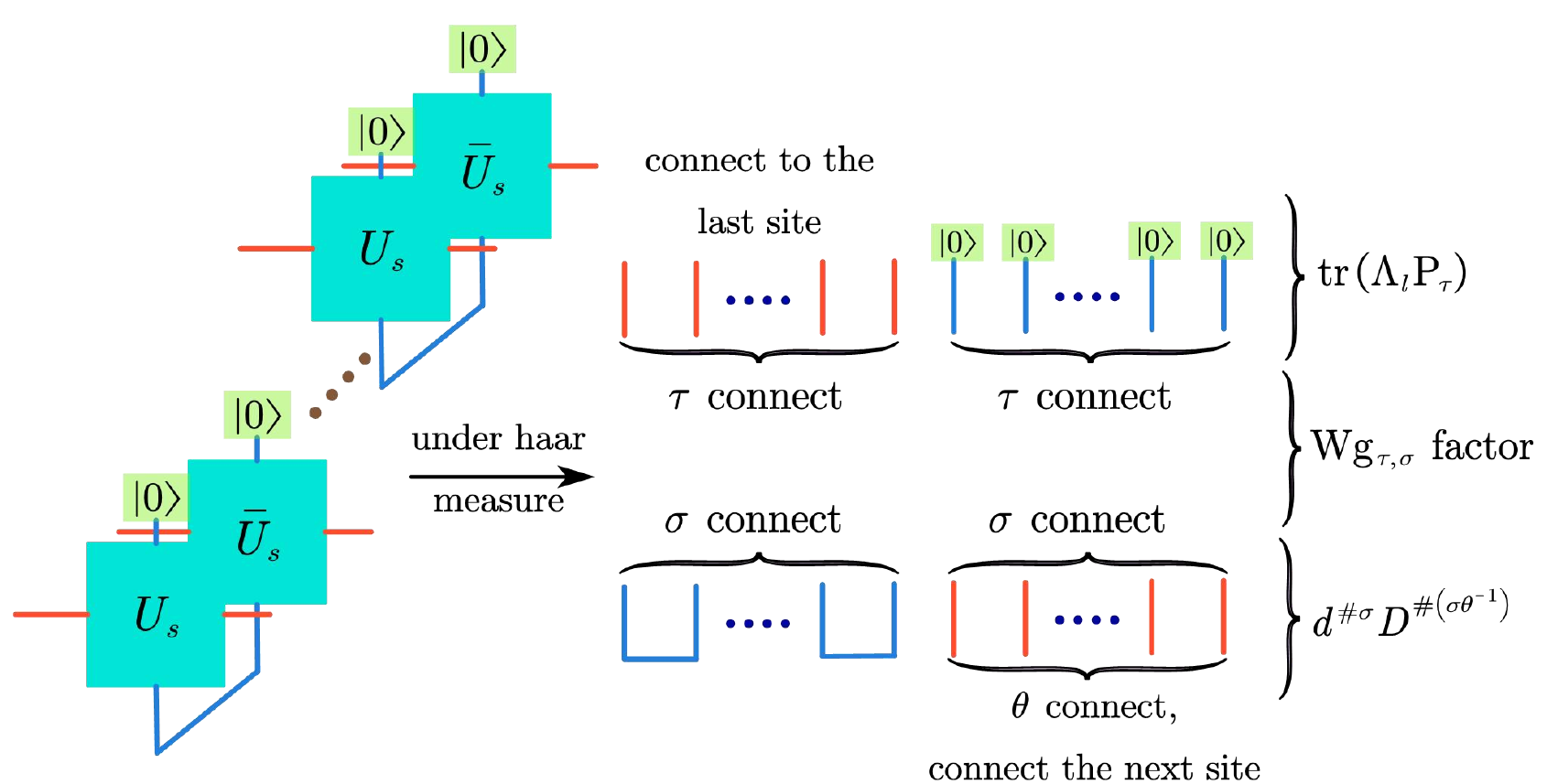}
    \caption{Representation diagram of the non-observation site}
    \label{fig:NOsite}
\end{figure}

Similar to the previous analysis, for the non-observation site, its corresponding matrix element is:
\begin{equation}\label{tematrixform}
\mathrm{Te}_{\sigma ,\theta}=\sum_{\tau \in S_k}{\mathrm{Wg}\left( \sigma \tau ^{-1},dD \right) d^{\#\tau}D^{\#\left( \tau \theta ^{-1} \right)}}.
\end{equation}
According to Eq.~(\ref{wginv}), we can set \( \theta = e \) in \( \mathrm{Te}_{\tau, \theta} \), so we can get that the first column of the \( \mathrm{Te} \) matrix, that is, the column corresponding to \( \theta = e \), is 
\begin{equation}\label{firstclomnofTe}
\mathrm{Te}_{\tau ,e}=\delta _{\tau ,e},
\end{equation}
thus the first element in the first column is 1, and its other elements are 0.

 Connecting all sites, the \( k \)th-order moment of $C(\hat{O}_l)$ is represented by the trace of the matrix product:
\begin{align}\label{Clocal}
\mathbb{E} \left( X_l ^k \right)   &=\mathrm{tr}\left( \mathrm{A}\times \mathrm{Te}^{n-1} \right) 
\\
&=\sum_{\delta ,\alpha \in S_k}{\sum_{\pi \in S_k}{\mathrm{Wg}\left( \delta \pi ^{-1},dD \right) \mathrm{tr}\left( \mathrm{P}_{\pi}^{\left( d \right)}O \right) D^{\#\left( \pi \alpha ^{-1} \right)}\left( \mathrm{Te}^{n-1} \right) _{\alpha ,\delta}}}.\label{localoperatortrace}
\end{align}
Specifically, for the local traceless Pauli string $\hat{Q}_l$ we can yield:
\begin{equation}\label{S25}
    \mathbb{E} \left( C\left( \hat{Q}_l \right) ^k \right) =\sum_{\delta ,\alpha \in S_k}{\sum_{\pi \in \mathcal{E} _k}{\mathrm{Wg}\left( \delta \pi ^{-1},dD \right) d^{\#\pi}D^{\#\left( \pi \alpha ^{-1} \right)}\left( \mathrm{Te}^{n-1} \right) _{\alpha ,\delta}}}
\end{equation}
where $k$ is even and when $k$ is odd, $\mathbb{E} ( C( \hat{Q}_l ) ^k )=0$.

Similarly, we can evaluate the mixed moments composed of outputs associated with distinct observables, $\mathbb{E}[X_{l1} X_{l2} \cdots X_{lk}]$, where the product $X_{l1} X_{l2} \cdots X_{lk}$ shares a structure analogous to that of $X_l^k$. The corresponding tensor network representation closely mirrors that in Fig.~\ref{fig:Cijk}, with the only difference being that the operators acting on distinct replica copies need not be identical. Analogously, we define the modified matrix at the observation site as
\begin{equation}
\mathrm{A}_{\delta ,\alpha}=\sum_{\pi \in S_k}{\mathrm{Wg}\left( \delta \pi ^{-1},dD \right) \mathrm{tr}\left( \mathrm{P}_{\pi}^{\left( d \right)}O_k \right) D^{\#\left( \pi \alpha ^{-1} \right)}},
\end{equation}
where $O_k = \hat{O}_1^{(1)} \otimes \hat{O}_2^{(1)} \otimes \cdots \otimes \hat{O}_k^{(1)}$, while the matrices at all other sites remain the standard Te. In particular, for $k=2$, this mixed moment reduces directly to the covariance of the distribution:
\begin{equation}\label{S30}
\mathbb{E} \left( X_{li}X_{lj} \right) =\sum_{\delta ,\alpha \in S_k}{\sum_{\pi \in S_k}{\mathrm{Wg}\left( \delta \pi ^{-1},dD \right) \mathrm{tr}\left( \mathrm{P}_{\pi}^{\left( d \right)}O_2 \right) D^{\#\left( \pi \alpha ^{-1} \right)}\left( \mathrm{Te}^{n-1} \right) _{\alpha ,\delta}}}.
\end{equation}

\subsection{ \( k \)-Order Moment of $X_g$}\label{Korderg}   
  
We next consider $X_g = \langle \psi \vert{}\hat{O}^{(1)}\otimes \hat{O}^{(2)}\otimes \cdots \otimes \hat{O}^{(n)}\vert{}\psi \rangle$. The product $X_g^k$ can be represented as shown in Fig. \ref{fig:Cgmk}. 
\begin{figure}[h]
    \centering
    \includegraphics[width=0.5\linewidth]{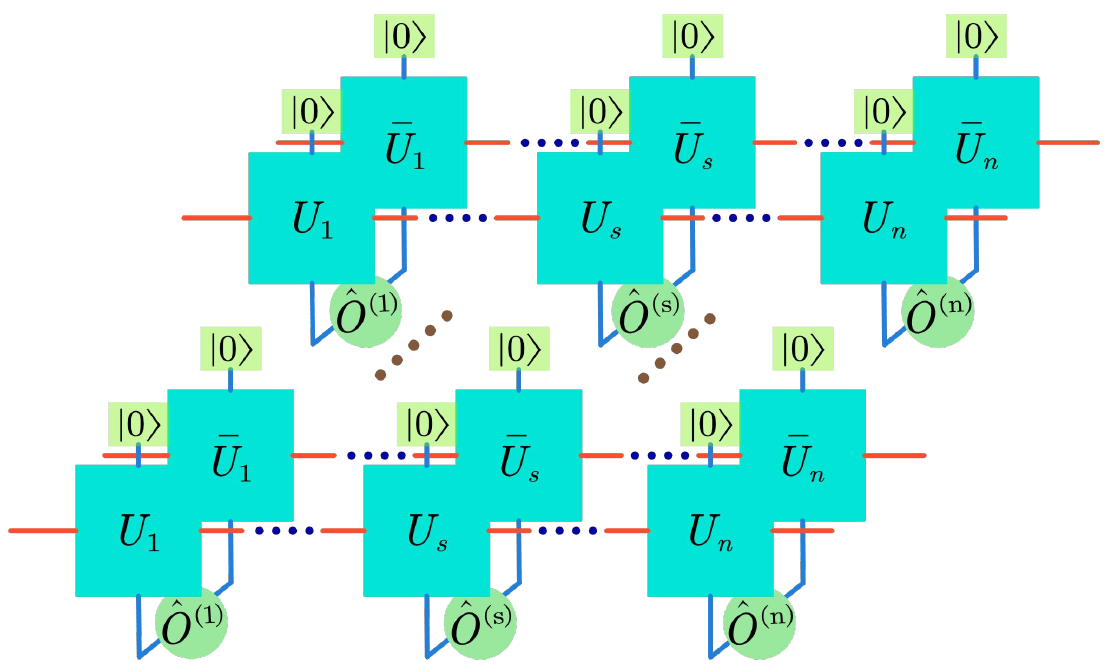}
    \caption{ Representation diagram of $X_g^k$}
    \label{fig:Cgmk}
\end{figure}
The matrix corresponding to each site is similar. A local site is shown as Fig. \ref{fig:Nositeg}. 

\begin{figure}[h]
    \centering
    \includegraphics[width=0.77\linewidth]{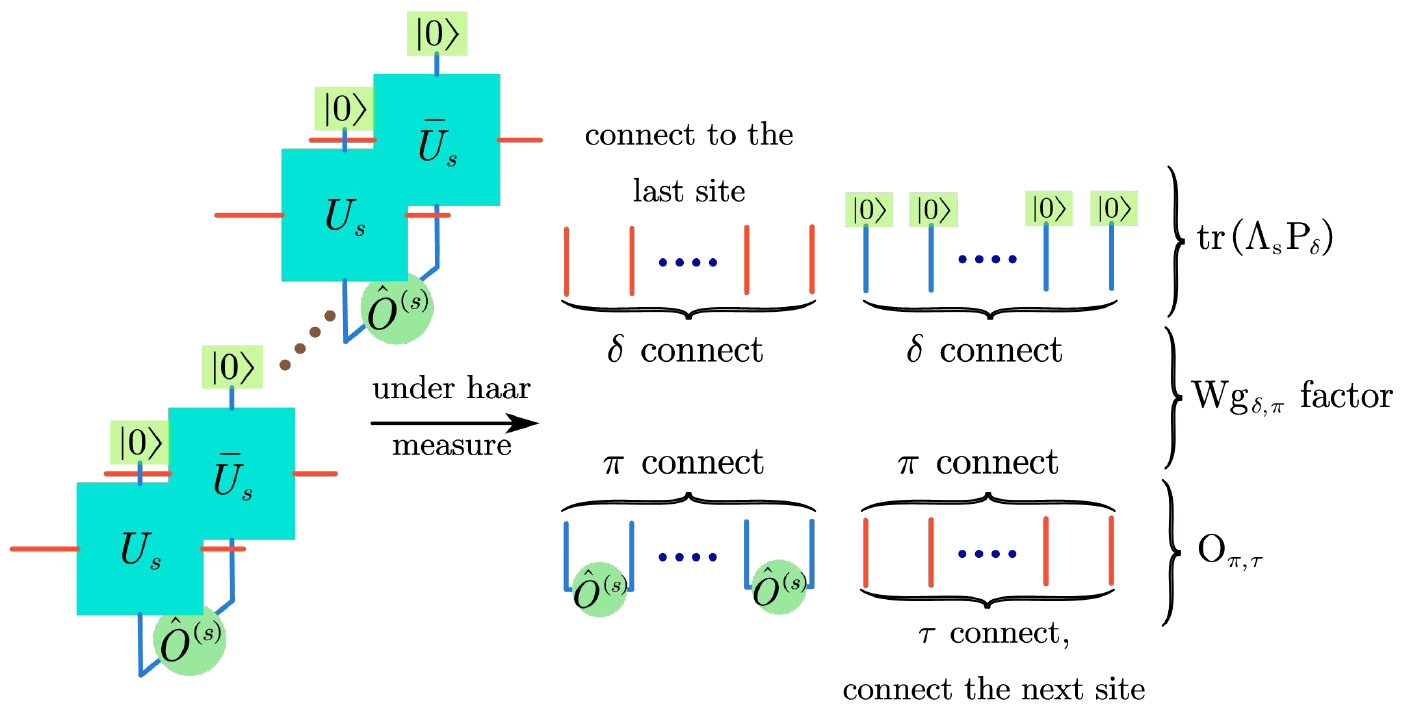}
    \caption{Representation diagram of a local observation site in $X_g^k$}
    \label{fig:Nositeg}
\end{figure}  

The matrix corresponding to this site is:
\begin{align}
\mathrm{A}_{_{\delta ,\tau}}^{\left( s \right)}&=\sum_{\pi \in S _k}{\mathrm{tr}}\left( \Lambda _s\mathrm{P}_{\delta} \right) \mathrm{Wg}\left( \delta \pi ^{-1},dD \right) \mathrm{O}_{\pi ,\tau}
\\
&=\sum_{\pi \in S _k}{\mathrm{Wg}\left( \delta \pi ^{-1},dD \right) \mathrm{tr}\left( \mathrm{P}_{\pi}O_s \right) D^{\#\left( \pi \tau ^{-1} \right)}},
\end{align}
where $O_s=\hat{O}^{\left( s \right)}\otimes \hat{O}^{\left( s \right)}\otimes \cdots \otimes \hat{O}^{\left( s \right)}$, and $\mathrm{tr}\left( \mathrm{P}_{\pi}O_s \right) =\prod_{(t_1,t_2,\cdots ,t_m)\in \pi}{\mathrm{tr}( \hat{O}^{( s ) m} )}$. For traceless Pauli-string operators $\hat{Q}^{(s)}$, it can be simplified as:
\begin{equation}\label{S31}
    \mathrm{A}_{\delta,\tau}=\sum_{\pi \in \mathcal{E} _k}{\mathrm{Wg}\left( \delta \pi ^{-1},dD \right) d^{\#\pi}D^{\#\left( \pi \tau ^{-1} \right)}},
\end{equation}
where $\pi\in \mathcal{E} _k$ and when $k$ is odd, $\mathrm{A}$ is a zero matrix. Connect all the sites, then, the \( k \)th-order moment of $X_{g}$ is:
\begin{equation}\label{Cglobal}
\mathbb{E} \left(  X_g^k \right) =\mathrm{tr}\left( \prod_s{\mathrm{A}^{\left( s \right)}} \right). 
\end{equation} 

\section{Gaussian and Non-Gaussian Regimes in Large-Parameter Limits}
  In this section, we investigate whether the local output $X_l$ and the global output $X_g$ exhibit asymptotic Gaussianity under various large-parameter limits.

\subsection{Local Output in the Large-$d$ Limit} 
In the limit $d \rightarrow \infty$, we evaluate the dominant contribution to the transfer matrix, \( \mathrm{Te} \) in Eq.~(\ref{tematrixform}). Substituting the asymptotic expansion of the Weingarten factor from Eq.~(\ref{Wgj}) into \( \mathrm{Te} \) yields:
\begin{equation}
\mathrm{Te}_{\delta ,\alpha}=\sum_{\sigma \in S_k}{\left[ \frac{\mathrm{Moeb}\left( \delta \sigma ^{-1} \right)}{\left( Dd \right) ^{2k-\#\left( \delta \sigma ^{-1} \right)}}+\mathcal{O} \left( \frac{1}{\left( Dd \right) ^{2k-\#\left( \delta \sigma ^{-1} \right) +2}} \right) \right] d^{\#\sigma}D^{\#\left( \sigma \alpha ^{-1} \right)}.}
\end{equation}
Retaining only the leading-order terms, this simplifies to:
\begin{equation} \label{Tej}
\mathrm{Te}_{\delta ,\alpha}\approx \sum_{\sigma \in S_k}{\frac{\mathrm{Moeb}\left( \delta \sigma ^{-1} \right)}{\left( Dd \right) ^{k-\#\left( \delta \sigma ^{-1} \right)}}\frac{1}{d^{k-\#\sigma}}\frac{1}{D^{k-\#\left( \sigma \alpha ^{-1} \right)}}.}
\end{equation}

Equation~(\ref{Tej}) reveals that in the large-\( d \) limit, the dominant contribution arises when \( \delta \sigma^{-1} = \sigma = e \), localizing the leading-order terms of \( \mathrm{Te} \) strictly to its first row. Furthermore, as shown in Eq.~(\ref{firstclomnofTe}), the first column of \( \mathrm{Te} \) trivially reduces to 1 for the first element and vanishes otherwise (assuming the matrix indices \( \tau \) and \( \theta \) are ordered by descending number of cycles \( \#\delta \) and \( \#\alpha \)):
\[
\left( \begin{matrix}
1 & 1/D & \cdots & 1/D^k \\
0 & 1/d & \cdots & 1/d \\
\vdots & \vdots & \ddots & \vdots \\
0 & 1/d^k & \cdots & 1/d^k \\
\end{matrix} \right)
\]
Consequently, the dominant terms of the \( \mathrm{Te}^{n-1} \) matrix are confined to the first row, yielding:
\begin{equation}
\mathrm{Te}_{e, \alpha} \approx  \frac{D^{\#\alpha}}{D^k}.
\end{equation}

Applying this asymptotic form to evaluate Eq.~(\ref{localoperatortrace}) and Eq.~(\ref{trOgeneral}), we obtain the leading-order terms of \( k \)th-order moment:
\begin{align}
\mathbb{E} \left( X_l ^k \right)  &\approx \mathrm{tr}\left( \mathrm{A}_{\delta ,e}\left( \mathrm{Te}^{n-1} \right) _{e,\alpha} \right) 
\\
&=\sum_{\delta \in S_k}{\sum_{\pi \in S_k}{\mathrm{Wg}\left( \delta \pi ^{-1},dD \right) \mathrm{tr}\left( \mathrm{P}_{\pi}O \right) D^{\#\left( \pi e^{-1} \right)}\left( D^{\#\delta}/D^k \right)}}
\\
&\approx \sum_{\delta \in S_k}{\sum_{\pi \in S_k}{\left( \mathrm{Moeb}\left( \delta \pi ^{-1} \right) \left( Dd \right) ^{-2k+\#\left( \delta \pi ^{-1} \right)} \right) \prod_{(t_1,t_2,\cdots ,t_m)\in \pi}{\mathrm{tr}\left( \hat{O}^{(1)m} \right) D^{\#\pi}\left( D^{\#\delta}/D^k \right)}}}.
\end{align}
When $\delta \pi^{-1} = e$, i.e. $\delta=\pi$, the above formula yields the leading order terms, 
\begin{align}
   \mathbb{E} \left( X_{l}^{k} \right) &\approx \sum_{\pi \in S_k}{\left( \mathrm{Moeb}\left( e \right) \left( Dd \right) ^{-2k+k} \right) \prod_{(t_1,t_2,\cdots ,t_m)\in \pi}{\mathrm{tr}\left( \hat{O}^{(1)m} \right) D^{\#\pi}\left( D^{\#\pi}/D^k \right)}}
\\
&=\sum_{\pi \in S_k}{\frac{D^{2\#\pi}}{D^{2k}d^k}\prod_{(t_1,t_2,\cdots ,t_m)\in \pi}{\mathrm{tr}\left( \hat{O}^{(1)m} \right)}}\label{S36}
\end{align}

Given that $\mathrm{tr}(\hat{O}^{(1)})=0$, the first moment vanishes, yielding $\mathbb{E}(X_l)=0$. Then we need to consider how these moments to asymptotically match those of a centered Gaussian distribution, i.e. $(k-1)!!\sigma ^k$.

The structure of Eq.~(\ref{S36}) closely parallels the $k$th-order moments of the Haar-random output $C_\mathrm{Haar} = \langle 0|U^{\dagger}\hat{O}U|0\rangle$. To calculate this, we use the $k$-fold Haar-twirling formula~\cite{Collins2006Integration}. In large physical dimension, we yield
\begin{align}
    \mathbb{E}_{U(d)} [C_{\mathrm{Haar}}^{k}]&=\sum_{\sigma ,\pi \in S_k}{\mathrm{Wg}}(\sigma \pi ^{-1},d)\mathrm{tr}\left( |0\rangle \langle 0|^{\otimes k}\mathrm{P}_{\sigma} \right) \mathrm{tr}\left( \hat{O}^{\otimes k}\mathrm{P}_{\pi} \right) 
\\
&=\left[ \sum_{\sigma \in S_k}{\mathrm{Wg}}(\sigma ,d) \right] \sum_{\pi \in S_k}{\mathrm{tr}\left( \hat{O}^{\otimes k}\mathrm{P}_{\pi} \right)}
\\
&\approx \frac{1}{d^k}\sum_{\pi \in S_k}{\mathrm{tr}\left( \hat{O}^{\otimes k}\mathrm{P}_{\pi} \right)}
\\
&
=\frac{1}{d^k}\sum_{\pi \in S_k}{\prod_{(t_1,t_2,\dots ,t_m)\in \pi}{\mathrm{tr}\left( \hat{O}^m \right)}},\label{S40}
\end{align}
where in the second line, for $\forall \sigma \in S_k$, $\mathrm{tr}\left( |0\rangle \langle 0|^{\otimes k}\mathrm{P}_{\sigma} \right) $=1 and $\forall \pi \in S_k$, $\sum_{\sigma \in S_k}{\mathrm{Wg}}(\sigma \pi ^{-1},d)=\sum_{\sigma \in S_k}{\mathrm{Wg}}(\sigma ,d)$.

For $C_\mathrm{Haar}$ to asymptotically approach a Gaussian distribution, a sufficient and nearly necessary condition is that $\hat{O}$ possesses a flat operator spectrum~\cite{Meckes2020Random,Anschuetz2025Unified}:
\begin{equation}\label{flatspectrualcondition}
    \frac{\|\hat{O}\|_{\infty}}
    {\sqrt{\mathrm{tr}(\hat{O}^{2})}}
    \rightarrow 0,
\end{equation}
where $\|\hat{O}\|_{\infty}$ denotes the operator norm, i.e.,
$\|\hat{O}\|_{\infty}=\max_i|\lambda_i(\hat{O})|$ for Hermitian $\hat{O}$.

Indeed, for a Hermitian operator $\hat{O}$,
\begin{equation}
    \frac{|{\rm tr}(\hat{O}^{m})|}
    {({\rm tr}(\hat{O}^{2}))^{m/2}}
    \leq
    \left(
    \frac{\|\hat{O}\|_{\infty}}
    {\sqrt{{\rm tr}(\hat{O}^{2})}}
    \right)^{m-2},
    \qquad m\geq 2 ,
\end{equation}
which shows that flat-spectral condition guarantees the term in Eq.~(\ref{S40}) corresponding to the disjoint transpositions (e.g., of the form $(12)(34)(56)$), dominate as the leading-order contributions, while all other terms are strictly suppressed. Hereafter, we denote the set of all such disjoint transpositions in $S_k$ as $\mathcal{P}_k$. In the $k$th-order symmetric group $S_k$, the combinatorial multiplicity of such disjoint transpositions is exactly $(k-1)!!$. Under flat-spectral condition, we yield:
\begin{align}
    \mathbb{E} _{U\left( d \right)}\left( C_\mathrm{Haar}^k \right) &\approx \frac{1}{d^k}\sum_{\pi \in S_k}{\prod_{(t_1,t_2,\dots ,t_m)\in \pi}{\mathrm{tr}\left( \hat{O}^m \right)}}
\\
&=\frac{k!}{\left( k/2 \right) !2^{k/2}}\frac{1}{d^k}\mathrm{tr}( \hat{O}^2 ) ^{k/2}+\sum_{\pi \in S_k\backslash\mathcal{P} _k}{\prod_{(t_1,t_2,\cdots ,t_m)\in \pi}{\mathrm{tr}( \hat{O}^m )}}
\\
&\approx \frac{k!}{\left( k/2 \right) !2^{k/2}}\frac{1}{d^k}\mathrm{tr}( \hat{O}^2 ) ^{k/2},
\end{align}
which matches the $k$th-order moment of a centered Gaussian distribution.

Similarly, for MPS local outputs, not all operators $\hat{O}^{(1)}$ yield asymptotic Gaussianity in the large-$d$ limit. As indicated by Eq.~(\ref{S36}), at finite $D$, this emergence also requires a flat-spectrum condition. The most natural operator fulfilling this requirement is our operator basis $\hat{Q}_l$, i.e. traceless Pauli string. Consequently, evaluating the $k$th-order moments $\mathbb{E} ( C( \hat{Q}_l ) ^k )$ via Eq.~(\ref{OPITAU}) yields the leading order contribution:
\begin{align}
    \mathbb{E} \left( C\left( \hat{Q}_l \right) ^k \right) &\approx \sum_{\pi \in \mathcal{E} _k}{\frac{D^{2\#\pi}}{D^{2k}d^k}d^{\#\pi}}
\\
&=\frac{k!}{\left( k/2 \right) !2^{k/2}}\frac{1}{D^kd^{k/2}}\left[ 1+\mathcal{O} \left( \frac{1}{D^2d} \right) \right] ,
\end{align}
which exactly matches the $k$th-order moment of a Gaussian distribution with variance \( \sigma^2 = 1/(dD^2) \). In subsequent sections, we will explicitly calculate the ratio of the remainder term to the leading order term and demonstrate that this  finite corrections decay polynomially with \( d \), extending the validity of this Gaussian approximation to finite but sufficiently large \( d \).

Furthermore, extending our focus to the distribution sequence $\{ C( \hat{O}_{li} ) \}_{i=1}^p$, where each observable is spanned by the set $\{ \hat{Q}_{lj} \}_{j=1}^m$. As established in Sec.~\ref{ourProofMethod}, proving that the sequence $\{ C( \hat{O}_{li} ) \}_{i=1}^p$ constitutes a joint Gaussian distribution reduces to demonstrating $C( \frac{1}{c}\sum_{j=1}^m{c_j\hat{Q}_{lj}} ) \sim \mathcal{N} ( 0,\sigma ^2 )$. Owing to the Clifford algebra, we immediately have:
\begin{align}
    (\frac{1}{c}\sum_{j=1}^m{c_j}\hat{Q}_{lj})^2&=(\frac{1}{c}\sum_{j=1}^m{c_j}\hat{Q}_{lj}^{\left( 1 \right)})^2\otimes I\otimes \cdots \otimes I
\\
&=I\otimes I\otimes \cdots \otimes I.
\end{align}
Then, following an identical calculation as that for the $k$th-order moments of $(\frac{1}{c}\sum_{j=1}^m{c_j}\hat{Q}_{lj})^2$, we obtain:
\begin{equation}
    \mathbb{E} \left[ C\left( \frac{1}{c}\sum_{j=1}^m{c_j}\hat{Q}_{lj} \right) ^k \right] \approx \frac{k!}{\left( k/2 \right) !2^{k/2}}\frac{1}{D^kd^{k/2}}\left[ 1+\mathcal{O} \left( \frac{1}{D^2d} \right) \right] .
\end{equation}

Analogous to the result in Eq.(\ref{S36}), taking the large-$d$ limit in Eq.(\ref{S30}), we obtain the covariance between the outputs as
\begin{align}
\mathbb{E} \left( C\left( \hat{O}_{li} \right) C\left( \hat{O}_{lj} \right) \right) &\approx \sum_{\pi \in \left\{ \left( 12 \right) ,e \right\}}{\frac{D^{2\#\pi}}{D^4d^2}\mathrm{tr}\left( \mathrm{P}_{\pi}O_2 \right)}
\\
&=\frac{\mathrm{tr}\left( \hat{O}_{i}^{\left( 1 \right)}\hat{O}_{j}^{\left( 1 \right)} \right)}{D^2d^2}+0
\end{align}
where the first and second terms correspond to the permutations $\pi = (12)$ and $\pi = e$, respectively, with the identity permutation contribution vanishing identically due to the tracelessness of the observables, $\operatorname{tr}(O_2) = \operatorname{tr}(\hat{O}_i^{(1)}) \operatorname{tr}(\hat{O}_j^{(1)}) = 0$.

Thus, we obtain the result stated in Theorem~\ref{Theorem_Local_d} in the asymptotic limit.
   
\subsection{Local Output in the Large-$(D,n)$ Limit with $D^2 \gg d^n$} 

Unlike the \( d \rightarrow \infty \) limit, the asymptotic Gaussian distribution emerging in the \( D \rightarrow \infty \) regime depends explicitly on the number of physical sites \( n \). We can clearly illustrate this by evaluating the second moment ($k=2$) via Eq.~(\ref{S25}):
\begin{equation}\label{localoperatork2}
    \mathbb{E} \left( C\left( \hat{Q}_l \right) ^2 \right) =\mathrm{tr}\left( \mathrm{A}\times \mathrm{Te}^{n-1} \right) ,
\end{equation}
where the matrix $\mathrm{A}$ for $k=2$ is given by:
\begin{equation}
    \left( \begin{matrix}
	\frac{1}{D^2d^2-1}&		\frac{-1}{Dd\left( D^2d^2-1 \right)}\\
	\frac{-1}{Dd\left( D^2d^2-1 \right)}&		\frac{1}{D^2d^2-1}\\
\end{matrix} \right) \left( \begin{matrix}
	0&		0\\
	0&		d\\
\end{matrix} \right) \left( \begin{matrix}
	D^2&		D\\
	D&		D^2\\
\end{matrix} \right) ,
\end{equation}
and the transfer matrix $\mathrm{Te}$ for $k=2$:
\begin{equation}
    \left( \begin{matrix}
	\frac{1}{D^2d^2-1}&		\frac{-1}{Dd\left( D^2d^2-1 \right)}\\
	\frac{-1}{Dd\left( D^2d^2-1 \right)}&		\frac{1}{D^2d^2-1}\\
\end{matrix} \right) \left( \begin{matrix}
	d^2&		0\\
	0&		d\\
\end{matrix} \right) \left( \begin{matrix}
	D^2&		D\\
	D&		D^2\\
\end{matrix} \right). 
\end{equation}
Directly evaluating Eq.~(\ref{localoperatork2}) yields:
\begin{equation}
    \mathbb{E} \left( C\left( \hat{Q}_l \right) ^2 \right) =\frac{-1}{D^2d^2-1}+\frac{D^2de^{n-1}}{D^2d^2-1}+\frac{Dd}{D^2d^2-1}\left( \frac{Dd+D}{D^2d+1}-\frac{De^{n-1}\left( d+1 \right)}{D^2d+1} \right) , 
\end{equation}
where $e=\frac{D^2d-d}{D^2d^2-1}$ denotes the second-largest eigenvalue of the $\mathrm{Te}$ matrix (the maximum eigenvalue being 1). In the large-$D$ regime, this expression simplifies to:
\begin{equation}\label{2momentresulatsoflarge_D}
    \mathbb{E} \left( C\left( \hat{Q}_l \right) ^2 \right) \approx \frac{1}{D^2d}+\frac{1}{d^n}.
\end{equation}
Equation (\ref{2momentresulatsoflarge_D}) indicates that the leading term of the variance behaves differently depending on whether $D^2\gg d^n$ or $D^2\ll d^n$, leading to distinct scaling behaviors as $n$ increases. 

Here, we focus on the regime $D \rightarrow \infty$ with $D^2 \gg d^n$, deferring the discussion of the $D^2 \ll d^n$ limit to the subsequent subsection.   

Parallel to the large-$d$ analysis, we first extract the leading-order contribution of $\mathrm{Te}$ for large $D$. From Eq.~(\ref{Tej}), the dominant term emerges when $\delta \sigma ^{-1}=\sigma \alpha ^{-1}=e$, which is on the diagonal of Te and $\mathrm{Te}_{\delta ,\delta}\approx d^{\#\delta}/d^k$.

In the large-$D$ limit, the asymptotic scaling of the summation in Eq.~(\ref{Tej}) with respect to $D$ and $d$ can be bounded as follows: 
\begin{align}
\sum_{\sigma \in S_k}{\frac{1}{\left( Dd \right) ^{k-\#\left( \delta \sigma ^{-1} \right)}}\frac{1}{d^{k-\#\sigma}}\frac{1}{D^{k-\#\left( \sigma \alpha ^{-1} \right)}}}&\leqslant \sum_{\sigma \in S_k}{\frac{d^{k-\#\sigma +\#\delta}}{d^k}\frac{d^{\#\sigma}}{d^k}\frac{1}{D^{k-\#\left( \delta \sigma ^{-1} \right)}}\frac{1}{D^{k-\#\left( \sigma \alpha ^{-1} \right)}}}
\\
&=\frac{d^{\#\delta}}{d^k}\sum_{\sigma \delta \in S_k}{\frac{1}{D^{k-\#\sigma}}\frac{1}{D^{k-\#\left( \sigma \delta \alpha ^{-1} \right)}}}
\\
&=\frac{d^{\#\delta}}{d^k}\sum_{\sigma \in S_k}{\frac{1}{D^{k-\#\sigma}}\frac{D^{\#\left( \sigma \delta \alpha ^{-1} \right)}}{D^k}}
\\
&\leqslant \frac{d^{\#\delta}}{d^k}\sum_{\sigma \in S_k}{\frac{1}{D^{k-\#\sigma}}\frac{D^{k-\#\sigma +\#\delta \alpha ^{-1}}}{D^k}}
\\
&=\mathcal{O} \left( \frac{d^{\#\delta}}{d^k}\frac{D^{\#\delta \alpha ^{-1}}}{D^k} \right)  . \label{orderofTematrix}
\end{align}
This derivation utilizes the translation invariance of the summation over the symmetric group $\sum_{\sigma}$ and the triangle inequality for the Cayley distance in the symmetric group~\cite{Dummit2004Abstract}:
\begin{equation}\label{triangleineq}  
\#\delta +\#\sigma -k\leqslant \#\left( \delta \sigma ^{-1} \right) \leqslant k-\left| \#\delta -\#\sigma \right|,   
\end{equation}
where $\delta ,\sigma \in S_k$. These two group theoretic properties will be employed frequently in subsequent proofs.

Although the upper bound in Eq.~(\ref{orderofTematrix}) is highly rough, it effectively illustrates the asymptotic magnitude of the elements in $\mathrm{Te}$ as $D \rightarrow \infty$:
\begin{equation}
\left( \begin{matrix}{}
	1&		\frac{1}{D}&		\cdots&		\frac{D^{\#\alpha}}{D^k}&		\cdots\\
	0&		\cdots&		\cdots&		\cdots&		\cdots\\
	\vdots&		\cdots&		\cdots&		\frac{d^{\#\delta}D^{\#\left( \delta \alpha ^{-1} \right)}}{d^kD^k}&		\cdots\\
	0&		\cdots&		\cdots&		\cdots&		\cdots\\
\end{matrix} \right). 
\end{equation}
The first element is exactly 1, with the remainder of the first column being identically 0. Elements in the first row scale approximately as $\frac{D^{\#\alpha}}{D^k}$, while elements in the remaining rows scale as $\frac{d^{\#\delta}D^{\#\delta \alpha ^{-1} }}{d^kD^k}$. 
Neglecting higher-order terms about $1/D$, $\mathrm{Te}$ effectively becomes diagonal with elements \( d^{\#\delta}/d^k \). Consequently, the dominant contribution to the $\mathrm{Te}^n$ matrix takes the form:
\begin{equation}
    \left( \begin{array}{llll}
	1&		\cdots&		\cdots&		\cdots\\
	0&		\frac{1}{d^n}&		\cdots&		\cdots\\
	\cdots&		\cdots&		\frac{d^{n\#\delta}}{d^{nk}}&		\cdots\\
	0&		\cdots&		\cdots&		\ddots\\
\end{array} \right) .\label{Te_Norder}
\end{equation}
Substituting this form into Eq.~(\ref{Clocal}) gives:
\begin{align}
\mathbb{E} \left( X_l ^k \right)  &=\mathrm{tr}\left( \mathrm{A}\times\mathrm{Te}^{n-1} \right) 
\\
&\approx \sum_{\delta =\alpha \in S_k}{\sum_{\pi \in S_k}{\mathrm{Wg}\left( \delta \pi ^{-1},Dd \right) \prod_{(t_1,t_2,\cdots ,t_m)\in \pi}{\mathrm{tr}\left( \hat{O}^{(1)m} \right)}D^{\#\left( \pi \alpha ^{-1} \right)}\left( \frac{d^{\#\alpha}}{d^k} \right) ^{n-1}}}.
\end{align}
Incorporating the asymptotic Weingarten factor from Eq.~(\ref{Wgj}) yields:
\begin{align}
\mathbb{E} \left( X_l ^k \right)  &\approx \sum_{\delta \in S_k}{\sum_{\pi \in S _k}{\left( \mathrm{Moeb}\left( \delta \pi ^{-1} \right) \left( Dd \right) ^{-2k+\#\left( \delta \pi ^{-1} \right)} \right) \prod_{(t_1,t_2,\cdots ,t_m)\in \pi}{\mathrm{tr}\left( \hat{O}^{(1)m} \right)}D^{\#\left( \pi \delta ^{-1} \right)}\left( \frac{d^{\#\delta}}{d^k} \right) ^{n-1}}} \label{Dinfinity}
\end{align}
The leading-order contribution in Eq.~(\ref{Dinfinity}) is localized at \( \pi \delta^{-1} = e \), reducing the expression to:
\begin{align}
\mathbb{E} \left( X_l ^k \right) &\approx \sum_{\pi \in S_k}{\mathrm{Moeb}\left( e \right) \left( Dd \right) ^{-2k+\#e}\prod_{(t_1,t_2,\cdots ,t_m)\in \pi}{\mathrm{tr}\left( \hat{O}^{(1)m} \right) D^{\#e}\left( \frac{d^{\#\pi}}{d^k} \right) ^{n-1}}}
\\
&=\sum_{\pi \in S_k}{\left( \frac{d^{\#\pi}}{d^k} \right) ^{n-1}\frac{1}{d^k}\prod_{(t_1,t_2,\cdots ,t_m)\in \pi}{\mathrm{tr}\left( \hat{O}^{(1)m} \right)}}
\\
&=\frac{1}{d^{nk}}\sum_{\pi \in S_k}{\prod_{(t_1,t_2,\dots ,t_m)\in \pi}{\mathrm{tr}\left[ \left( \hat{O}^{\left( 1 \right)}\otimes I\otimes \cdots \otimes I \right) ^m \right]}},
\label{EKD}
\end{align}

The result above is consistent with the structure of Eq.~(\ref{S40}). It is worth noting that, for such a local operator $\hat{O}=\hat{O}^{(1)}\otimes I^{\otimes(n-1)}$ with an arbitrary operator $\hat{O}^{(1)}$, the flat-spectrum condition in Eq.~(\ref{flatspectrualcondition}) is readily satisfied. The spectrum of $\hat{O}$ consists of the eigenvalues of $\hat{O}^{(1)}$, each repeated
$d^{n-1}$ times. Therefore,
\begin{equation}
    {\rm tr}(\hat{O}^{2})
    =
    d^{n-1}{\rm tr}(\hat{O}^{(1)2}),
    \qquad
    \|\hat{O}\|_{\infty}
    =
    \|\hat{O}^{(1)}\|_{\infty},
\end{equation}
and hence
\begin{equation}
    \frac{\|\hat{O}\|_{\infty}}
    {\sqrt{{\rm tr}(\hat{O}^{2})}}
    =
    \frac{\|\hat{O}^{(1)}\|_{\infty}}
    {d^{(n-1)/2}\sqrt{{\rm tr}(\hat{O}^{(1)2})}} .
\end{equation}

Thus, for sufficiently large system size $n$, any local operator
$\hat{O}^{(1)}$ satisfies the flat-spectrum condition. Incorporating the condition $\mathrm{tr}(\hat{O}^{(1)}) = 0$, simplifying Eq.~(\ref{EKD}) yields:
\begin{equation}\label{S67}
    \frac{1}{d^{nk}}\sum_{\pi \in S_k}{\prod_{(t_1,t_2,\dots ,t_m)\in \pi}{\mathrm{tr}\left[ \left( \hat{O}^{\left( 1 \right)}\otimes I\otimes \cdots \otimes I \right) ^m \right]}}=\frac{k!}{\left( k/2 \right) !2^{k/2}}\frac{1}{d^{nk/2}}\left[ 1+\mathcal{O} \left( \frac{1}{d^n} \right) \right]  .
\end{equation}

The leading term in Eq.~(\ref{S67}) corresponding to $\pi \in \mathcal{P}_k$ exactly matches the \( k \)-th moment of a Gaussian distribution with variance \( \sigma^2 = 1/d^n \). In the large-$n$ limit, the sub-leading corrections vanish relative to the dominant term. In subsequent section, we rigorously demonstrate that the ratio of these residual terms to the leading-order contribution in $\mathbb{E} ( C( \hat{O}_l ) ^k ) $ approaches zero for $\left( D,n \right) \rightarrow \infty $, and we also derive explicit corrections for finite $(D,n)$.

Since our analysis applies to arbitrary traceless operators, any linear combination of them naturally preserves the established derivation, thereby extending our results to a full GP. Furthermore, evaluating Eq.~(\ref{S30}) under the asymptotic limits considered in this section readily yields the asymptotic covariance between distinct outputs: $\mathbb{E} ( C( \hat{O}_{li} ) C( \hat{O}_{lj} ) ) \approx \mathrm{tr}( \hat{O}_{i}^{( 1)}\hat{O}_{j}^{( 1 )} ) /d^{n+1}$.

\subsection{Local Output in the Large-$(D,n)$ Limit with $D^2 \ll d^n$}\label{spectraldecomposition}

From Eq.~(\ref{Te_Norder}), we can see that as $n\rightarrow \infty $ (for fixed $D$), the dominant part of the $\mathrm{Te}^n$ matrix shifts to the first row. According to the properties of the Te matrix, we can analyse the dominant term of $\mathbb{E} ( X_l ^k ) $ from the perspective of the spectral decomposition of the Te matrix.

What we are concerned with is:
\begin{equation}
    \mathbb{E} \left( X_l^k \right) =\mathrm{tr}\left( \mathrm{A}\times \mathrm{Te}^{n-1} \right) .  
\end{equation}
Since the Weingarten matrix is the inverse of the Gram matrix (i.e. $\mathrm{Wg}_{d}^{-1}=G_{d}$ with element $d^{\#\delta \sigma ^{-1}}$), we can rewrite these matrix with Gram matrix as follows:
\begin{equation}
    \mathrm{Te}=G_{dD}^{-1}XG_D.
    \label{eq:Te_WXG}
\end{equation}
and 
\begin{equation}
 \mathrm{A}=G_{dD}^{-1}X_OG_D
   \label{eq:BE_def_compact}
\end{equation}
where $X=d^{\#\sigma}$ and $X_O=\mathrm{tr}\left( \mathrm{P}_{\pi}O \right) $ are diagonal.
In Ref.~\cite{Daniel2023Typical}, it was shown that, Te is diagonalizable and after a suitable ordering of the permutation indices according to their fixed-point structure, Te is block upper triangular. More precisely, if
\begin{equation}
    \mathrm{Fix}(\sigma)=\{i:\sigma(i)=i\},
\end{equation}
then the matrix elements of Te vanish unless the fixed points of the column index are also fixed points of the row index:
\begin{equation}
    \mathrm{Te}_{\tau\sigma}=0
    \quad\text{unless}\quad
    \mathrm{Fix}(\sigma)\subseteq \mathrm{Fix}(\tau).
\label{eq:fixed_point_triangular}
\end{equation}
Let $N_f(\sigma)$ represent the fixed point number of symmetric group elements $\sigma$. Then when the subscript $\sigma$ of Te matrix are arranged from large to small according to $N_f(\sigma)$, Te matrix can be written as an upper triangle in the form of:
\begin{equation}
    \left( \begin{matrix}
	1&		U_{k,k-2}&		U_{k,k-3}&		\cdots&		U_{k,0}\\
	0&		B_{k-2}&		U_{k-2,k-3}&		\cdots&		U_{k-2,0}\\
	0&		0&		B_{k-3}&		\cdots&		U_{k-3,0}\\
	\vdots&		\vdots&		\vdots&		\ddots&		\vdots\\
	0&		0&		0&		\cdots&		B_0\\
\end{matrix} \right) ,\label{Tematrixform}
\end{equation}
where $B_i$, $i=0,1,\cdots ,k-2$ represents a diagonal block matrix corresponding to $N_f(\sigma)=i$.
Thus the spectrum of Te can be obtained from its diagonal blocks. Among them, the eigenvalue 1 is the largest eigenvalue of the Te matrix, and it is non-degenerate~\cite{Daniel2023Typical}. Its right eigenvector is therefore the identity permutation vector:
\begin{equation}
    \mathrm{Te}|e\rangle=|e\rangle.
\end{equation}

In addition, the eigenvalues of the transfer matrix $\mathcal{T}_e$ are non-negative. As shown later in Sec.~\ref{reslargeDNglobal}, in the large-$D$ limit, the remaining eigenvalues are upper bounded in magnitude by $\frac{1}{d}[ 1+\mathcal{O} ( \frac{k^2}{D} ) ]$.

We now explain how to evaluate the contribution of a given spectral component of Te to
\(\mathrm{tr}(\mathrm{A}\mathrm{Te}^n)\). 

The matrix Te is not symmetric in the standard permutation basis, but it is self-adjoint with respect to the positive metric
\begin{equation}
    H=G_{dD}X^{-1}G_{dD}.
    \label{eq:H_metric_def}
\end{equation}
Indeed,
\begin{equation}
H\mathrm{Te}=G_{dD}X^{-1}G_{dD}\,G_{dD}^{-1}XG_D=G_{dD}G_D,
\end{equation}
whereas
\begin{equation}
\mathrm{Te}^{\dagger}H=G_DX\,G_{dD}^{-1}G_{dD}X^{-1}G_{dD}=G_DG_{dD}.
\end{equation}
Since $G_{dD}$ and $G_D$ are convolution matrices generated by class functions on \(S_k\) (the $G_D$ matrix can be diagonalized by a matrix that is independent of $D$), they commute:
\begin{equation}
    G_{dD}G_D=G_DG_{dD}.
\end{equation}
Hence
\begin{equation}
H\mathrm{Te}=\mathrm{Te}^\dagger H.
    \label{eq:H_self_adjoint}
\end{equation}
Let $|v_{f}^{\left( i \right)}\rangle $ be a right eigenvector of Te, associated with an eigenvalue $\lambda _{f}^{\left( i \right)}$:
\begin{equation}
    \mathrm{Te}|v_{f}^{\left( i \right)}\rangle =\lambda _{f}^{\left( i \right)}|v_{f}^{\left( i \right)}\rangle .
    \label{eq:vf_right_eigenvector}
\end{equation}
Here, $f$ denotes the matrix block in the Te matrix associated with $f$ fixed points and $i$ denotes the $i$-th eigenvalue in this diagonal block. 
The corresponding left eigenvector can be chosen as
\begin{equation}
 \langle L_{f}^{\left( i \right)}|=\frac{\langle v_{f}^{\left( i \right)}|H}{\langle v_{f}^{\left( i \right)}|H|v_{f}^{\left( i \right)}\rangle}=\frac{\langle v_{f}^{\left( i \right)}|G_{dD}X^{-1}G_{dD}}{\langle v_{f}^{\left( i \right)}|G_{dD}X^{-1}G_{dD}|v_{f}^{\left( i \right)}\rangle}.
    \label{eq:left_eigenvector_general}
\end{equation}
Then
\begin{equation}
    \langle L_{f}^{\left( i \right)}|\mathrm{Te}=\lambda _{f}^{\left( i \right)}\langle L_{f}^{\left( i \right)}|,\qquad \langle L_{f}^{\left( i \right)}|v_{f}^{\left( i \right)}\rangle =1.
\end{equation}
The spectral component associated with this chosen eigenvector is:
\begin{equation}
    \Pi _{f}^{\left( i \right)}=|v_{f}^{\left( i \right)}\rangle \langle L_{f}^{\left( i \right)}|=\frac{|v_{f}^{\left( i \right)}\rangle \langle v_{f}^{\left( i \right)}|H}{\langle v_{f}^{\left( i \right)}|H|v_{f}^{\left( i \right)}\rangle}.
    \label{eq:spectral_projector_single_vf}
\end{equation}
The contribution of this component to \(\mathrm{Te}^n\) is
\begin{equation}
 \left. \mathrm{Te}^n \right|_{v_{f}^{\left( i \right)}}={\lambda _{f}^{\left( i \right)}}^n\Pi _{v_{f}^{\left( i \right)}}.
\end{equation}
Therefore
\begin{align}
    \left. \mathrm{tr}\!\left( \mathrm{ATe}^{n-1} \right) \right|_{v_{f}^{\left( i \right)}}&={\lambda _{f}^{\left( i \right)}}^{n-1}\mathrm{tr}\!\left( \mathrm{A}\Pi _{v_{f}^{\left( i \right)}} \right) 
\\
&={\lambda _{f}^{\left( i \right)}}^{n-1}\frac{\langle v_{f}^{\left( i \right)}|H\mathrm{A}|v_{f}^{\left( i \right)}\rangle}{\langle v_{f}^{\left( i \right)}|H|v_{f}^{\left( i \right)}\rangle}.
    \label{eq:trace_BE_Te_vf_general_1}
\end{align}
Using Eqs.~\eqref{eq:BE_def_compact} and \eqref{eq:H_metric_def}, we find
\begin{align}
    H\mathrm{A}&=G_{dD}X^{-1}G_{dD}G_{dD}^{-1}X_OG_D
\\
&=G_{dD}X^{-1}X_OG_D.
\end{align}
Since both \(X\) and $X_O$ are diagonal in the permutation basis,
\begin{equation}
  X^{-1}X_O=\frac{1}{d^{\#\pi}}\prod_{(t_1,t_2,\dots ,t_m)\in \pi}{\mathrm{tr}\left[ \left( \hat{O}^{\left( 1 \right)} \right) ^m \right]}.
\end{equation}
Substituting this into Eq.~\eqref{eq:trace_BE_Te_vf_general_1}, we obtain the compact formula
\begin{equation}
    \left. \mathrm{tr}\!\left( \mathrm{ATe}^{n-1} \right) \right|_{v_{f}^{\left( i \right)}}={\lambda _{f}^{\left( i \right)}}^{n-1}\frac{\langle v_{f}^{\left( i \right)}|G_{dD}X^{-1}X_OG_D|v_{f}^{\left( i \right)}\rangle}{\langle v_{f}^{\left( i \right)}|G_{dD}X^{-1}G_{dD}|v_{f}^{\left( i \right)}\rangle}.
    \label{eq:trace_BE_Te_vf_general}
\end{equation}

Notably, as the Te matrix is upper triangular, it admits a complete set of eigenvectors. For each such eigenvector, let $(v_m^{(i)})_\delta$ denote its element at position $\delta$, which satisfies $(v_m^{(i)})_\delta = 0$ whenever $f(\delta) < m$

As a special case, take the eigenvector associated with the simple eigenvalue \(1\). 
This right eigenvector is
\begin{equation}
    |v_k\rangle=|e\rangle,
    \qquad
    \lambda_k=1.
\end{equation}
we have 
\begin{equation}
    \Pi _k=|e\rangle \langle L_k|=\frac{|e\rangle \langle e|H}{\langle e|H|e\rangle}
\end{equation}
Then Eq.~\eqref{eq:trace_BE_Te_vf_general} gives
\begin{align}
    \left.
    \mathrm{tr}\!\left(
    \mathrm{A}\mathrm{Te}^{n-1}
    \right)
    \right|_{|e\rangle}
    &=
    \frac{
    \langle e|G_{dD}X^{-1}X_OG_D|e\rangle
    }{
    \langle e|G_{dD}X^{-1}G_{dD}|e\rangle
    }.
    \label{eq:trace_BE_Te_e_general_formula}
\end{align}
The numerator is
\begin{align}
    \langle e|G_{dD}X^{-1}X_OG_D|e\rangle
    &=
    \sum_{\pi \in S_k}{G_{dD}}(e,\pi )G_D(\pi ,e)\frac{1}{d^{\#\pi}}\mathrm{tr}\left( \mathrm{P}_{\pi}O \right)               \\
    &=
    \sum_{\pi\in S_k}
    (dD)^{\#\pi}D^{\#\pi}  \frac{1}{d^{\#\pi}}\prod_{(t_1,t_2,\dots ,t_m)\in \pi}{\mathrm{tr}\left[ \left( \hat{O}^{\left( 1 \right)} \right) ^m \right]}                 \notag\\
    &=
   \sum_{\pi \in S_k}{D^{2\#\pi}}\prod_{(t_1,t_2,\dots ,t_m)\in \pi}{\mathrm{tr}\left[ \left( \hat{O}^{\left( 1 \right)} \right) ^m \right]}.
\end{align}
The denominator is
\begin{align}
    \langle e|G_{dD}X^{-1}G_{dD}|e\rangle
    &=
    \sum_{\sigma\in S_k}
    G_{dD}(e,\sigma)d^{-\#\sigma}G_{dD}(\sigma,e)
    \notag\\
    &=
    \sum_{\sigma\in S_k}
    d^{\#\sigma}D^{2\#\sigma}.
\end{align}
Therefore
\begin{equation}
\left. \mathrm{tr}\!\left( \mathrm{ATe}^{n-1} \right) \right|_{\lambda =1}=\frac{\sum_{\pi \in S_k}{D^{2\#\pi}\prod_{(t_1,t_2,\dots ,t_m)\in \pi}{\mathrm{tr}\left[ \left( \hat{O}^{\left( 1 \right)} \right) ^m \right]}}}{\sum_{\sigma \in S_k}{d^{\#\sigma}D^{2\#\sigma}}}.   \label{eq:trace_BE_Te_lambda1_explicit}
\end{equation}

Because the absolute value of the other eigenvalues of the Te matrix are less than $\frac{1}{d}[ 1+\mathcal{O} ( \frac{k^2}{D} ) ]$, when $n$ or $d$ is large enough, the dominant term of $\mathrm{tr}(\mathrm{ATe}^{n-1})$ is $ \mathrm{tr}( \mathrm{ATe}^{n-1} )|_{\lambda =1}$. 

Then when $D$ is large and $d^n \gg D^2$ or $d$ is large:
\begin{align}
    \mathrm{tr}\!\left( \mathrm{ATe}^n \right) &\approx \left. \mathrm{tr}\!\left( \mathrm{ATe}^n \right) \right|_{\lambda =1}
\\
&=\sum_{\pi \in S_k}{\prod_{(t_1,t_2,\dots ,t_m)\in \pi}{\left( D^2 \right) ^m\mathrm{tr}\left[ \left( \hat{O}^{\left( 1 \right)} \right) ^m \right] /\sum_{\sigma \in S_k}{d^{\#\sigma}D^{2\#\sigma}}}}
\\
&=\sum_{\pi \in S_k}{\prod_{(t_1,t_2,\dots ,t_m)\in \pi}{\mathrm{tr}\left[ \left( \hat{O}^{\left( 1 \right)}\otimes \mathcal{I} _{D^2} \right) ^m \right] /\sum_{\sigma \in S_k}{d^{\#\sigma}D^{2\#\sigma}}}}
\\
&=\frac{\left( dD^2 \right) ^k}{\overline{dD^2}}\sum_{\pi \in S_k}{\prod_{(t_1,t_2,\dots ,t_m)\in \pi}{\mathrm{tr}\left[ \left( \hat{O}^{\left( 1 \right)}\otimes \mathcal{I} _{D^2} \right) ^m \right] /}}\left( dD^2 \right) ^k
\\
&=\left( 1+\mathcal{O} \left( \frac{k^2}{dD^2} \right) \right) \sum_{\pi \in S_k}{\prod_{(t_1,t_2,\dots ,t_m)\in \pi}{\mathrm{tr}\left[ \left( \hat{O}^{\left( 1 \right)}\otimes \mathcal{I} _{D^2} \right) ^m \right] /}\left( dD^2 \right) ^k}\label{Tedominantofe1} 
\end{align}
where $\mathcal{I} _{D^2}$ is an identity matrix of dimension $D^2$ and in the third line we employ the generating function for Stirling numbers~\cite{Richard2011Combinatorics}:
\begin{equation}\label{stelingnum}
\overline{Dd}=\sum_{\sigma \in S_k}{\left( Dd \right) ^{\#\sigma}}=\prod_{i=0}^{k-1}{\left( Dd+i \right)}.
\end{equation}

The result in Eq.~(\ref{Tedominantofe1}) is also consistent with the structure of Eq.~(\ref{S40}). Analogous to the scenario in the preceding section, the operator in Eq.~(\ref{Tedominantofe1}) can be regarded as a local operator $\hat{O}^{(1)}\otimes \mathcal{I}_{D^2}$. Consequently, for sufficiently large $D^2$, this composite operator inherently satisfies the flat-spectrum condition regardless of the arbitrary local operator $\hat{O}^{(1)}$. Incorporating the condition $\mathrm{tr}(\hat{O}^{(1)}) = 0$, simplifying Eq.~(\ref{Tedominantofe1}) yields:
\begin{align}
    \mathrm{tr}\!\left( \mathrm{ATe}^n \right) \approx \left[ 1+\mathcal{O} \left( \frac{k^2}{dD^2} \right) \right] \left[ 1+\mathcal{O} \left( \frac{1}{dD^2} \right) \right] \left( \frac{k!}{\left( k/2 \right) !2^{k/2}}\frac{1}{d^{k/2}D^k} \right).\label{S101}
\end{align}

The leading term in Eq.~(\ref{S101}) corresponding to $\pi \in \mathcal{P}_k$ exactly matches the $k$th-order moment of a Gaussian distribution with variance \( \sigma^2 = 1/dD^2 \), while the remaining terms are suppressed by $1/D^2$. We will derive explicit corrections for finite $(D,n)$ in a subsequent section.

 \subsection{Global Output in the Large-$d$ Limit}
 
In this subsection, we demonstrate that in the $d \rightarrow \infty$ limit at finite $D$, $C(\hat{Q}_g)$ does not converge to a Gaussian distribution. This deviation can be explicitly captured by the excess kurtosis $\kappa$, a standard metric for Gaussianity. 

For a zero mean distribution, the excess kurtosis is defined as:
\begin{equation}
\kappa =\frac{\mathbb{E}\left( x^4 \right)}{\sigma ^4}-3,
\end{equation}
where $\mathbb{E}(x^4) \equiv m_4$ is the fourth moment and $\sigma^2 \equiv m_2$ is the variance. A true Gaussian distribution strictly yields $\kappa = 0$. To verify the non-Gaussianity of $C(\hat{Q}_g)$ in the large-$d$ limit, we can directly verify that its excess kurtosis does not vanish.  

Returning to Eq.~(\ref{S31}) and Eq.~(\ref{Cglobal}), the second and fourth moments of $C(\hat{Q}_g)$ are governed by:
  \begin{equation}\label{S112}
  \mathrm{E}\left( C\left( \hat{Q}_g \right) ^2 \right) =\mathrm{tr}\left( \mathrm{A}^n \right).
\end{equation}   
For $k=2$, the matrix $\mathrm{A}$ takes the form:
\begin{align}
\mathrm{A}&=\sum_{\pi}{\mathrm{Wg}\left( \delta \pi ^{-1},dD \right) d^{\#\pi}D^{\#\pi \tau ^{-1}}}
\\
&=\left( \begin{matrix}
	\frac{1}{D^2d^2-1}&		\frac{-1}{Dd\left( D^2d^2-1 \right)}\\
	\frac{-1}{Dd\left( D^2d^2-1 \right)}&		\frac{1}{D^2d^2-1}\\
\end{matrix} \right) \left( \begin{matrix}
	0&		0\\
	0&		d\\
\end{matrix} \right) \left( \begin{matrix}
	D^2&		D\\
	D&		D^2\\
\end{matrix} \right) 
\\
&=\left( \begin{matrix}
	\frac{-1}{\left( D^2d^2-1 \right)}&		\frac{-D^2d}{Dd\left( D^2d^2-1 \right)}\\
	\frac{Dd}{D^2d^2-1}&		\frac{D^2d}{D^2d^2-1}
\end{matrix} \right).
\end{align}
The eigenvalues of $\mathrm{A}$ are $0$ and $\frac{D^2d-1}{D^2d^2-1}$. Consequently,  
\begin{align}
\mathbb{E}\left( C\left( \hat{Q}_g \right) ^2 \right) &=\mathrm{tr}\left( \mathrm{A}^n \right) =\sum_i{\lambda _{i}^{n}}
\\
&=\left( \frac{D^2d-1}{D^2d^2-1} \right) ^n.
\end{align}
In the large-$d$ limit, the second moment scales asymptotically as $m_2\sim d^{-n}$.

For the $k=4$, the matrix $\mathrm{A}$ possesses six distinct eigenvalues.    
In the asymptotic limit $d\rightarrow \infty$, these eigenvalues simplify to $e_1=0$ (multiplicity 15), $e_2\sim (1-D^{-2})d^{-3}$ (multiplicity 3), $e_3\sim \mathcal{O} \left( d^{-3} \right) $ (multiplicity 1), $e_4\sim (1+2D^{-2})d^{-2}$ (multiplicity 1), $e_5\sim \mathcal{O} \left( d^{-3} \right) $ (multiplicity 2), and $e_6\sim (1-D^{-2})d^{-2}$ (multiplicity 2). Consequently, the leading-order scaling of the fourth moment is $m_4\sim \left( \frac{1+2/D^2}{d^2} \right) ^n+2\left( \frac{1-1/D^2}{d^2} \right) ^n$.

Using these moments, the asymptotic excess kurtosis $\kappa_4$ of $C(\hat{Q}_g)$ evaluates to:
\begin{align}
\kappa &=\left( \left( \frac{1+2/D^2}{d^2} \right) ^n+2\left( \frac{1-1/D^2}{d^2} \right) ^n \right) d^{2n}-3
\\
&=\left( 1+\frac{2}{D^2} \right) ^n+2\left( 1-\frac{1}{D^2} \right) ^n-3\ne 0.
\end{align}
This explicitly demonstrates that for $d\rightarrow \infty$ at finite $D$, $\kappa \neq 0$, confirming that $C(\hat{Q}_g)$ is strictly non-Gaussian.  

To gain deeper insight, we analyze the asymptotic behavior of arbitrary higher-order moments $\mathbb{E} ( C( \hat{Q}_g ) ^k )  $ in the large-$d$ limit.   

Following Eq.~(\ref{S112}), evaluating the asymptotic eigenvalues of $\mathrm{A}_{\delta,\tau} = \sum_{\pi \in \mathcal{E}_k} \mathrm{Wg}(\delta \pi^{-1}, dD) d^{\#\pi} D^{\#(\pi \tau^{-1})}$ becomes essential. Here, the maximal number of cycles, $\#\pi = k/2 \equiv m$, corresponds precisely to $\pi$ are disjoint transpositions. Hence, in the large-$d$ limit, the matrix $\mathrm{A}$ is well-approximated by:
\begin{equation}
\mathrm{M}_{\delta,\tau}=\sum_{\pi\in\mathcal{P}_{2m}}\mathrm{Wg}(\delta\pi^{-1},dD)d^{\#\pi}D^{\#(\pi\tau^{-1})},\qquad \delta,\tau\in S_{2m}, 
\end{equation}
where $\mathcal{P}_{2m}$ denotes the set of disjoint permutations in $S_{2m}$. 

Because the Weingarten matrix is the formal inverse of the Gram matrix, the above formulation can be recast as:
\begin{equation}
    \mathrm{M}=d^m G_{dD}^{-1} X_P G_D,
\end{equation}
where $X_P$ is a diagonal matrix whose diagonal entries are $1$ on the index set $\mathcal{P}_{2m}$ and $0$ elsewhere.

This factorization naturally bounds the rank such that $\operatorname{rank} (\mathrm{M}) \le |\mathcal{P}_{2m}|=(2m-1)!!$. Consequently, at least $(2m)!-(2m-1)!!$ eigenvalues are identically zero. 

To extract the non-trivial eigenvalues, we exploit the linear algebra property that cyclic permutations of matrix products (e.g. $AB$ and $BA$) preserve non-zero eigenvalues and their algebraic multiplicities. We thus consider the modified matrix:
\begin{equation}
    \mathrm{B}=d^m X_P G_D G_{dD}^{-1} X_P, 
\end{equation}
which acts entirely within the $\mathcal{P}_{2m}$ subspace.

According to Eq.~(\ref{Wgj}), in the large-$d$ limit, the inverse Gram matrix expands as:
\begin{equation}
   G_{dD}^{-1}=(dD)^{-2m}\left[ I+\mathcal{O} (d^{-1}) \right] . 
\end{equation}
Substituting this expansion into $\mathrm{B}$ yields: 
\begin{align}
    \mathrm{B}&= d^m X_P G_D G_{dD}^{-1} X_P \nonumber\\
    &= d^m(dD)^{-2m} X_P G_D X_P+\mathcal{O}(d^{-m-1}) \nonumber\\ 
    &= d^{-m}\mathbb{G}(D)+\mathcal{O}(d^{-m-1}),
\end{align}
where we have introduced the sub-Gram matrix $\mathbb{G}(D)$:
\begin{equation}
    \mathbb{G}(D)=D^{-2m}X_PG_DX_P. 
\end{equation}
Element-wise, we can yield: 
\begin{equation}
    \mathbb{G}_{\pi,\rho}(D)=D^{\#(\pi\rho^{-1})-2m} \quad \pi,\rho \in \mathcal{P}_{2m}  \label{eqGamma_def_appendix} 
\end{equation}
The explicit diagonalization of this matrix follows from known combinatorial results: 

\begin{lemma}[Banica--Curran pairing Gram matrix]
\label{lem:BC_pairing_gram}
Let $\mathcal{P}_o(2m)$ be the set of all pair partitions of $[2m]$. For $\pi,\rho \in \mathcal{P}_o(2m)$, define
\begin{equation}
G^{(o)}_{2m,n}(\pi,\rho) = n^{|\pi\vee\rho|},
\label{eq:BC_gram_def}
\end{equation}
where $\pi\vee\rho$ is the join of $\pi$ and $\rho$ in the lattice of set partitions, and $|\pi\vee\rho|$ is the number of blocks of this join. Then
\begin{equation}
G^{(o)}_{2m,n} = \sum_{\lambda\vdash m} \left[ \prod_{(i,j)\in\lambda} \left(n+2j-i-1\right) \right] P_{2\lambda},
\label{eq:BC_diagonalization}
\end{equation}
where $\lambda\vdash m$ runs over all partitions of $m$ and $P_{2\lambda}$ is the projection onto the irreducible $S_{2m}$-module $V_{2\lambda}$ inside $\mathbb{C}[\mathcal{P}_o(2m)]$. Consequently, the eigenvalue associated with $V_{2\lambda}$ is
\begin{equation}
f_n(\lambda) = \prod_{(i,j)\in\lambda} \left(n+2j-i-1\right),
\label{eq:fn_lambda}
\end{equation}
with multiplicity $f^{2\lambda} = \dim V_{2\lambda}$.
\end{lemma}
Here, $\lambda\vdash m$ denotes an integer partition of $m$ (i.e., $\lambda=(\lambda_1,\lambda_2,\ldots)$ with $\sum \lambda_i =m$ and $\lambda_1 \ge \lambda_2 \ge \ldots > 0$), and $(i,j) \in \lambda$ corresponds to traversing all boxes in the Young diagram of $\lambda$.
\begin{proof}
This is the orthogonal (pairing) case of Theorem~4.1 in Ref.~\cite{Banica2010Gram}. In their notation, $\mathcal{P}_o$ denotes the category of all pairings, and the corresponding Gram matrix is defined by Eq.~\eqref{eq:BC_gram_def}. The diagonal form in Eq.~\eqref{eq:BC_diagonalization} is the spectral decomposition used in the proof of their determinant formula.
\end{proof}

We now relate the matrix $\mathbb{G}(D)$ in Eq.~(\ref{eqGamma_def_appendix}) to $G^{(o)}_{2m,n}$. A disjoint permutation can be represented equivalently as a pair partition via
\begin{equation}
(a_1b_1)\cdots(a_mb_m) \quad\longleftrightarrow\quad \big\{\{a_1,b_1\},\ldots,\{a_m,b_m\}\big\}.
\end{equation}
Thus, we formally identify $\mathcal{P}_{2m}$ with $\mathcal{P}_o(2m)$ in what follows.

\begin{lemma}
\label{lem:join_cycles}
For any $\pi,\rho \in \mathcal{P}_{2m}$, one has
\begin{equation}
\#(\pi\rho^{-1}) = 2|\pi\vee\rho|.
\label{eq:cycles_join_identity}
\end{equation}
\end{lemma}

\begin{proof}
Note that $\rho^{-1}=\rho$. Superimpose the two pairings $\pi$ and $\rho$ as a two-colored graph on the vertex set $[2m]$, with one edge from $\pi$ and one edge from $\rho$ incident to each vertex. Every connected component is an alternating even cycle (allowing the length-two case when $\pi$ and $\rho$ share a pair).

The blocks of $\pi\vee\rho$ are precisely the connected components of this graph. Indeed, the join $\pi\vee\rho$ is the smallest set partition coarser than both $\pi$ and $\rho$, which is the equivalence relation generated by the pair relations of $\pi$ and $\rho$. Hence, two vertices belong to the same block of $\pi\vee\rho$ if and only if they lie in the same connected component of the superposition graph. Thus, $|\pi\vee\rho|$ is the number of connected components.

It remains to count the cycles of $\pi\rho$ on a single connected component. Suppose one component has $2r$ vertices, labeled as $a_1,b_1,a_2,b_2,\ldots,a_r,b_r$ with indices modulo $r$, such that $\pi(a_s)=b_s$ and $\rho(b_s)=a_{s+1}$. Then $\rho(a_s)=b_{s-1}$, which implies
\begin{align}
\pi\rho(a_s) &= \pi(b_{s-1}) = a_{s-1}, \\
\pi\rho(b_s) &= \pi(a_{s+1}) = b_{s+1}.
\end{align}
Hence, this connected component contributes exactly two cycles to $\pi\rho$: one supported on $\{a_1,\ldots,a_r\}$ and another on $\{b_1,\ldots,b_r\}$. Since every connected component contributes two cycles, we obtain $\#(\pi\rho) = 2|\pi\vee\rho|$. Using $\rho^{-1}=\rho$ yields Eq.~\eqref{eq:cycles_join_identity}.
\end{proof}

Synthesizing Lemma~\ref{lem:BC_pairing_gram} and Lemma~\ref{lem:join_cycles}, we obtain the explicit diagonalization of $\mathbb{G}(D)$. Consequently, the asymptotic eigenvalues of $\mathrm{A}$ in the large-$d$ limit are given by:
\begin{equation}\label{dinfinGmatrixevlen}
    \lambda _{\mu ,a}(\mathrm{A})=d^{-k/2}D^{-k}\prod_{(i,j)\in \mu}{\left( D^2+2j-i-1 \right)}+O(d^{-k/2-1}), 
\end{equation}
where the index $a=1,\ldots,f^{2\mu}$ accounts for the multiplicity of the eigenvalue associated with the partition $\mu\vdash m$. The remaining $(2m)!-(2m-1)!!$ eigenvalues are identically zero. By examining the case of $k=4$, this general formula perfectly recovers the asymptotic eigenvalues of $\mathrm{A}$ derived explicitly earlier.  

For finite $k$ with $nk^2\ll D^2$, Eq.~(\ref{dinfinGmatrixevlen}) indicates that the eigenvalues are dominated by $D$, converging to $d^{-k/2}$. Substituting this result back into Eq.~(\ref{Cglobal}), the deviation of the $k$th-order moment from that of a true Gaussian distribution scales as:
\begin{equation}
    \mathcal{O} \left( \frac{nk^2}{D^2} \right) .
\end{equation}
This scaling reveals that taking $(d,n) \rightarrow \infty$ alone is insufficient to guarantee Gaussianity for $C(\hat{Q}_g)$; however, if one subsequently takes the limit $D^2/n\rightarrow \infty$, the Gaussian distribution is naturally recovered. 

Conversely, in the strict $D = 1$ limit, only the single-row partition (i.e. $\mu =(k/2)$ with Young diagram boxes $(i,j)=(1,1), (1,2), \ldots, (1,k/2)$) yields a non-zero eigenvalue; all other partitions contain the box $(2,1)$, which strictly sends the eigenvalue to zero. Substituting this into Eq.~(\ref{Cglobal}) yields:
\begin{equation}
   \mathbb{E} \left( C\left( \hat{Q}_g \right) ^k\right) \sim d^{-kn/2}\left[ \left( k-1 \right) !! \right] ^n=d^{-kn/2}\left[ \frac{k!}{\left( k/2 \right) !2^{k/2}} \right] ^n,
\end{equation}
which exactly matches the $k$th-order moment of an $n$th-order Gaussian chaos $Z_n=g_1g_2\cdots g_n$, where $g_i\sim \mathcal{N}\left( 0,\sigma ^2 \right)$.

If we strictly take $(d,n) \rightarrow \infty$ while keeping $D$ finite, Eq.~(\ref{S112}) dictates that the $k$th-order moment is dominated by the maximum eigenvalue $\lambda_{\max}$. Identifying the single-row partition as the source of this maximum eigenvalue, in large $k$ regime, we evaluate:
\begin{align}
\mathbb{E}\left( C\left( \hat{Q}_g \right) ^k \right) &\approx \left( \lambda _{\max} \right) ^n
\\
&=\left( d^{-k/2}\prod_{r=0}^{k/2-1}{\left( 1+\frac{2r}{D^2} \right)}+\mathcal{O} \left( d^{-k/2-1} \right) \right) ^n
\\
&\approx \left( d^{-k/2}\prod_{r=0}^{k/2-1}{\left( 1+\frac{2r}{D^2} \right)} \right) ^n
\\
&=d^{-kn/2}\left( \frac{\Gamma \left( k/2+D^2/2 \right)}{\Gamma \left( D^2/2 \right)}\left( \frac{2}{D^2} \right) ^{k/2} \right) ^n
\\
&\approx \left( \frac{\sqrt{2\pi}}{\Gamma \left( D^2/2 \right)}\left( \frac{k}{eD^2d} \right) ^{k/2}\left( \frac{k}{2} \right) ^{\frac{D^2-1}{2}} \right) ^n,
\end{align}
where $\Gamma(\cdot)$ is the Gamma function. In the final line, we apply Stirling's approximation $\Gamma(x) \sim \sqrt{2\pi} x^{x-1/2} e^{-x}$ and $k \gg D^2$.

For an $n$th-order Gaussian chaos, the moments obey the bound~\cite{Janson_1997_Gaussian}:
\begin{equation}
    \frac{\mathbb{E} \left( \left| X \right|^k \right) ^{1/k}}{\mathbb{E} \left( \left| X \right|^2 \right) ^{1/2}}\leqslant \left( k-1 \right) ^{n/2}.
\end{equation}
While a standard Gaussian distribution scales as $\frac{\mathbb{E}(|X|^k)^{1/k}}{\mathbb{E}(|X|^2)^{1/2}} \sim k^{1/2}$, for our operator we find: 
\begin{equation}
    \frac{\mathbb{E}\left( C\left( \hat{Q}_g \right) ^k \right) ^{1/k}}{\mathbb{E}\left( C\left( \hat{Q}_g \right) ^2 \right) ^{1/2}}\sim \left( \frac{k}{eD^2} \right) ^{n/2},
\end{equation}
which confirms that in the large-$d$ limit, the higher-order moments of $C(\hat{Q}_g)$ closely mimic those of an $n$th-order Gaussian chaos.  

In conclusion, under the joint limit $(d,n) \rightarrow \infty$ at finite $D$, the global measurement distribution $C(\hat{Q}_g)$ does not approach a Gaussian. Instead, it interpolates between a pure Gaussian and an $n$th-order Gaussian chaos $Z_n = \prod_{i=1}^n g_i$. Taking $D \rightarrow \infty$ recovers the pure Gaussian limit, whereas setting $D=1$ perfectly reduces to the Gaussian chaos. 

Furthermore, we note that such an non-Gaussianity can extend to more general cases. Consider an arbitrary operator formed by a finite linear combination of Pauli strings. The resulting asymptotic distribution naturally takes the composite form $C(\hat{O}_g) = \sum_i c_i C(\hat{Q}_{gi})$, which generally remains non-Gaussian. This result demonstrates a clear breakdown of Gaussian universality for global observables in this regime.

\subsection{Global Output in the Large-$(D,n)$ Limit}\label{sectionglobal}
In this subsection, we evaluate the $k$th-order moment of the global output in the large-$(D,n)$ limit, which requires computing the trace of the matrix product:

\begin{equation}\label{S145}
\mathbb{E}\left( X_g ^k \right) =\mathrm{tr}\left( \prod_s{\mathrm{A}^{\left( s \right)}} \right) . 
\end{equation}  

First, we establish that in the large-$D$ limit, the matrix $\mathrm{A}^{(s)}$ becomes diagonal. Recall its definition:
\begin{align}
    \mathrm{A}^{(s)}_{\delta,\tau}&=\sum_{\pi \in S_k}{\mathrm{Wg}\left( \delta \pi ^{-1},dD \right) \mathrm{tr}\left( \mathrm{P}_{\pi}O_s \right) D^{\#\left( \pi \tau ^{-1} \right)}}
\\
&\approx \sum_{\pi \in S_k}{\mathrm{Moeb}\left( \delta \pi ^{-1} \right) \left( Dd \right) ^{-2k+\#\left( \delta \pi ^{-1} \right)}\prod_{(t_1,t_2,\cdots ,t_m)\in \pi}{\mathrm{tr}}\left( \hat{O}^{(s)m} \right) D^{\#\left( \pi \tau ^{-1} \right)}},
\end{align}
where the matrix elements attain their leading-order contributions when $\delta =\pi =\tau $, confirming that the dominant entries reside strictly on the diagonal. Substituting this asymptotic form into Eq.~(\ref{S145}) yields:
\begin{align}
    \mathbb{E} \left( X_{g}^{k} \right) &\approx \sum_{\pi}{\prod_s{\frac{1}{d^k}}}\prod_{(t_1,t_2,\cdots ,t_m)\in \pi}{\mathrm{tr}}\left( \hat{O}^{(s)m} \right) 
\\
&=\sum_{\pi}{\frac{1}{d^{kn}}}\prod_{(t_1,t_2,\cdots ,t_m)\in \pi}{\mathrm{tr}}\left( \hat{O}^m \right),
\end{align}
where $\hat{O}=\hat{O}^{(1)}\otimes \hat{O}^{(2)}\otimes \cdots \otimes \hat{O}^{(n)}$. The result above is also consistent with the structure of Eq.~(\ref{S40}). This necessitates imposing constraints on the global operator $\hat{O}$. We assume $\hat{O}$ is spanned by a set of operator bases $\{ \hat{Q}_{gj} \}_{j=1}^m$ that satisfy the Clifford algebra. Let us first consider the $k$th-order moments associated with these operators bases in the large-$D$ limit:
\begin{equation}
\mathbb{E} \left( C\left( \hat{Q}_g \right) ^k \right) = \mathrm{tr}\left( \mathrm{A}^n \right),
\end{equation}
where the matrix $\mathrm{A}$ is given by:
\begin{equation}
\mathrm{A}_{\delta,\tau}=\sum_{\pi \in \mathcal{E} _k}{\mathrm{Wg}\left( \delta \pi ^{-1},dD \right) d^{\#\pi}D^{\#\left( \pi \tau ^{-1} \right)}}.
\end{equation}

We can formally decompose this as a matrix product $\mathrm{A}=WX^{\prime}Y$. Through similarity transformations, $\mathrm{A}$ is similar to $X^{\prime}YW$, which is in turn similar to $(X^{\prime})^{1/2}YW(X^{\prime})^{1/2}$. Since $Y$ and $W$ are simultaneously diagonalizable and both are symmetric, their product $YW$ is symmetric. Consequently, the transformed matrix $(X^{\prime})^{1/2}YW(X^{\prime})^{1/2}$ is also symmetric, guaranteeing that $\mathrm{A}$ is diagonalizable. Furthermore, because $X^{\prime}$, $Y$, and $W$ are all positive definite or positive sdefinite, the eigenvalues of $\mathrm{A}$ are strictly non-negative.     

In the limit $D \rightarrow \infty$, the matrix $\mathrm{A}$ becomes diagonally dominant, allowing its eigenvalues to be robustly approximated by its diagonal entries. Specifically, for $\delta \notin \mathcal{E}_k$, the diagonal elements are suppressed by at least $\mathcal{O}(D^{-1})$ and thus vanish asymptotically. Conversely, when $\delta \in \mathcal{E}_k $, the diagonal elements scale asymptotically as $\mathrm{A}_{\pi ,\pi} \approx d^{\#\pi}/d^k$.   

Applying this diagonal approximation, the $k$th-order moment of $\mathbb{E} ( C( \hat{Q}_g ) ^k ) $ evaluates to:  
\begin{align} 
\mathbb{E} \left( C\left( \hat{Q}_g \right) ^k \right) &=\sum_i{\lambda _{i}^{n}}
\\
&\approx \sum_{\pi \in \mathcal{E}_k }{\mathrm{A}_{\pi ,\pi}^{n}}
\\
&=\sum_{\pi \in \mathcal{E}_k}{\frac{d^{n\#\pi}}{d^{nk}}}
\\
&=\frac{k!}{\left( k/2 \right) !2^{k/2}}\frac{1}{d^{nk/2}}+\mathcal{O} \left( \frac{1}{d^{nk/2+n}} \right).\label{ecojdominant} 
\end{align}

Furthermore, extending our focus to the distribution sequence $\{ C( \hat{O}_{gi})\}_{i=1}^p$, the analysis in Sec.~\ref{ourProofMethod} dictates that establishing its convergence to a Gaussian process requires proving that the sequence $\{ C( \hat{Q}_{gj})\}_{j=1}^m$ is jointly Gaussian. This requires that the mixed moments $\mathbb{E} ( C( \hat{Q}_{gj_1} ) C( \hat{Q}_{gj_2} ) \cdots C( \hat{Q}_{gj_k}) ) $ formed by $\{ C( \hat{Q}_{gj})\}_{j=1}^m$ satisfy the pairing rule, while the uniqueness of the distribution is guaranteed by Carleman's condition.

\begin{lemma}[Isserlis's theorem{\cite{Isserlis_1918_Product}}]
Let $Z_1,\ldots,Z_k$ be centered jointly Gaussian random variables. For odd
$k$, one has
\begin{equation}
    \mathbb E[Z_1\cdots Z_k]=0,
\end{equation}
while for even $k$,
\begin{equation}
    \mathbb E[Z_1\cdots Z_{k}]
    =
    \sum_{\mathcal P\in\mathcal P_k}
    \prod_{\{a,b\}\in\mathcal P}
    \mathbb E[Z_aZ_b],
\label{eq:wick_theorem}
\end{equation}
where $\mathcal P_k$ denotes the set of all disjoint pairings of
$\{1,\ldots,k\}$.
\end{lemma}

By the Isserlis's theorem, it is sufficient to show that the mixed
moments of the outputs $\{ C( \hat{Q}_{gj})\}_{j=1}^m$ asymptotically satisfy the pairing rule.

The product $C( \hat{Q}_{j_1} ) C( \hat{Q}_{j_2} ) \cdots C( \hat{Q}_{j_k} )$ (throughout the remainder of this subsection, we suppress the subscript $g$ for notational simplicity), where operator bases are allowed to appear repeatedly, shares a structure analogous to that of $X_g^k$. Its tensor network diagram closely parallels that in Fig.~\ref{fig:Cgmk}, except that the operators across different copies need not be identical. Similarly, we define the corresponding matrix on each local site:
\begin{equation}
    \mathrm{A}_{_{\delta ,\tau}}^{\left( s \right)}=\sum_{\pi \in S_k}{\mathrm{Wg}\left( \delta \pi ^{-1},dD \right) \mathrm{tr}\left( \mathrm{P}_{\pi}Q_s \right) D^{\#\left( \pi \tau ^{-1} \right)}},
\end{equation}
where $Q_s=\hat{Q}_{j_1}^{\left( s \right)}\otimes \hat{Q}_{j_2}^{\left( s \right)}\otimes \cdots \otimes \hat{Q}_{j_k}^{\left( s \right)}$. Similarly, in the large-$D$ limit, the matrix
$\mathrm{A}^{(s)}$ becomes diagonal. Therefore,
\begin{align}
\mathbb{E}\left[
C( \hat{Q}_{j_1} ) C( \hat{Q}_{j_2} ) \cdots C( \hat{Q}_{j_k} )
\right]
\nonumber&=\operatorname{tr}
\left(
\prod_s \mathrm{A}^{(s)}
\right)
\\
&\approx
\sum_{\pi\in S_k}
\prod_s
\frac{1}{d^k}
\operatorname{tr}
\left(
\mathrm{P}_{\pi}O_s
\right)
\label{hunheju}
\\
&=
\sum_{\pi\in S_k}
\frac{1}{d^{kn}}
\prod_{(t_1,t_2,\ldots,t_l)\in\pi}
\operatorname{tr}
\left(
\hat Q_{t_1}
\hat Q_{t_2}
\cdots
\hat Q_{t_l}
\right).
\label{eq:mixed_moment_clifford}
\end{align}

Notice that any permutation $\pi \in S_k$ can be uniquely decomposed into a set of disjoint cycles. This cycle structure naturally induces a set partition of the index set $K = \{1, 2, \dots, k\}$. Let $\Pi_k$ denote the set of all possible set partitions of $K$. For a given partition $\Delta \in \Pi_k$, it is composed of disjoint blocks, which we denote by $B \in \Delta$. And we denote the positions in this block by $B=\{t_1,t_2,\ldots,t_l\}$.

For a specific block $B$ of size $l = |B|$, the elements can form $(l-1)!$ distinct cyclic permutations. We define the block contribution by summing over all such cyclic permutations:
\begin{equation}\label{eq:cycle_definition}
   \mathcal{C} (B)=\sum_{\sigma \in \mathrm{Cyc(}B)}{\mathrm{tr}\left( \hat{Q}_{\sigma \left( t_1 \right)}\hat{Q}_{\sigma \left( t_2 \right)}\cdots \hat{Q}_{\sigma \left( t_l \right)} \right) .}
\end{equation}
where $\mathrm{Cyc}(B)$ denotes the set of all cyclic permutations supported on the block $B$.
By grouping the permutations according to their underlying set partitions, the summation over the symmetric group can be systematically rewritten as a sum over all partitions and a product over their respective blocks:
\begin{equation}\label{S162}
    \sum_{\pi \in S_k} \frac{1}{d^{kn}} \prod_{(t_1,\dots,t_l)\in \pi} \mathrm{tr}\left( \hat{Q}_{t_1}\cdots \hat{Q}_{t_l} \right)
    = \frac{1}{d^{kn}} \sum_{\Delta \in \Pi_k} \prod_{B \in \Delta} \mathcal{C}(B).
\end{equation}




Returning to the Eq.~(\ref{eq:cycle_definition}), due to the cyclic invariance of the trace, every single cycle of length $l$
corresponds to $l$ equivalent permutations in the complete permutation sum.
Therefore,
\begin{equation}
\mathcal{C} (B)=\frac{1}{l}\sum_{\sigma \in S_l}{\mathrm{tr}\left( \hat{Q}_{\sigma \left( t_1 \right)}\hat{Q}_{\sigma \left( t_2 \right)}\cdots \hat{Q}_{\sigma \left( t_l \right)} \right) },
\label{eq:cycle_symmetrization}
\end{equation}
where $S_l$ denotes the permutation group acting on the $l$ positions inside
this cycle and $\left| \left. S_l \right| \right. =l!$.

To evaluate this expression, define the auxiliary element
\begin{equation}
    \hat{Q}(\boldsymbol{x})=\sum_{i=1}^l{x_i\hat{Q}_{t_i}} .
\end{equation}
Expanding its $l$-th power gives
\begin{equation}
\mathrm{tr}\left( \hat{Q}(\boldsymbol{x})^l \right) =\sum_{a_1,...,a_l=1}^l{x_{a_1}}\cdots x_{a_l}\mathrm{tr}\left( \hat{Q}_{t_{a_1}}\cdots \hat{Q}_{t_{a_l}} \right) .
\end{equation}
Hence, the coefficient of the monomial
$x_1x_2\cdots x_l$ contains exactly all permutations of the $l$ positions:
\begin{equation}
\mathrm{Coeff}_{x_1x_2\cdots x_l}
\left\{
\mathrm{tr}
\left[
\hat Q(\boldsymbol{x})^l
\right]
\right\}
=
\sum_{\sigma\in S_l}
\mathrm{tr}
\left(
\hat Q_{{\sigma(t_1)}}
\hat Q_{{\sigma(t_2)}}
\cdots
\hat Q_{{\sigma(t_l)}}
\right).
\label{eq:coefficient_permutation}
\end{equation}
Consequently,
\begin{equation}
\mathcal C(B)
=
\frac{1}{l}
\mathrm{Coeff}_{x_1x_2\cdots x_l}
\left\{
\mathrm{tr}
\left(
\hat Q(\boldsymbol{x})^l
\right)
\right\}.
\label{eq:cycle_coefficient_relation}
\end{equation}

The Clifford algebra gives
\begin{equation}
\begin{aligned}
\hat Q(\boldsymbol{x})^2
&=
\sum_{a,b=1}^l{x_ax_b\hat{Q}_{t_a}\hat{Q}_{t_b}}
\\
&=\sum_{a,b=1}^l{x_ax_b\delta _{t_at_b}I}.
\end{aligned}
\label{eq:clifford_square}
\end{equation}
Therefore,
\begin{equation}
\mathrm{tr}
\left(
\hat Q(\boldsymbol{x})^{2r}
\right)
=
{d^n}\left(
\sum_{a,b=1}^l{x_ax_b\delta _{t_at_b}}
\right)^r ,
\label{eq:even_power}
\end{equation}
while
\begin{equation}
\mathrm{tr}
\left(
\hat Q(\boldsymbol{x})^{2r+1}
\right)
=0 .
\label{eq:odd_power}
\end{equation}

For an odd $l$, Eq.~(\ref{eq:odd_power}) implies that every
coefficient $x_1x_2\cdots x_l$ in the polynomial expansion of
$\mathrm{tr}(\hat Q(\boldsymbol{x})^l)$ vanishes. Therefore, for odd $l$, the sum over all possible cyclic orderings in $\mathcal{C}(B)$ cancels out and yields a zero contribution:
\begin{equation}
    \mathcal{C}(B) = 0, \qquad l \text{ is odd}.
\end{equation}

For an even $l = 2r$, the non-vanishing contributions are obtained by extracting the coefficient of $x_1x_2\cdots x_{2r}$ from Eq.~(\ref{eq:even_power}). Such a coefficient is generated by grouping the $2r$ variables into $r$ pairs. Consequently, we have
\begin{align}
    \mathcal{C}(B) &=\frac{2^rr!}{l}d^n\sum_{\mathcal{P} \in \mathcal{P} _l}{\prod_{\{a,b\}\in \mathcal{P}}{\delta _{t_at_b}}} \\
    &=\frac{2^{r-1}}{d^{n(r-1)}}\times(r-1)!\sum_{\mathcal{P} \in \mathcal{P} _l}{\prod_{\{a,b\}\in \mathcal{P}}{\mathrm{tr}}\left( \hat{Q}_{t_a}\hat{Q}_{t_b} \right)}.
    \label{eq:even_cycle_pairing}
\end{align}
where the factor $2^r r!$ represents the total number of orderings for a fixed pairing, combining the $r!$ permutations of the pairs with a factor of $2^r$ from the internal swaps within each pair. Hence, the contribution of an even $l$ does not vanish; rather, after summing over all possible cyclic permutations, it reduces to a sum over products of pairwise traces.

Substituting the expression for $\mathcal{C}(B)$ back into Eq.~(\ref{S162}), we observe that any set partition $\Delta$ containing a block of odd size evaluates to zero. Consequently, for an odd $k$, the entire sum in Eq.~(\ref{S162}) vanishes. For an even $k$, the surviving contributions arise exclusively from partitions $\Delta$ composed entirely of even-sized blocks. 

By combining Eq.~(\ref{S162}) and Eq.~(\ref{eq:even_cycle_pairing}), the evaluation of the sum over permutations is reduced to a two-step combinatorial problem: partitioning the index set $K$ into even blocks, and subsequently forming disjoint pairs within each block weighted by the cyclic permutation factor $(r-1)!$. Combinatorially, this process is strictly isomorphic to an alternative procedure: first, selecting a global disjoint pairing $\mathcal{P} \in \mathcal{P}_k$ on the entire index set $K$ (thereby forming $m = k/2$ pairs), and second, partitioning these $m$ pairs into $s$ ($s=1,\cdots ,m$) disjoint blocks, where the pairs inside each block are arranged in a specific cyclic order.

The number of ways to arrange $m$ elements into $s$ non-empty cycles is precisely given by the unsigned Stirling numbers of the first kind, $\left[ \begin{smallmatrix} k/2 \\ s \end{smallmatrix} \right]$ (in particular $\left[ \begin{smallmatrix} k/2 \\ k/2 \end{smallmatrix} \right]=1$)~\cite{Richard2011Combinatorics}. As seen from Eq.~(\ref{eq:even_cycle_pairing}), partitioning the $k/2$ pairs into $s$ cycles introduces a cycle-dependent weight coefficient $\left( \frac{2}{d^n} \right)^{k/2-s}$. Collecting these elements together, we arrive at the final evaluation:
\begin{align}
    \frac{1}{d^{kn}} \sum_{\Delta \in \Pi_k} \prod_{B \in \Delta} \mathcal{C}(B) 
    &= \frac{1}{d^{kn}}\sum_{\mathcal{P} \in \mathcal{P}_k} \prod_{\{a,b\} \in \mathcal{P}} \mathrm{tr}\left( \hat{Q}_{j_a}\hat{Q}_{j_b} \right) \times \sum_{s=1}^{k/2} \begin{bmatrix} k/2 \\ s \end{bmatrix} \left( \frac{2}{d^n} \right)^{k/2-s} \nonumber \\
    &= \left[ 1 + \mathcal{O}\left( \frac{k^2}{d^n} \right) \right] \frac{1}{d^{kn}}\sum_{\mathcal{P} \in \mathcal{P}_k} \prod_{\{a,b\} \in \mathcal{P}} \mathrm{tr}\left( \hat{Q}_{j_{a}}\hat{Q}_{j_{b}} \right).
    \label{eq:final_stirling_sum}
\end{align}
We will explicitly compute the exact summation involving the Stirling numbers in the subsequent section, for now, we simply present the leading-order contribution and its asymptotic scaling in large-$n$ regime.

It follows that
\begin{equation}
\mathbb{E} \left[ C(\hat{Q}_{j_1})\cdots C(\hat{Q}_{j_k}) \right] \approx \left[ 1+\mathcal{O} \left( \frac{k^2}{d^n} \right) \right] \sum_{\mathcal{P} \in \mathcal{P}_k}{\prod_{\{a,b\}\in \mathcal{P}}{\mathbb{E}}}\left[ C(\hat{Q}_{j_a})C(\hat{Q}_{j_b}) \right] .
\label{eq:asymptotic_wick}
\end{equation}
where we have used the result of Eq.\eqref{eq:mixed_moment_clifford} with $k=2$.
Consequently, for any normalized real linear combination
\begin{equation}
    \hat Q(\mathbf c)
    =
    \frac{1}{c}\sum_{j=1}^m c_j\hat Q_j,
    \qquad
    c^2=\sum_j c_j^2 ,
\end{equation}
the mixed moments of $C(\hat{Q}(\mathbf{c}))^k$ converge to the pairing rule. Together with the
Carleman condition discussed previously, this proves that
$C(\hat Q(\mathbf c))$ is asymptotically Gaussian in the large-$(D,n)$ limit.

Equation~(\ref{hunheju}) also shows that the leading Gaussian contribution to
the mixed moments originates from the diagonal elements of
$\mathrm{A}^{(s)}$. The contribution of the off-diagonal elements will be
analyzed in the following section, which gives the finite-$D$ corrections.

 \section{Finite-Parameter Deviations from Asymptotic Gaussianity}

\subsection{Residual Terms for $C( \hat{O}_l )$ in the Large-$d$ Regime}

In this subsection, we derive rigorous upper bounds for the residual terms of $\mathbb{E}(C(\hat{O}_l)^k)$ in the large-$d$ regime for every operator $\hat{O}_l$ within the space of Clifford generators, confirming that their ratio to the leading-order contribution vanishes as $d \to \infty$. As noted in Sec.~\ref{ourProofMethod}, this finite-parameter correction also fully captures the departure of the sequence $\mathscr{C}_l$ from the GP limit in the large $d$ regime.  

From Eq.~(\ref{Tematrixform}), the matrix $\mathrm{Te}^n$ admits the following block structure:
\begin{equation}
    \left( \begin{matrix}
	1&		\langle u|\sum\nolimits_{i=0}^{n-1}{\mathrm{Te}^{\prime}}^i\\
	0&		{\mathrm{Te}^{\prime}}^n\\
\end{matrix} \right), 
\end{equation}
where $\mathrm{Te}^{\prime}$ denotes the sub-matrix of $\mathrm{Te}$ with the first row and column removed, and the bra $\langle u|$ corresponds to the first row of $\mathrm{Te}$ excluding its first element 1.

Defining the vector $\left( 1,\langle u|\sum\nolimits_{i=0}^{\alpha}{\mathrm{Te}^{\prime}}^i \right) \equiv \langle {L_k}_{,\alpha}|$, we take the asymptotic limit $\alpha\rightarrow \infty$:
\begin{equation}
\left( 1,\langle u|\sum\nolimits_{i=0}^{\infty}{{\mathrm{Te}^{\prime}}^i} \right) =\langle L_k|, 
\end{equation}
which precisely corresponds to the left vector of the spectral projector associated with the unit eigenvalue of the Te matrix (equivalently, the projector onto the subspace with exactly $k$ fixed points). Consequently, the trace can be expanded as: 
\begin{align}
     \mathrm{tr}\left( \mathrm{A}\times \mathrm{Te}^{n-1} \right) &=\mathrm{tr}\left( \mathrm{A}|e\rangle \langle {L_k}_{,n-2}| \right) +\mathrm{tr}\left( \mathrm{A}^{\prime}\times {\mathrm{Te}^{\prime}}^{n-1} \right) 
\\
&=\mathrm{tr}\left( \langle L_k|\mathrm{A}|e\rangle \right) +\mathrm{tr}\left( \mathrm{A}^{\prime}\times {\mathrm{Te}^{\prime}}^{n-1} \right) -\mathrm{tr}\left( \langle u|\sum\nolimits_{i=n-1}^{\infty}{{\mathrm{Te}^{\prime}}^i\mathrm{A}^{''}|e\rangle} \right),  \label{ATen__}  
\end{align}
where $n\geqslant 2$ and $\mathrm{A}_{\delta,\tau}=\sum_{\pi \in \mathcal{E} _k}{\mathrm{Wg}\left( \delta \pi ^{-1},dD \right) d^{\#\pi}D^{\#\left( \pi \tau ^{-1} \right)}}$ corresponding to normalize operator $\hat{O}_l$ from the space of Clifford generator. Here, $\mathrm{A}^{\prime}$ represents the sub-matrix of $\mathrm{A}$ with the first row and column removed, and $\mathrm{A}^{''}$ is the sub-matrix of $\mathrm{A}$ without the first row. The first term in Eq.~(\ref{ATen__}) captures the dominant contribution from the unit eigenvalue sector of $\mathrm{Te}$, whereas the second and third terms are the residual terms. From Eq.~(\ref{Tedominantofe1}), this leading contribution is given by\begin{equation}\label{S183}\mathrm{tr}\left( \langle L_k|\mathrm{A}|e\rangle \right) =\left( 1+\mathcal{O} \left( \frac{k^2}{dD^2} \right) \right) \frac{1}{\left( dD^2 \right) ^k}\sum_{\pi \in S_k}{\prod_{(t_1,t_2,\dots ,t_m)\in \pi}{\mathrm{tr}\left[ \left( \hat{O}^{(1)}\otimes \mathcal{I} _{D^2} \right) ^m \right]}}.\end{equation}Since $\mathrm{tr}\big( \hat{O}^{(1)} \big) =0$ and $\big( {\hat{O}^{(1)}} \big)^2=I$, the permutation $\pi$ must belong to $\mathcal{E}_k$ (i.e., the subset of the symmetric group $S_k$ containing no odd cycles) for the terms in Eq.~(\ref{S183}) to be non-vanishing. The set $\mathcal{E}_k$ can be partitioned according to the number of disjoint cycles into $\mathcal{E}_{k, k/2}, \mathcal{E}_{k, k/2 - 1}, \dots, \mathcal{E}_{k, 1}$, where $\mathcal{E}_{k, k/2}$ is the set of disjoint transpositions, with a cardinality of $\vert \mathcal{E}_{k, k/2} \vert = (k-1)!!$. It is straightforward to verify that the ratio of the cardinalities between these consecutive subsets is bounded by $\mathcal{O}(k^2)$. For $\pi \in \mathcal{E}_{k, i}$, we have\begin{equation}\frac{1}{\left( dD^2 \right) ^k}\prod_{(t_1,t_2,\dots ,t_m)\in \pi}{\mathrm{tr}\left[ \left( \hat{O}^{(1)}\otimes \mathcal{I}_{D^2} \right) ^m \right]}=\frac{1}{\left( dD^2 \right) ^{k-i}}.
\end{equation}
This implies that the summation over subsets other than $\mathcal{E}_{k, k/2}$ in Eq.~(\ref{S183}) is suppressed by a factor of $\mathcal{O} \left( \frac{k^2}{dD^2} \right)$. Thus, we obtain\begin{equation}\mathrm{tr}\left( \langle L_k|\mathrm{A}|e\rangle \right) =\frac{k!}{\left( k/2 \right) !2^{k/2}}\frac{1}{\left( dD^2 \right) ^{k/2}}\left( 1+\mathcal{O} \left( \frac{k^2}{dD^2} \right) \right),\end{equation}
valid under the condition $\frac{k^2}{dD^2}\ll 1$.

To rigorously bound the remaining two residual terms in Eq.~(\ref{ATen__}), we first evaluate the elements of $\mathrm{Te}^{\prime}$:
\begin{align}
    \left( \mathrm{Te}^{\prime} \right) _{\delta ,\alpha}&\leqslant \sum_{\sigma \in S_k}{\left| \mathrm{Wg}\left( \delta \sigma ^{-1},dD \right) \right|d^{\#\sigma}D^{\#\sigma \alpha ^{-1}}}
\\
&=\sum_{\sigma \in S_k}{\left| \mathrm{Wg}\left( \delta \sigma ^{-1},x \right) \right|\frac{x^{\#\sigma}}{D^{\#\sigma}}D^{\#\left( \sigma \alpha ^{-1} \right)}}
\\
&=\overline{\left( \mathrm{Te}^{\prime} \right) _{\delta ,\alpha}},
\end{align}
where $x \equiv dD$, and $\overline{\mathrm{Te}^{\prime}}$ denotes the element-wise upper-bound matrix of $\mathrm{Te}^{\prime}$. We now demonstrate that $\left( \overline{\mathrm{Te}^{\prime}} \right) ^n$ decays element-wise relative to $\overline{\mathrm{Te}^{\prime}}$ in the large-$d$ regime, governed by a suppression factor of $\mathcal{O} [ ( \frac{k^2D^2}{d} ) ^n ] $. Evaluating the element-wise upper bound for $\left( \overline{\mathrm{Te}^{\prime}} \right)^2$ yields:
\begin{align}
   \left[ \left( \overline{\mathrm{Te}^{\prime}} \right) ^2 \right] _{\delta ,\tau}\,&=\sum_{\sigma \in S_k}{\sum_{\alpha \in S_k\backslash\{e\}}{\sum_{\beta \in S_k}{\left| \mathrm{Wg}\left( \delta \sigma ^{-1},x \right) \right|\frac{x^{\#\sigma}}{D^{\#\sigma}}D^{\#\left( \sigma \alpha ^{-1} \right)}\left| \mathrm{Wg}\left( \alpha \beta ^{-1},x \right) \right|\frac{x^{\#\beta}}{D^{\#\beta}}D^{\#\left( \beta \tau ^{-1} \right)}}}}
\\
&\leqslant \sum_{\sigma \in S_k}{\sum_{\alpha \in S_k\backslash\{e\}}{\sum_{\beta \in S_k}{\left| \mathrm{Wg}\left( \delta \sigma ^{-1},x \right) \right|\frac{x^{\#\sigma}}{D^{\#\sigma}}D^{k+\#\sigma -\#\alpha}\left| \mathrm{Wg}\left( \alpha \beta ^{-1},x \right) \right|\frac{x^{\#\beta}}{D^{\#\beta}}D^{k-\#\beta +\#\tau}}}}
\\
&=\sum_{\sigma \in S_k}{\sum_{\alpha \in S_k\backslash\{e\}}{\sum_{\beta \in S_k}{\left| \mathrm{Wg}\left( \delta \sigma ^{-1},x \right) \right|x^{\#\sigma}D^{k-\#\alpha}\left| \mathrm{Wg}\left( \alpha \beta ^{-1},x \right) \right|\frac{x^{\#\beta}}{D^{2\#\beta}}D^{2k}D^{\#\tau -k}}}}
\\
&=\sum_{\sigma \in S_k}{\sum_{\alpha \in S_k\backslash\{e\}}{\sum_{\beta \in S_k}{\left| \mathrm{Wg}\left( \delta \sigma ^{-1},x \right) \right|x^{\#\sigma}D^{k-\#\alpha}\left| \mathrm{Wg}\left( \alpha \beta ^{-1},x \right) \right|\frac{x^{\#\beta}}{D^{2\#\beta}}D^{2k}\frac{D^{\#\sigma +\#\tau -k}}{D^{\#\sigma}}}}}
\\
&\leqslant \sum_{\sigma \in S_k}{\sum_{\alpha \in S_k\backslash\{e\}}{\sum_{\beta \in S_k}{\left| \mathrm{Wg}\left( \delta \sigma ^{-1},x \right) \right|x^{\#\sigma}D^{k-\#\alpha}\left| \mathrm{Wg}\left( \alpha \beta ^{-1},x \right) \right|\frac{x^{\#\beta}}{D^{2\#\beta}}D^{2k}\frac{D^{\#\left( \sigma \tau ^{-1} \right)}}{D^{\#\sigma}}}}}
\\
&=\sum_{\sigma \in S_k}{\left| \mathrm{Wg}\left( \delta \sigma ^{-1},x \right) \right|x^{\#\sigma}\frac{D^{\#\left( \sigma \tau ^{-1} \right)}}{D^{\#\sigma}}\sum_{\alpha \in S_k\backslash\{e\}}{\sum_{\beta \in S_k}{D^{3k-\#\alpha}\left| \mathrm{Wg}\left( \alpha \beta ^{-1},x \right) \right|\frac{x^{\#\beta}}{D^{2\#\beta}}}}}
\\
&=\left( \overline{\mathrm{Te}^{\prime}} \right) _{\delta ,\tau}\sum_{\alpha \in S_k\backslash\{e\}}{\sum_{\beta \in S_k}{D^{3k-\#\alpha}\left| \mathrm{Wg}\left( \alpha \beta ^{-1},x \right) \right|\frac{x^{\#\beta}}{D^{2\#\beta}}}},
\end{align}
where the second and fifth lines follow directly from the triangle inequalities $\#\sigma \alpha ^{-1}\leqslant k-\left| \#\sigma -\#\alpha \right|$ and $\#\sigma +\#\tau -k\leqslant \#(\sigma \tau ^{-1})$, respectively.

Further evaluating the decay factor gives:
\begin{align}
    &\sum_{\alpha \in S_k\backslash \{e\}}{\sum_{\beta \in S_k}{D^{3k-\#\alpha}\left| \mathrm{Wg}\left( \alpha \beta ^{-1},x \right) \right|\frac{x^{\#\beta}}{D^{2\#\beta}}}}
\\
&\leqslant \sum_{\beta \in S_k}{\sum_{\alpha \in S_k\backslash \{e\}}{D^{k-\#\alpha}\omega\frac{\left| \mathrm{Moeb}\left( \alpha \beta ^{-1} \right) \right|}{x^{k-\#\alpha \beta ^{-1}}}\frac{1}{\left( x/D^2 \right) ^{k-\#\beta}}}}
\\
&=\omega\sum_{\beta \in S_k}{\frac{1}{\left( x/D^2 \right) ^{k-\#\beta}}\sum_{\alpha \beta \in S_k\backslash \{e\}}{D^{k-\#\alpha \beta}\frac{\left| \left. \mathrm{Moeb}\left( \alpha \right) \right| \right.}{x^{k-\#\alpha}}}}
\\
&=\omega\sum_{\beta \in S_k}{\frac{1}{\left( x/D^2 \right) ^{k-\#\beta}}\sum_{\alpha \in S_k\backslash \{\beta ^{-1}\}}{D^{k-\#\alpha \beta}\frac{\left| \left. \mathrm{Moeb}\left( \alpha \right) \right| \right.}{x^{k-\#\alpha}}}}
\\
&\leqslant \omega\sum_{\beta \in S_k}{\frac{1}{\left( x/D^2 \right) ^{k-\#\beta}}\sum_{\alpha \in S_k\backslash \{\beta ^{-1}\}}{D^{k-\left( \#\alpha +\#\beta -k \right)}\frac{\left| \mathrm{Moeb}\left( \alpha \right) \right|}{x^{k-\#\alpha}}}}
\\
&=\omega\sum_{\beta \in S_k}{\frac{1}{\left( x/D^3 \right) ^{k-\#\beta}}\sum_{\alpha \in S_k\backslash \{\beta ^{-1}\}}{\frac{\left| \mathrm{Moeb}\left( \alpha \right) \right|}{\left( x/D \right) ^{k-\#\alpha}}}},\label{S158}
\end{align}
where the second line utilizes the asymptotic expansion of the Weingarten factor from Eq.~(\ref{upperboundofwg}), with $\omega \equiv (1-6k^{7/2}/x^2)^{-1}$. The third line exploits the translation invariance of the summation over $\alpha$.

To bound Eq.~(\ref{S158}), we require two combinatorial inequalities. First, recall that the Catalan numbers $c_n$ are globally bounded by $c_n \leqslant 4^n$ for all $n \geqslant 0$~\cite{Richard2011Combinatorics}. For any permutation $\sigma \in S_k$ containing exactly $\#\sigma =k-l$ cycles with lengths $m_1, m_2, \dots, m_{k-l}$, the sum of the cycle lengths strictly obeys $\sum_{i=1}^{k-l} m_i = k$. Consequently, the absolute value of the M\"obius coefficient associated with $\sigma$ can be bounded by:
\begin{equation}\label{upperboundofMoeb}
    |\mathrm{Moeb}(\sigma)| = \prod_{i=1}^{k-l} c_{m_i-1} \leqslant \prod_{i=1}^{k-l} 4^{m_i-1} = 4^{\sum_{i=1}^{k-l} m_i - (k-l)} = 4^{k - (k-l)} = 4^l.
\end{equation}
Next, the upper bound on the number of permutations in $S_k$ with exactly $k-l$ cycles is given by the unsigned Stirling number of the first kind, denoted as $\begin{bmatrix} k \\ k-l \end{bmatrix}$. Because any permutation with $k-l$ cycles can be decomposed into at least $l$ transpositions, and the ordering of $l$ independent transpositions overcounts the permutations by a factor of $l!$, we recover the classical upper bound~\cite{Richard2011Combinatorics}:
\begin{equation}\label{upperboundofpermutation}
    \begin{bmatrix} k \\ k-l \end{bmatrix} \leqslant \frac{1}{l!} \binom{k}{2}^l \leqslant \frac{1}{l!} \left( \frac{k^2}{2} \right)^l.
\end{equation}
Combining the bounds from Eq.~(\ref{upperboundofMoeb}) and Eq.~(\ref{upperboundofpermutation}) yields: 
\begin{equation}
    \sum_{\beta \in S_k}{\frac{1}{\left( x/D^3 \right) ^{k-\#\beta}}\leqslant \sum_{l=0}^{k-1}{\frac{1}{l!}\left( \frac{k^2D^3/2}{x} \right) ^l}},
\end{equation}
and 
\begin{equation}
    \sum_{\alpha \in S_k}{\frac{\left| \mathrm{Moeb}\left( \alpha \right) \right|}{\left( x/D \right) ^{k-\#\alpha}}}\leqslant \sum_{l=0}^{k-1}{\frac{1}{l!}\left( \frac{2k^2D}{x} \right) ^l}.
\end{equation}
For sufficiently large $x$, these expansions are dominated by their leading terms (i.e., $l=0$, giving a value of $1$), which correspond strictly to $\alpha=e$ and $\beta=e$, respectively. Returning to Eq.~(\ref{S158}), because the summation explicitly excludes the case $\alpha = \beta ^{-1}$, the two parts of the sum cannot simultaneously attain their dominant contributions. Therefore, the leading non-vanishing term scales as:
\begin{equation}
    \omega\sum_{\beta \in S_k}{\frac{1}{\left( x/D^3 \right) ^{k-\#\beta}}\sum_{\alpha \in S_k\backslash \{\beta ^{-1}\}}{\frac{\left| \mathrm{Moeb}\left( \alpha \right) \right|}{\left( x/D \right) ^{k-\#\alpha}}}}=\mathcal{O} \left( \frac{k^2D^2}{d} \right)\label{Sk_esum},
\end{equation}
valid under the condition $k^2D^2/d\ll 1$. This establishes the element-wise decay of $\left( \overline{\mathrm{Te}^{\prime}} \right) ^n$. 

Next, we proceed to derive the upper bound for the second term in Eq.~(\ref{ATen__}):
\begin{align}
    \mathrm{tr}\left( \mathrm{A}^{\prime}\times {\mathrm{Te}^{\prime}}^{n-1} \right) &\leqslant \mathrm{tr}\left( \overline{\mathrm{A}^{\prime}}\times \overline{\mathrm{Te}^{\prime}} \right) 
\\
&=\sum_{\delta \in S_k\backslash \{e\}}{\sum_{\pi \in \mathcal{E} _k}{\sum_{\alpha \in S_k\backslash \{e\}}{\sum_{\beta \in S_k}{\left| \mathrm{Wg}\left( \delta \pi ^{-1},x \right) \right|\frac{x^{\#\pi}}{D^{\#\pi}}D^{\#\left( \pi \alpha ^{-1} \right)}\left| \mathrm{Wg}\left( \alpha \beta ^{-1},x \right) \right|\frac{x^{\#\beta}}{D^{\#\beta}}D^{\#\left( \beta \delta ^{-1} \right)}}}}}
\\
&\leqslant \sum_{\delta \in S_k\backslash \{e\}}{\sum_{\pi \in \mathcal{E} _k}{\sum_{\alpha \in S_k\backslash \{e\}}{\sum_{\beta \in S_k}{\left| \mathrm{Wg}\left( \delta \pi ^{-1},x \right) \right|\frac{x^{\#\pi}}{D^{\#\pi}}D^{k-\#\alpha +\#\pi}\left| \mathrm{Wg}\left( \alpha \beta ^{-1},x \right) \right|\frac{x^{\#\beta}}{D^{\#\beta}}D^{k-\#\beta +\#\delta}}}}}
\\
&=\sum_{\delta \in S_k\backslash \{e\}}{\sum_{\pi \in \mathcal{E} _k}{\sum_{\alpha \in S_k\backslash \{e\}}{\sum_{\beta \in S_k}{\left| \mathrm{Wg}\left( \delta \pi ^{-1},x \right) \right|x^{\#\pi}D^{k-\#\alpha}\left| \mathrm{Wg}\left( \alpha \beta ^{-1},x \right) \right|\frac{x^{\#\beta}}{D^{2\#\beta}}D^{k+\#\delta}}}}}\\
&\leqslant \sum_{\delta \in S_k\backslash \{e\}}{\sum_{\pi \in \mathcal{E} _k}{\sum_{\alpha \in S_k\backslash \{e\}}{\sum_{\beta \in S_k}{\omega ^2\frac{D^{\#\delta}}{D^k}\frac{\left| \mathrm{Moeb}\left( \delta \pi ^{-1} \right) \right|}{x^{k-\#\delta \pi ^{-1}}}\frac{x^{\#\pi}}{x^k}\frac{D^k}{D^{\#\alpha}}\frac{\left| \mathrm{Moeb}\left( \alpha \beta ^{-1} \right) \right|}{x^{k-\#\alpha \beta ^{-1}}}\frac{x^{\#\beta}}{x^k}\frac{D^{2k}}{D^{2\#\beta}}}}}}
\\
&=\sum_{\delta \pi \in S_k\backslash \{e\}}{\frac{D^{\#\delta \pi}}{D^k}\frac{\left| \mathrm{Moeb}\left( \delta \right) \right|}{x^{k-\#\delta}}\sum_{\pi \in \mathcal{E} _k}{\sum_{\alpha \in S_k\backslash \{e\}}{\sum_{\beta \in S_k}{\omega ^2\frac{x^{\#\pi}}{x^k}\frac{D^k}{D^{\#\alpha}}\frac{\left| \mathrm{Moeb}\left( \alpha \beta ^{-1} \right) \right|}{x^{k-\#\alpha \beta ^{-1}}}\frac{x^{\#\beta}}{x^k}\frac{D^{2k}}{D^{2\#\beta}}}}}}
\\
&\leqslant \sum_{\delta \pi \in S_k\backslash \{e\}}{\frac{D^{\#\pi}}{D^k}\frac{\left| \mathrm{Moeb}\left( \delta \right) \right|}{\left( x/D \right) ^{k-\#\delta}}\sum_{\pi \in \mathcal{E} _k}{\sum_{\alpha \in S_k\backslash \{e\}}{\sum_{\beta \in S_k}{\omega ^2\frac{x^{\#\pi}}{x^k}\frac{D^k}{D^{\#\alpha}}\frac{\left| \mathrm{Moeb}\left( \alpha \beta ^{-1} \right) \right|}{x^{k-\#\alpha \beta ^{-1}}}\frac{x^{\#\beta}}{x^k}\frac{D^{2k}}{D^{2\#\beta}}}}}}
\\
&=\sum_{\delta \in S_k\backslash \{\pi ^{-1}\}}{\frac{D^{\#\pi}}{D^k}\frac{\left| \mathrm{Moeb}\left( \delta \right) \right|}{\left( x/D \right) ^{k-\#\delta}}\sum_{\pi \in \mathcal{E} _k}{\sum_{\alpha \in S_k\backslash \{e\}}{\sum_{\beta \in S_k}{\omega ^2\frac{x^{\#\pi}}{x^k}\frac{D^k}{D^{\#\alpha}}\frac{\left| \mathrm{Moeb}\left( \alpha \beta ^{-1} \right) \right|}{x^{k-\#\alpha \beta ^{-1}}}\frac{x^{\#\beta}}{x^k}\frac{D^{2k}}{D^{2\#\beta}}}}}}
\\
&\leqslant \sum_{l=0}^{k-1}{\frac{1}{l!}\left( \frac{2k^2D}{x} \right) ^l\sum_{\pi \in \mathcal{E} _k}{\sum_{\alpha \in S_k\backslash \{e\}}{\sum_{\beta \in S_k}{\omega ^2\frac{\left( xD \right) ^{\#\pi}}{\left( xD \right) ^k}\frac{D^k}{D^{\#\alpha}}\frac{\left| \mathrm{Moeb}\left( \alpha \beta ^{-1} \right) \right|}{x^{k-\#\alpha \beta ^{-1}}}\frac{x^{\#\beta}}{x^k}\frac{D^{2k}}{D^{2\#\beta}}}}}}
\\
&=\omega ^2\left[ 1+\mathcal{O} \left( \frac{2k^2}{d} \right) \right] \sum_{\pi \in \mathcal{E} _k}{\frac{\left( xD \right) ^{\#\pi}}{\left( xD \right) ^k}\sum_{\beta \in S_k}{\sum_{\alpha \in S_k\backslash \{e\}}{\frac{D^k}{D^{\#\alpha}}\frac{\left| \mathrm{Moeb}\left( \alpha \beta ^{-1} \right) \right|}{x^{k-\#\alpha \beta ^{-1}}}}\frac{1}{\left( x/D^2 \right) ^{k-\#\beta}}}}
\\
&=\omega ^2\left[ 1+\mathcal{O} \left( \frac{2k^2}{d} \right) \right] \sum_{\pi \in \mathcal{E} _k}{\frac{\left( xD \right) ^{\#\pi}}{\left( xD \right) ^k}\sum_{\beta \in S_k}{\sum_{\alpha \beta \in S_k\backslash \{e\}}{\frac{D^k}{D^{\#\alpha \beta}}\frac{\left| \mathrm{Moeb}\left( \alpha \right) \right|}{x^{k-\#\alpha}}}\frac{1}{\left( x/D^2 \right) ^{k-\#\beta}}}}
\\
&\leqslant \omega ^2\left[ 1+\mathcal{O} \left( \frac{2k^2}{d} \right) \right] \sum_{\pi \in \mathcal{E} _k}{\frac{\left( xD \right) ^{\#\pi}}{\left( xD \right) ^k}\sum_{\beta \in S_k}{\sum_{\alpha \in S_k\backslash \{\beta ^{-1}\}}{\frac{D^k}{D^{\#\alpha +\#\beta -k}}\frac{\left| \mathrm{Moeb}\left( \alpha \right) \right|}{x^{k-\#\alpha}}}\frac{1}{\left( x/D^2 \right) ^{k-\#\beta}}}}
\\
&=\omega ^2\left[ 1+\mathcal{O} \left( \frac{2k^2}{d} \right) \right] \sum_{\pi \in \mathcal{E} _k}{\frac{\left( xD \right) ^{\#\pi}}{\left( xD \right) ^k}\sum_{\beta \in S_k}{\sum_{\alpha \in S_k\backslash \{\beta ^{-1}\}}{\frac{\left| \mathrm{Moeb}\left( \alpha \right) \right|}{\left( x/D \right) ^{k-\#\alpha}}}\frac{1}{\left( x/D^3 \right) ^{k-\#\beta}}}}.
\end{align}
Notice that the double summation over $\alpha$ and $\beta$ in the final line is identical to the expression bounded in Eq.~(\ref{S158}). Hence, the second term is bounded by:
\begin{align}
    \mathrm{tr}\left( \mathrm{A}^{\prime}\times {\mathrm{Te}^{\prime}}^{n-1} \right) &\leqslant \omega ^2\left[ 1+\mathcal{O} \left( \frac{2k^2}{d} \right) \right] \mathcal{O} \left( \frac{k^2D^2}{d} \right) \sum_{\pi \in \mathcal{E} _k}{\frac{\left( xD \right) ^{\#\pi}}{\left( xD \right) ^k}}
\\
&=\mathcal{O} \left( \frac{k^2D^2}{d} \right) \sum_{\pi \in \mathcal{E} _k}{\frac{\left( xD \right) ^{\#\pi}}{\left( xD \right) ^k}},
\end{align}
valid under the condition $\frac{k^2D^2}{d}\ll 1$.

Next, we establish the upper bound for the third residual term in Eq.~(\ref{ATen__}). The element-wise upper bound for the bra $\langle u|$ is derived as:
\begin{align}
    \langle u|_{\alpha}&\leqslant \omega \sum_{\sigma \in S_k}{\frac{\left| \mathrm{Moeb}\left( e\sigma ^{-1} \right) \right|}{\left( Dd \right) ^{k-\#e\sigma ^{-1}}}}\frac{1}{d^{k-\#\sigma}}\frac{D^{\#\sigma \alpha ^{-1}}}{D^k}
\\
&\leqslant \omega \sum_{\sigma \in S_k}{\frac{\left| \mathrm{Moeb}\left( \sigma \right) \right|}{\left( Dd \right) ^{k-\#\sigma}}}\frac{1}{d^{k-\#\sigma}}\frac{D^{k-\#\sigma +\#\alpha}}{D^k}
\\
&=\omega \sum_{\sigma \in S_k}{\frac{\left| \mathrm{Moeb}\left( \sigma \right) \right|}{\left( d^2 \right) ^{k-\#\sigma}}}\frac{D^{\#\alpha}}{D^k}
\\
&\leqslant \omega \sum_{l=0}^{k-1}{\frac{1}{l!}\left( \frac{2k^2}{d^2} \right) ^l}\frac{D^{\#\alpha}}{D^k}
\\
&=\left[ 1+\mathcal{O} \left( \frac{k^2}{d^2} \right) +\mathcal{O} \left( \frac{k^{7/2}}{D^2d^2} \right) \right] \frac{D^{\#\alpha}}{D^k}.
\end{align}
Consequently, the third term is bounded by:
\begin{align}
    \left( \langle u|\sum\nolimits_{i=n-1}^{\infty}{{\mathrm{Te}^{\prime}}^i} \right) _{\tau}&\leqslant \left( \overline{\langle u|}\sum\nolimits_{i=n-1}^{\infty}{\left( \overline{\mathrm{Te}^{\prime}} \right) ^i} \right) _{\tau}
\\
&=\mathcal{O} \left[ \left( \frac{k^2D^2}{d} \right) ^{n-2} \right] \left( \overline{\langle u|}\overline{\mathrm{Te}^{\prime}} \right) _{\tau}
\\
&=\mathcal{O} \left[ \left( \frac{k^2D^2}{d} \right) ^{n-2} \right] \sum_{\alpha \in S_k\backslash\{e\}}{\frac{D^{\#\alpha}}{D^k}\sum_{\beta \in S_k}{\omega}}\frac{\left| \mathrm{Moeb}\left( \alpha \beta ^{-1} \right) \right|}{x^{k-\#\alpha \beta ^{-1}}}\frac{x^{\#\beta}}{x^k}\frac{D^{\#\beta \tau ^{-1}}}{D^{\#\beta}}
\\
&\leqslant \mathcal{O} \left[ \left( \frac{k^2D^2}{d} \right) ^{n-2} \right] \sum_{\alpha \in S_k\backslash \{e\}}{\sum_{\beta \in S_k}{\frac{D^{\#\alpha}}{D^k}\omega}}\frac{\left| \mathrm{Moeb}\left( \alpha \beta ^{-1} \right) \right|}{x^{k-\#\alpha \beta ^{-1}}}\frac{x^{\#\beta}}{x^k}\frac{D^{k-\#\beta +\#\tau}}{D^{\#\beta}}
\\
&=\mathcal{O} \left[ \left( \frac{k^2D^2}{d} \right) ^{n-2} \right] \overline{\langle u|}_{\tau}\sum_{\alpha \in S_k\backslash \{e\}}{\sum_{\beta \in S_k}{\omega \frac{\left| \mathrm{Moeb}\left( \alpha \beta ^{-1} \right) \right|}{x^{k-\#\alpha \beta ^{-1}}}\frac{x^{\#\beta}}{x^k}\frac{D^{k-\#\beta +\#\alpha}}{D^{\#\beta}}}}
\\
&\leqslant \mathcal{O} \left[ \left( \frac{k^2D^2}{d} \right) ^{n-2} \right] \overline{\langle u|}_{\tau}\sum_{\alpha \in S_k\backslash \{\beta ^{-1}\}}{\sum_{\beta \in S_k}{\omega \frac{\left| \mathrm{Moeb}\left( \alpha \right) \right|}{\left( xD \right) ^{k-\#\alpha}}\frac{1}{\left( x/D^3 \right) ^{k-\#\beta}}}}
\\
&=\mathcal{O} \left[ \left( \frac{k^2D^2}{d} \right) ^{n-1} \right] \overline{\langle u|}_{\tau},
\end{align}
where the final step invokes the bound from Eq.~(\ref{Sk_esum}). Combining this result allows us to bound the full trace: 
\begin{align}
    \mathrm{tr}\left( \langle u|\sum\nolimits_{i=n-1}^{\infty}{{\mathrm{Te}^{\prime}}^i\mathrm{A}^{''}|e\rangle} \right) &\leqslant \mathcal{O} \left[ \left( \frac{k^2D^2}{d} \right) ^{n-1} \right] \overline{\langle u|}\overline{\mathrm{A}^{''}}|e\rangle 
\\
&=\mathcal{O} \left[ \left( \frac{k^2D^2}{d} \right) ^{n-1} \right] \sum_{\delta \in S_k\backslash \{e\}}{\frac{D^{\#\delta}}{D^k}\sum_{\pi \in \mathcal{E} _k}{\omega \frac{\left| \mathrm{Moeb}\left( \delta \pi ^{-1} \right) \right|}{x^{k-\#\delta \pi ^{-1}}}\frac{x^{\#\pi}}{x^k}\frac{D^{\#\pi e^{-1}}}{D^{\#\pi}}}}
\\
&=\mathcal{O} \left[ \left( \frac{k^2D^2}{d} \right) ^{n-1} \right] \sum_{\pi \in \mathcal{E} _k}{\frac{x^{\#\pi}}{x^k}}\sum_{\delta \pi \in S_k\backslash \{e\}}{\omega \frac{D^{\#\delta \pi}}{D^k}}\frac{\left| \mathrm{Moeb}\left( \delta \right) \right|}{x^{k-\#\delta}}
\\
&\leqslant \mathcal{O} \left[ \left( \frac{k^2D^2}{d} \right) ^{n-1} \right] \sum_{\pi \in \mathcal{E} _k}{\frac{\left( xD \right) ^{\#\pi}}{\left( xD \right) ^k}}\sum_{\delta \in S_k\backslash \{\pi ^{-1}\}}{\omega}\frac{\left| \mathrm{Moeb}\left( \delta \right) \right|}{\left( x/D \right) ^{k-\#\delta}}
\\
&\leqslant \mathcal{O} \left[ \left( \frac{k^2D^2}{d} \right) ^{n-1} \right] \sum_{\pi \in \mathcal{E} _k}{\frac{\left( xD \right) ^{\#\pi}}{\left( xD \right) ^k}}\sum_{l=0}^{k-1}{\omega \frac{1}{l!}\left( \frac{2k^2}{d} \right) ^l}
\\
&=\mathcal{O} \left[ \left( \frac{k^2D^2}{d} \right) ^{n-1} \right] \mathcal{O} \left[ 1+\left( \frac{k^2}{d} \right) \right] \sum_{\pi \in \mathcal{E} _k}{\frac{1}{\left( dD^2 \right) ^{k-\#\pi}}},
\end{align}
valid under the condition $\frac{k^2D^2}{d}\ll 1$.

Combining the bounds derived above, yields
\begin{align}
 \mathrm{tr}\left( \mathrm{A}\times \mathrm{Te}^{n-1} \right)  &=\sum_{\pi \in \mathcal{E} _k}{\frac{1}{\left( dD^2 \right) ^{k-\#\pi}}} \left[ 1+\mathcal{O}\left( \frac{k^2}{dD^2} \right) +\mathcal{O}\left( \frac{k^2D^2}{d} \right) +\mathcal{O}\left( \frac{k^2D^2}{d} \right) ^{n-1} \right] 
\\
&=\sum_{\pi \in \mathcal{E} _k}{\frac{1}{\left( dD^2 \right) ^{k-\#\pi}}}\left[1+\mathcal{O}  \left( \frac{k^2D^2}{d} \right)\right] 
\\
&=\frac{k!}{\left( k/2 \right) !2^{k/2}}\frac{1}{\left( dD^2 \right) ^{k/2}} \left[ 1+\mathcal{O}\left( \frac{k^2D^2}{d} \right) +\mathcal{O}\left( \frac{k^2}{dD^2} \right) \right] 
\\
&=\frac{k!}{\left( k/2 \right) !2^{k/2}}\frac{1}{\left( dD^2 \right) ^{k/2}}\left[1+\mathcal{O} \left( \frac{k^2D^2}{d} \right)\right].
\end{align} 
This result rigorously demonstrates that in the regime where $k^2D^2/d\ll 1$, the $k$th-order moment of the distribution $C(\hat{O}_l)$ deviates from that of the pure Gaussian distribution $\mathcal{N}\left( 0,\frac{1}{dD^2} \right) $ by at most a residual correction of $\mathcal{O} \left( \frac{k^2D^2}{d} \right)$.

\subsection{Residual Terms for $C( \hat{O}_l )$ in the Large-$(D,n)$ Regime with $D^2\gg d^n$}

In this subsection, we derive rigorous upper bounds for the residual terms of $\mathbb{E}(C(\hat{O}_l)^k)$ in the large-$D$ regime satisfying $D^2\gg d^n$ for a general traceless operator $\hat{O}_l$, confirming that their ratio to the leading-order contribution vanishes as $(D,n) \to \infty$. As noted in Sec.~\ref{ourProofMethod}, this finite-parameter correction also fully captures the departure of the sequence $\mathscr{C}_l$ from the GP limit in the large $(D,n)$ regime. What we consider is:
\begin{equation}
    \mathbb{E} \left( C\left( \hat{O}_l \right) ^k \right) =\mathrm{tr}\left( \mathrm{A}\times \mathrm{Te}^{n-1} \right) ,
\end{equation}
where $\mathrm{A}_{_{\delta ,\tau}}^{\left( s \right)}=\sum_{\pi \in S_k}{\mathrm{Wg}\left( \delta \pi ^{-1},dD \right) \mathrm{tr}\left( \mathrm{P}_{\pi}O_s \right) D^{\#\left( \pi \tau ^{-1} \right)}}$ corresponds to a general traceless operator.
Applying the upper bound of the Weingarten factor from Eq.~(\ref{upperboundofwg}) yields:
\begin{align}
    \mathrm{Te}_{\delta ,\alpha}&=\sum_{\sigma \in S_k}{\mathrm{Wg}\left( \delta \sigma ^{-1},dD \right) d^{\#\sigma}D^{\#(\sigma \alpha ^{-1})}}
\\
&\leqslant \sum_{\sigma \in S_k}{\omega \frac{\left| \mathrm{Moeb}\left( \delta \sigma ^{-1} \right) \right|}{\left( dD \right) ^{k-\#\delta \sigma ^{-1}}}\frac{1}{d^{k-\#\sigma}}\frac{1}{D^{k-\#\sigma \alpha ^{-1}}}},
\end{align}
where $\omega =\frac{1}{1-6k^{7/2}/\left( D^2d^2 \right)}$, and $\delta ,\alpha \in S_k$. Furthermore, the local dimension $d$ obeys the inequality:
\begin{equation}
\frac{d^{\#(\delta \sigma ^{-1})+\#\sigma}}{d^{2k}}\leqslant \frac{d^{k-\#\sigma +\#\delta +\#\sigma}}{d^{2k}}=\frac{d^{\#\delta}}{d^k}.
\end{equation}
Consequently,
\begin{equation}
    \mathrm{Te}_{\delta ,\alpha}\leqslant \frac{d^{\#\delta}}{d^k}\sum_{\sigma \in S_k}{\omega \frac{\left| \mathrm{Moeb}\left( \delta \sigma ^{-1} \right) \right|}{D^{k-\#\delta \sigma ^{-1}}}\frac{1}{D^{k-\#\sigma \alpha ^{-1}}}}.
\end{equation}
We define the matrix $B$ as:
\begin{equation}
    B_{\delta ,\alpha}=\sum_{\sigma \in S_k}{\omega \frac{\left| \mathrm{Moeb}\left( \delta \sigma ^{-1} \right) \right|}{D^{k-\#\delta \sigma ^{-1}}}\frac{1}{D^{k-\#\sigma \alpha ^{-1}}}},
\end{equation}
so that the trace expansion becomes:
\begin{align}
\mathrm{tr}\left( \mathrm{A}\times \mathrm{Te}^{n-1} \right) &\leqslant\mathrm{tr}\left( \mathrm{Te}^{n/2-1}\times \mathrm{A}\times \mathrm{Te}^{n/2} \right) 
\\
&=\mathrm{tr}\left( \left( XB \right) ^{n/2-1}\times \mathrm{A}\times \left( XB \right) ^{n/2} \right),   \label{ATe_nlocalD_2d_n}
\end{align}
where $X=\mathrm{diag}\left( \frac{d^{\#\delta}}{d^k} \right) $. Here we take $n$ to be even; the case when $n$ is odd follows analogously and does not affect the result. We observe that:
\begin{align}
    X B&=\frac{d^{\#\delta}}{d^k}\sum_{\sigma \in S_k}{\omega \frac{\left| \mathrm{Moeb}\left( \delta \sigma ^{-1} \right) \right|}{D^{k-\#\delta \sigma ^{-1}}}\frac{1}{D^{k-\#\sigma \alpha ^{-1}}}}
\\
&=\sum_{\sigma \in S_k}{\omega \left| \mathrm{Moeb}\left( \delta \sigma ^{-1} \right) \right|\frac{d^{\#\delta -\#\sigma}}{D^{k-\#\delta \sigma ^{-1}}}\frac{d^{\#\sigma -\#\alpha}}{D^{k-\#\sigma \alpha ^{-1}}}}\frac{d^{\#\alpha}}{d^k}
\\
&\leqslant \sum_{\sigma \in S_k}{\omega \left| \mathrm{Moeb}\left( \delta \sigma ^{-1} \right) \right|\frac{d^{k-\#\delta \sigma ^{-1}}}{D^{k-\#\delta \sigma ^{-1}}}\frac{d^{k-\#\sigma \alpha ^{-1}}}{D^{k-\#\sigma \alpha ^{-1}}}}\frac{d^{\#\alpha}}{d^k}
\\
&=\sum_{\sigma \in S_k}{\omega \left| \mathrm{Moeb}\left( \delta \sigma ^{-1} \right) \right|\frac{1}{\left( D/d \right) ^{k-\#\delta \sigma ^{-1}}}\frac{1}{\left( D/d \right) ^{k-\#\sigma \alpha ^{-1}}}\frac{d^{\#\alpha}}{d^k}}
\\
&=B^{\left( 1 \right)}X,
\end{align}
where the third line invokes the triangle inequality $\left. \left| \#\delta -\#\sigma \right. \right|\leqslant k-\#\delta \sigma ^{-1}$. We further define:
\begin{equation}
    B^{\left( i \right)}=\sum_{\sigma \in S_k}{\omega \left| \mathrm{Moeb}\left( \delta \sigma ^{-1} \right) \right|\frac{1}{\left( D/d^i \right) ^{k-\#\delta \sigma ^{-1}}}\frac{1}{\left( D/d^i \right) ^{k-\#\sigma \alpha ^{-1}}}}.
\end{equation}

Similarly, evaluating the product $BX$ yields:
\begin{align}
    BX&=\sum_{\sigma \in S_k}{\omega \frac{\left| \mathrm{Moeb}\left( \delta \sigma ^{-1} \right) \right|}{D^{k-\#\delta \sigma ^{-1}}}\frac{1}{D^{k-\#\sigma \alpha ^{-1}}}}\frac{d^{\#\alpha}}{d^k}
\\
&=\frac{d^{\#\delta}}{d^k}\sum_{\sigma \in S_k}{\omega \left| \mathrm{Moeb}\left( \delta \sigma ^{-1} \right) \right|\frac{d^{\#\sigma -\#\delta}}{D^{k-\#\delta \sigma ^{-1}}}\frac{d^{\#\alpha -\#\sigma}}{D^{k-\#\sigma \alpha ^{-1}}}}
\\
&\leqslant \frac{d^{\#\delta}}{d^k}\sum_{\sigma \in S_k}{\omega \left| \mathrm{Moeb}\left( \delta \sigma ^{-1} \right) \right|\frac{d^{k-\#\delta \sigma ^{-1}}}{D^{k-\#\delta \sigma ^{-1}}}\frac{d^{k-\#\sigma \alpha ^{-1}}}{D^{k-\#\sigma \alpha ^{-1}}}}
\\
&=X B^{\left( 1 \right)}
\end{align}

The column sums of $B^{\left( i \right)}$ can be strictly upper-bounded as follows:
\begin{align}
    \sum_{\delta \in S_k}{{B^{\left( i \right)}}_{\delta ,\alpha}}&=\sum_{\delta \in S_k}{\sum_{\sigma \in S_k}{\omega \left| \mathrm{Moeb}\left( \delta \sigma ^{-1} \right) \right|\frac{1}{\left( D/d^i \right) ^{k-\#\delta \sigma ^{-1}}}\frac{1}{\left( D/d^i \right) ^{k-\#\sigma \alpha ^{-1}}}}}
\\
&=\sum_{\sigma \in S_k}{\sum_{\delta \sigma \in S_k}{\omega \left| \mathrm{Moeb}\left( \delta \right) \right|\frac{1}{\left( D/d^i \right) ^{k-\#\delta}}\frac{1}{\left( D/d^i \right) ^{k-\#\sigma \alpha ^{-1}}}}}
\\
&=\sum_{\sigma \in S_k}{\sum_{\delta \in S_k}{\omega \left| \mathrm{Moeb}\left( \delta \right) \right|\frac{1}{\left( D/d^i \right) ^{k-\#\delta}}\frac{1}{\left( D/d^i \right) ^{k-\#\sigma \alpha ^{-1}}}}}
\\
&\leqslant \omega \sum_{\sigma \in S_k}{\sum_{l=0}^{k-1}{\frac{1}{l!}}}\left( \frac{2k^2d^i}{D} \right) ^l\frac{1}{\left( D/d^i \right) ^{k-\#\sigma \alpha ^{-1}}}
\\
&=\omega \left[ 1+\mathcal{O} \left( \frac{2k^2d^i}{D} \right) \right] \sum_{\sigma \in S_k}{\frac{1}{\left( D/d^i \right) ^{k-\#\sigma \alpha ^{-1}}}}
\\
&\leqslant \omega \left[ 1+\mathcal{O} \left( \frac{2k^2d^i}{D} \right) \right] \sum_{l=0}^{k-1}{\frac{1}{l!}}\left( \frac{k^2d^i/2}{D} \right) ^l
\\
&= \left[ 1+\mathcal{O} \left( \frac{5k^2d^i/2}{D} \right) \right],
\end{align}
where we have applied the combinatorial bounds from Eq.~(\ref{upperboundofMoeb}) and Eq.~(\ref{upperboundofpermutation}). Returning to Eq.~(\ref{ATe_nlocalD_2d_n}), we find:
\begin{align}
    \mathrm{tr}\left( \mathrm{A}\times \mathrm{Te}^{n-1} \right) &\leqslant\mathrm{tr}\left( \left( XB \right) ^{n/2-1}\times \mathrm{A}\times \left( XB \right) ^{n/2} \right) 
\\
&\leqslant \mathrm{tr}\left( B^{\left( 1 \right)}B^{\left( 2 \right)}\cdots B^{\left( n/2-1 \right)}\times {X}^{n/2-1}\times \mathrm{A}\times {X}^{n/2}\times B^{\left( n/2-1 \right)}\cdots B^{\left( 1 \right)}B \right) 
\\
&=\mathrm{tr}\left( {X}^{n/2-1}\mathrm{A}{X}^{n/2}\times B^{\left( n/2-1 \right)}\cdots B^{\left( 1 \right)}B\times B^{\left( 1 \right)}B^{\left( 2 \right)}\cdots B^{\left( n/2-1 \right)} \right) 
\\
&\leqslant \sum_{\delta \in S_k}{\max_{\alpha \in S_k} \left\{ \left( {X}^{n/2-1}\mathrm{A}{X}^{n/2} \right) _{\delta ,\alpha} \right\}}
\\
&\times \prod_{i=0}^{n/2-1}{\left[ 1+\mathcal{O} \left( \frac{5k^2d^i/2}{D} \right) \right] \prod_{i=1}^{n/2-1}{\left[ 1+\mathcal{O} \left( \frac{5k^2d^i/2}{D} \right) \right]}}
\\
&=\sum_{\delta \in S_k}{\max_{\alpha \in S_k} \left\{ \left( {X}^{n/2-1}\mathrm{A}{X}^{n/2} \right) _{\delta ,\alpha} \right\}}\times \left[ 1+\mathcal{O} \left( \frac{k^2d^{n/2}}{D} \right) \right], 
\end{align}
where the fourth line exploits the submultiplicativity of the induced matrix 1-norm, $\|AB\|_1 \le \|A\|_1 \|B\|_1$, with $\parallel A\parallel _1=\underset{j}{\max}\sum_i{\left| \left. A_{ij} \right| \right.}$. This bound is valid under the condition $k^2d^{n/2}/D\ll 1$.

Further evaluation yields:
\begin{align}
    &\sum_{\delta \in S_k}{\max_{\alpha \in S_k} \left\{ \left( X^{n/2-1}\mathrm{A}X^{n/2} \right) _{\delta ,\alpha} \right\}}\times \left[ 1+\mathcal{O} \left( \frac{k^2d^{n/2}}{D} \right) \right] 
\\
&=\sum_{\delta \in S_k}{\max_{\alpha \in S_k} \left\{ \frac{d^{\left( \frac{1}{2}n-1 \right) \#\delta}}{d^{\left( \frac{1}{2}n-1 \right) k}}\sum_{\pi \in S_k}{\mathrm{Wg}\left( \delta \pi ^{-1},dD \right){\mathrm{tr}\left( \mathrm{P}_{\pi}O \right)}D^{\#(\pi \sigma ^{-1})}}\frac{d^{\frac{1}{2}n\#\alpha}}{d^{\frac{1}{2}nk}} \right\}}\left[ 1+\mathcal{O} \left( \frac{k^2d^{n/2}}{D} \right) \right]  
\\
&\leqslant \sum_{\delta \in S_k}{\max_{\alpha \in S_k} \left\{ \frac{d^{\left( \frac{1}{2}n-1 \right) \#\delta}}{d^{\left( \frac{1}{2}n-1 \right) k}}\sum_{\pi \in S_k}{\omega \frac{\left| \mathrm{Moeb}\left( \delta \pi ^{-1} \right) \right|}{\left( Dd \right) ^{k-\#\delta \pi ^{-1}}}\frac{{\mathrm{tr}\left( \mathrm{P}_{\pi}O \right)}}{d^k}\frac{1}{D^{k-\#\pi \alpha ^{-1}}}}\frac{d^{\frac{1}{2}n\#\alpha}}{d^{\frac{1}{2}nk}} \right\}}\left[ 1+\mathcal{O} \left( \frac{k^2d^{n/2}}{D} \right) \right] 
\\
&=\sum_{\delta \in S_k}{\max_{\alpha \in S_k} \left\{ \frac{d^{\left( \frac{1}{2}n-1 \right) \#\delta}}{d^{\left( \frac{1}{2}n-1 \right) k}}\sum_{\pi \in S_k}{\omega \frac{d^{\frac{1}{2}n\#\pi}}{d^{\frac{1}{2}nk}}\frac{{\mathrm{tr}\left( \mathrm{P}_{\pi}O \right)}}{d^k}\frac{\left| \mathrm{Moeb}\left( \delta \pi ^{-1} \right) \right|}{\left( Dd \right) ^{k-\#\delta \pi ^{-1}}}\frac{d^{\frac{1}{2}n\left( \#\alpha -\#\pi \right)}}{D^{k-\#\pi \alpha ^{-1}}}} \right\}}\left[ 1+\mathcal{O} \left( \frac{k^2d^{n/2}}{D} \right) \right] 
\\
&\leqslant \sum_{\delta \in S_k}{\max_{\alpha \in S_k} \left\{ \frac{d^{\left( \frac{1}{2}n-1 \right) \#\delta}}{d^{\left( \frac{1}{2}n-1 \right) k}}\sum_{\pi \in S_k}{\omega \frac{d^{\frac{1}{2}n\#\pi}}{d^{\frac{1}{2}nk}}\frac{{\mathrm{tr}\left( \mathrm{P}_{\pi}O \right)}}{d^k}\frac{\left| \mathrm{Moeb}\left( \delta \pi ^{-1} \right) \right|}{\left( Dd \right) ^{k-\#\delta \pi ^{-1}}}\frac{1}{\left( D/d^{n/2} \right) ^{k-\#\pi \alpha ^{-1}}}} \right\}}\left[ 1+\mathcal{O} \left( \frac{k^2d^{n/2}}{D} \right) \right] 
\\
&=\sum_{\delta \in S_k}{\frac{d^{\left( \frac{1}{2}n-1 \right) \#\delta}}{d^{\left( \frac{1}{2}n-1 \right) k}}\sum_{\pi \in S_k}{\omega \frac{d^{\frac{1}{2}n\#\pi}}{d^{\frac{1}{2}nk}}\frac{{\mathrm{tr}\left( \mathrm{P}_{\pi}O \right)}}{d^k}\frac{\left| \mathrm{Moeb}\left( \delta \pi ^{-1} \right) \right|}{\left( Dd \right) ^{k-\#\delta \pi ^{-1}}}}}\left[ 1+\mathcal{O} \left( \frac{3k^2d^{n/2}/2}{D} \right) \right] 
\\
&=\omega \left[ 1+\mathcal{O} \left( \frac{3k^2d^{n/2}/2}{D} \right) \right] \sum_{\pi \in S_k}{\frac{d^{\frac{1}{2}n\#\pi}}{d^{\frac{1}{2}nk}}\frac{{\mathrm{tr}\left( \mathrm{P}_{\pi}O \right)}}{d^k}\sum_{\delta \in S_k}{\frac{d^{\left( \frac{1}{2}n-1 \right) \#\delta}}{d^{\left( \frac{1}{2}n-1 \right) k}}\frac{\left| \mathrm{Moeb}\left( \delta \pi ^{-1} \right) \right|}{\left( Dd \right) ^{k-\#\delta \pi ^{-1}}}}}
\\
&=\omega \left[ 1+\mathcal{O} \left( \frac{3k^2d^{n/2}/2}{D} \right) \right] \sum_{\pi \in S_k}{\frac{d^{\frac{1}{2}n\#\pi}}{d^{\frac{1}{2}nk}}\frac{{\mathrm{tr}\left( \mathrm{P}_{\pi}O \right)}}{d^k}\sum_{\delta \pi \in S_k}{\frac{d^{\left( \frac{1}{2}n-1 \right) \#\delta \pi}}{d^{\left( \frac{1}{2}n-1 \right) k}}\frac{\left| \mathrm{Moeb}\left( \delta \right) \right|}{\left( Dd \right) ^{k-\#\delta}}}}
\\
&\leqslant \omega \left[ 1+\mathcal{O} \left( \frac{3k^2d^{n/2}/2}{D} \right) \right] \sum_{\pi \in S_k}{\frac{d^{\frac{1}{2}n\#\pi}}{d^{\frac{1}{2}nk}}\frac{{\mathrm{tr}\left( \mathrm{P}_{\pi}O \right)}}{d^k}\frac{d^{\left( \frac{1}{2}n-1 \right) \#\pi}}{d^{\left( \frac{1}{2}n-1 \right) k}}\sum_{\delta \pi \in S_k}{\frac{\left| \mathrm{Moeb}\left( \delta \right) \right|}{\left( Dd/d^{n/2-1} \right) ^{k-\#\delta}}}}
\\
&=\omega \left[ 1+\mathcal{O} \left( \frac{3k^2d^{n/2}/2}{D} \right) \right] \sum_{\pi \in S_k}{\frac{d^{\left( n-1 \right) \#\pi}}{d^{\left( n-1 \right) k}}\frac{{\mathrm{tr}\left( \mathrm{P}_{\pi}O \right)}}{d^k}\sum_{\delta \in S_k}{\frac{\left| \mathrm{Moeb}\left( \delta \right) \right|}{\left( Dd/d^{n/2-1} \right) ^{k-\#\delta}}}}
\\
&\leqslant \omega \left[ 1+\mathcal{O} \left( \frac{3k^2d^{n/2}/2}{D} \right) \right] \sum_{\pi \in S_k}{\frac{d^{\left( n-1 \right) \#\pi}}{d^{\left( n-1 \right) k}}\frac{{\mathrm{tr}\left( \mathrm{P}_{\pi}O \right)}}{d^k}\sum_{l=0}^{k-1}{\frac{1}{l!}\left( \frac{2k^2d^{n/2-2}}{D} \right) ^l}}
\\
&=\left[ 1+\mathcal{O} \left( \frac{k^2d^{n/2}}{D} \right) \right] \sum_{\pi \in S_k}{\frac{d^{\left( n-1 \right) \#\pi}}{d^{\left( n-1 \right) k}}\frac{\mathrm{tr}\left( \mathrm{P}_{\pi}O \right)}{d^k}},\label{S286}
\end{align}
valid under the condition $k^2d^{n/2}/D\ll 1$.
Since $O = \hat{O}^{(1)} \otimes \cdots \otimes \hat{O}^{(1)}$ with $\operatorname{tr}(\hat{O}^{(1)}) = 0$, $\operatorname{tr}(\mathrm{P}_\pi O)$ is non-vanishing only if $\pi$ has no fixed points. We denote the set  by $\mathcal{D}_k = \{\sigma \in S_k \mid \sigma(i) \neq i, \, \forall i\}$. The set $\mathcal{D}_k$ can be partitioned according to the number of disjoint cycles into $\mathcal{D}_{k, k/2}, \mathcal{D}_{k, k/2 - 1}, \dots, \mathcal{D}_{k, 1}$, where $\vert{}\mathcal{D}_{k, k/2}\vert{} = \frac{k!}{(k/2)!\, 2^{k/2}}$.
The cardinality $|\mathcal{D}_{k,c}|$ corresponds to the 2-associated Stirling numbers of the first kind. By established combinatorial properties~\cite{Richard2011Combinatorics}, the sequence $\{|\mathcal{D}_{k,c}|\}$ is strictly log-concave with respect to $c$. This ensures that the growth ratio between adjacent subsets, $|\mathcal{D}_{k, c-1}|/|\mathcal{D}_{k, c}|$, increases monotonically as $c$ decreases. Therefore, the maximum growth ratio across all cycle reduction steps occurs at the initial transition ($c = k/2 \to k/2 - 1$), i.e. $R_k=|\mathcal{D} _{k,k/2-1}|/|\mathcal{D} _{k,k/2}|
$

To explicitly evaluate $R_k$, we compute both cardinalities using the conjugacy class size formula. The subset $\mathcal{D}_{k, k/2-1}$ comprises permutations with exactly $k/2-1$ cycles of length $\ge 2$. This constraint permits only two cycle structures: one $4$-cycle with $(k/2-2)$ $2$-cycles, or two $3$-cycles with $(k/2-3)$ $2$-cycles. Summing their respective conjugacy class sizes gives
\begin{equation}
    |\mathcal{D}_{k, k/2-1}| = \frac{k!}{4 \cdot 2^{k/2-2} (k/2-2)!} + \frac{k!}{18 \cdot 2^{k/2-3} (k/2-3)!}.
\end{equation}

Taking the ratio of these expressions, we obtain
\begin{equation}
    R_k = \frac{k(k-2)}{4} + \frac{k(k-2)(k-4)}{18} = \frac{k(k-2)(2k+1)}{36}.
\end{equation}
Consequently, we establish that 
the ratio between consecutive terms in the sequence $\vert{}\mathcal{D}_{k, k/2}\vert{}, \vert{}\mathcal{D}_{k, k/2 - 1}\vert{}, \dots, \vert{}\mathcal{D}_{k, 1}\vert{}$ does not exceed $\mathcal{O} \left( k^3 \right) $. 
This implies that the summation over subsets other than $\mathcal{D}_{k, k/2}$ in Eq.~(\ref{S286}) is suppressed by a factor of $\mathcal{O} \left( \frac{k^3}{d^n} \right)$. Thus, we obtain

\begin{align}
   \mathrm{tr}\left( \mathrm{A}\times \mathrm{Te}^{n-1} \right) &= \left[ 1+\mathcal{O} \left( \frac{k^2d^{n/2}}{D} \right) \right] \sum_{\pi \in S_k}{\frac{d^{(n-1)\#\pi}}{d^{(n-1)k}}\frac{\mathrm{tr}\left( \mathrm{P}_{\pi}O \right)}{d^k}}
\\
&=\frac{k!}{\left( k/2 \right) !2^{k/2}}\frac{1}{d^{kn/2}}\frac{\mathrm{tr}( \hat{O}^{( 1) 2}) ^{k/2}}{d^{k/2}}\left[ 1+\mathcal{O} \left( \frac{k^2d^{n/2}}{D} \right) \right] \left[ 1+\mathcal{O} \left( \frac{k^3}{d^n} \right) \right]  
\\
&=\frac{k!}{\left( k/2 \right) !2^{k/2}}\frac{1}{d^{kn/2}}\frac{\mathrm{tr}( \hat{O}^{( 1) 2}) ^{k/2}}{d^{k/2}}\left[ 1+\mathcal{O} \left( \frac{k^2d^{n/2}}{D} \right) +\mathcal{O} \left( \frac{k^3}{d^n} \right) \right]  
\end{align}

This derivation reveals that the dominant contribution to $\mathrm{tr}( \mathrm{A}\times \mathrm{Te}^{n-1} )$ arises when $\pi$ consists entirely of disjoint transpositions in $S_k$. Consequently, in the asymptotic regime where $\frac{k^2d^{n/2}}{D}\ll1$ and $\frac{k^3}{d^n}\ll1$, the $k$th-order moment of the distribution $C(\hat{O})$ deviates from that of the pure Gaussian distribution $\mathcal{N}\left( 0,\frac{1}{d^{n}} \right) $ by at most a residual correction of $\mathcal{O} \left( \frac{k^2d^{n/2}}{D} \right) +\mathcal{O} \left( \frac{k^3}{d^n} \right)$.

Similarly, if we consider operators in the space of Clifford generators, we can obtain a more compact finite correction: in the asymptotic regime where $\frac{k^2d^{n/2}}{D}\ll1$ and $\frac{k^2}{d^n}\ll 1$, the $k$th-order moment of the distribution $C(\hat{O})$ deviates from that of the pure Gaussian distribution $\mathcal{N}\left( 0,\frac{1}{d^{n}} \right) $ by at most a residual correction of $\mathcal{O} \left( \frac{k^2d^{n/2}}{D} \right) +\mathcal{O} \left( \frac{k^2}{d^n} \right)$. This is consistent with the results presented in our main text.

\subsection{Residual Terms for $C( \hat{O}_l )$ in the Large-$(D,n)$ Regime with $D^2\ll d^n$}
In this subsection, we derive rigorous upper bounds for the residual terms of $\mathbb{E}(C(\hat{O}_l)^k)$ in the large-$D$ regime satisfying $D^2\ll d^n$ for general traceless operator $\hat{O}_l$, confirming that their ratio to the leading-order contribution vanishes as $(D,n) \to \infty$. 
 

 We next calculate the remainder terms arising from the other eigenvalues.
\begin{equation}
    \left. \mathrm{tr}\!\left( \mathrm{ATe}^n \right) \right|_{v_{f}^{\left( i \right)}}={\lambda _{f}^{( i )}}^n\frac{\langle v_{f}^{\left( i \right)}|G_{dD}X^{-1}X_OG_D|v_{f}^{\left( i \right)}\rangle}{\langle v_{f}^{\left( i \right)}|G_{dD}X^{-1}G_{dD}|v_{f}^{\left( i \right)}\rangle}.\label{S292}
\end{equation}

First, consider $\langle v|G_{dD}X^{-1}G_{dD}|v\rangle $, where $|v\rangle $ is an arbitrary normalized vector that satisfies $\langle v|v\rangle=1 $.
It is noted that $G_{dD}=( Dd ) ^{\#\delta \sigma ^{-1}}$ is a Gram matrix and approaches a diagonal matrix under large-$D$ conditions, then:
\begin{align}
    \langle v|G_{dD}X^{-1}G_{dD}|v\rangle &=\left( Dd \right) ^{2k}\langle v|\left( I+K \right) X^{-1}\left( I+K \right) |v\rangle 
\\
&=\left( Dd \right) ^{2k}\langle v|X^{-1/2}X^{1/2}\left( I+K \right) X^{-1}\left( I+K \right) X^{1/2}X^{-1/2}|v\rangle 
\\
&=\left( Dd \right) ^{2k}\langle w|X^{1/2}\left( I+K \right) X^{-1/2}X^{-1/2}\left( I+K \right) X^{1/2}|w\rangle 
\\
&=\left( Dd \right) ^{2k}\left\| \left. \left( I+X^{1/2}KX^{-1/2} \right) |w\rangle \right\| _{2}^{2} \right.\label{vGXGv} 
\end{align}
where $K=\frac{1}{(Dd)^{k-\#\delta \sigma ^{-1}}}$, $\delta \ne \sigma $, is symmetric, and $|w\rangle =X^{-1/2}|v\rangle $.
We denote $X^{1/2}KX^{-1/2}=M$, then:
\begin{align}
   \left\|  M|w\rangle \right\| _{2}^{2}  &=\sum_{\alpha}{\left| \left( M|w\rangle \right) _{\alpha} \right|^2}
\\
&=\sum_{\alpha}{\left|  \sum_{\beta}{M_{\alpha ,\beta}}w_{\beta} \right|  ^2}
\\
&\leqslant \sum_{\alpha}{\left( \sum_{\beta}{\left. \left| M_{\alpha ,\beta} \right. \right|} \right)}\left( \sum_{\beta}{\left| \left. M_{\alpha ,\beta} \right| \right.}\left| \left. |w_{\beta} \right|^2 \right. \right) 
\\
&=\sum_{\alpha}{R}\sum_{\beta}{\left| \left. M_{\alpha ,\beta} \right| \right.}\left| \left. |w_{\beta} \right|^2 \right. 
\\
&=R\sum_{\beta}{\sum_{\alpha}{\left| \left. M_{\alpha ,\beta} \right| \right.}\left| \left. |w_{\beta} \right|^2 \right.}
\\
&=RC\langle w|w\rangle .
\end{align}

where $R$ denotes the upper bound of the row sum, $C$ denotes the upper bound of the column sum respectively. Here, in the third line, we use the Cauchy-Schwarz inequality.
Consider the upper bound of the row sum of the $M$ matrix:
\begin{align}
    \sum_{\sigma \in S_k}{M_{\delta ,\sigma}}&=\sum_{\sigma \in S_k\backslash\{\delta \}}{d^{\left( \#\delta -\#\sigma \right) /2}\frac{1}{D^{k-\#\delta \sigma ^{-1}}}}
\\
&\leqslant \sum_{\sigma \in S_k\backslash\{\delta \}}{\frac{\left( d^{1/2} \right) ^{k-\#\delta \sigma ^{-1}}}{D^{k-\#\delta \sigma ^{-1}}}}
\\
&\leqslant \sum_{l=1}^{k-1}{\frac{1}{l!}}\left( \frac{d^{1/2}}{D} \right) ^l
\\
&=\mathcal{O} \left( \frac{k^2d^{1/2}/2}{D} \right).
\end{align}
Similarly, we can get that the row sum upper bound is $\mathcal{O} \left( \frac{k^2d^{1/2}/2}{D} \right) $.

By the triangle inequality,
\begin{equation}
    \left\| \left. \left( I+M \right) |w\rangle \right\| _2 \right. \leqslant \left\| \left. w \right\| _2+ \right. \left\| \left. Mw \right\| _2 \right. \leqslant \left( 1+\sqrt{RC} \right) \left\| \left. w \right\|_2 \right.,
\end{equation}
Similarly, the lower bound is:
\begin{equation}
   \left\| \left( I+M \right) |w\rangle \right\| _2\geqslant \left\| w \right\| _2-\left\| Mw \right\| _2\geqslant \left( 1-\sqrt{RC} \right) \left\| w \right\| _2.
\end{equation}
Substituting the above result back into Eq.~(\ref{vGXGv}), we obtain:
\begin{equation}
    \langle v|G_{dD}X^{-1}G_{dD}|v\rangle =\left( Dd \right) ^{2k}\langle v|X^{-1}|v\rangle \left[ 1+\mathcal{O} \left( \frac{k^2d^{1/2}/2}{D} \right) \right].
\end{equation}
Note that $d^k\langle v|X^{-1}|v\rangle \geqslant \langle v|v\rangle $, then :
\begin{equation}
    \langle v|G_{dD}X^{-1}G_{dD}|v\rangle \geqslant \left( D^2d \right) ^k\left[ 1+\mathcal{O} \left( \frac{k^2d^{1/2}}{D} \right) \right] 
\end{equation}

Next, we calculate:
\begin{equation}
    {\lambda _{f}^{\left( i \right)}}^n\langle v_{f}^{\left( i \right)}|G_{dD}X^{-1}X_OG_D |v_{f}^{\left( i \right)}\rangle .\label{S267}
\end{equation}
where $|v_{f}^{\left( i \right)}\rangle $ is the eigenvector corresponding to the diagonal block with a fixed point number $f$ and  $\langle v_{f}^{\left( i \right)}|v_{f}^{\left( i \right)}\rangle =1$
Consider the diagonal block corresponding to the fixed point number $m$, $m=k-2,k-3,\cdots,0$, so that $|v_{m}^{\left( i \right)}\rangle $ satisfies $( v_{m}^{( i )} ) _{\delta}=0$, $N_f\left( \delta \right) <m$. Donate the set $\mathcal{F} _m$ as $\left\{ \delta \left| \delta \in S_k,N_f\left( \delta \right) \geqslant m \right. \right\} $.

Then
\begin{equation}\label{S268}
    \langle v_{m}^{\left( i \right)}|G_{dD}X^{-1}X_OG_D|v_{m}^{\left( i \right)}\rangle =\left( D^2d \right) ^k\sum_{\delta ,\alpha \in \mathcal{F} _m}{\sum_{\pi \in S_k}{v_{\delta}v_{\alpha}\frac{1}{\left( Dd \right) ^{k-\delta \pi ^{-1}}}}}\frac{1}{D^{k-\pi \alpha ^{-1}}}\frac{\mathrm{tr}\left( \mathrm{P}_{\pi}O \right)}{d^{\#\pi}},
\end{equation}
here we will abbreviate $( v_{m}^{( i )} ) _{\delta}$ as $v_\delta$, and this notation will be used throughout thereafter.

Bounding Eq.~(\ref{S268}) yields:
\begin{align}
 \sum_{\delta ,\alpha \in \mathcal{F} _m}{\sum_{\pi \in S_k}{v_{\delta}v_{\alpha}\frac{1}{\left( Dd \right) ^{k-\#\delta \pi ^{-1}}}\frac{1}{D^{k-\#\pi \alpha ^{-1}}}}}\frac{\mathrm{tr}\left( \mathrm{P}_{\pi}O \right)}{d^{\#\pi}}<\left( \sum_{\delta ,\alpha \in \mathcal{F} _m}{\left| v_{\delta} \right|\left. \left| v_{\alpha} \right. \right|} \right) \left( \sum_{\delta ,\alpha \in \mathcal{F} _m}{\sum_{\pi \in S_k}{\frac{1}{\left( Dd \right) ^{k-\#\delta \pi ^{-1}}}\frac{1}{D^{k-\#\pi \alpha ^{-1}}}\frac{\mathrm{tr}\left( P_{\delta}O \right)}{d^{\#\pi}}}} \right) . 
\end{align}
Since $\sum_{\delta \in \mathcal{F} _m}{v_{\delta}^{2}}=1$
\begin{equation}
   \left( \sum_{\delta \in \mathcal{F} _m}{\left| v_{\delta} \right|} \right) \left( \sum_{\alpha \in \mathcal{F} _m}{\left| v_{\alpha} \right|} \right) \leqslant \frac{\left| \mathcal{F} _m \right|}{\sqrt{\left| \mathcal{F} _m \right|}}\frac{\left| \mathcal{F} _m \right|}{\sqrt{\left| \mathcal{F} _m \right|}}=\left| \mathcal{F} _m \right|,
\end{equation}
where $\left| \left. \mathcal{F} _m \right| \right. $ denotes the elements number in the set $\mathcal{F} _m$.

The number of permutations with exactly $r$ fixed points is $\binom{k}{r}D_{k-r}$, where \(D_{k-r}\) denotes the number of derangements of \(k-r\) elements.
Since \(D_{k-r}\leq (k-r)!\), we have
\begin{equation}
    \binom{k}{r}D_{k-r}
    \leq
    \binom{k}{r}(k-r)!
    =
    \frac{k!}{r!}
    \leq
    k^{\,k-r}.
\end{equation}
Consequently,
\begin{align}
\left| \mathcal{F} _m \right|&=\sum_{r=m}^k{\binom{k}{r}D_{k-r}}\le \sum_{r=m}^k{k^{\,k-r}}
\\
&=\frac{k^{\,k-m+1}-1}{k-1}<2k^{k-m}.
\end{align}

Next, we calculate:
\begin{equation}
    \sum_{\delta ,\alpha \in \mathcal{F} _m}{\sum_{\pi \in S_k}{\frac{1}{\left( Dd \right) ^{k-\#\delta \pi ^{-1}}}\frac{1}{D^{k-\#\pi \alpha ^{-1}}}\frac{\mathrm{tr}\left( \mathrm{P}_{\pi}O \right)}{d^{\#\pi}}}},
\end{equation}
 For an even fixed point number $m$, with $\delta \in \mathcal{F}_m$ and $\pi$ belong to no fixed point subset $ \mathcal{D} _k$, we can obtain
\begin{equation}
  \underset{\delta \in \mathcal{F} _m,\pi \in \mathcal{D} _k}{\max}\left( \frac{D^{\#\delta \pi ^{-1}}}{D^k} \right) =\frac{D^{m/2+k-m}}{D^k},
\end{equation}
The case where $\#\delta \pi ^{-1}$ is largest occurs when the fixed point part in $\delta$ is multiplied by the disjoint pair in $\pi$, and the other parts of the two happen to be inverses of each other, forming the maximum cycle number $k-m$. In this case, $\#\pi \geqslant m/2+1$. Then
\begin{align}
    &\sum_{\delta ,\alpha \in \mathcal{F} _m}{\sum_{\pi \in \mathcal{D} _k}{\frac{1}{\left( Dd \right) ^{k-\#\delta \pi ^{-1}}}}}\frac{1}{D^{k-\#\pi \alpha ^{-1}}}\frac{\mathrm{tr}\left( \mathrm{P}_{\pi}O \right)}{d^{\#\pi}}
    \\
    &\leqslant \left| \mathcal{F} _m \right|^2\sum_{i=k/2}^1{\left| \mathcal{D} _{k,i} \right|}\underset{\pi \in \mathcal{D} _{k,i},\delta \in \mathcal{F} _m}{\max}\left[ \frac{\left( D^2d \right) ^{\#\delta \pi ^{-1}}}{\left( D^2d \right) ^k}\frac{d^{k/2}}{d^{\#\pi}} \right]   
\\
&\leqslant \left| \mathcal{F} _m \right|^2\left| \mathcal{D} _{k,k/2} \right|\left[ \frac{1}{\left( D^2d \right) ^{m/2}}\sum_{i=k/2}^{m/2+1}{\left( dk^3 \right) ^{k/2-i}}+\sum_{i=m/2}^1{\left( dk^3 \right) ^{k/2-i}\frac{1}{\left( D^2d \right) ^{m/2+\left( m/2+1-i \right)}}} \right]  
\\
&=4\left( k^2 \right) ^{k-m}\frac{k!}{\left( k/2 \right) !2^{k/2}}\frac{1}{\left( D^2d \right) ^{m/2}}\mathcal{O} \left[ \left( dk^3 \right) ^{k/2-m/2} \right] \left[ 1+\mathcal{O} \left( \frac{k^3}{D^2} \right) \right] 
\\
&=\mathcal{O} \left[ \left( d^{1/2}k^{7/2} \right) ^{k-m} \right] \frac{k!}{\left( k/2 \right) !2^{k/2}}\frac{1}{\left( D^2d \right) ^{m/2}},\label{S279}
\end{align}
where in the second line, we used the normalization condition $\mathrm{tr}(\hat{O}^{(1)2})=d$ and in the third line, we sum over different subsets $\mathcal{D}_{k,i}$ where $i=k/2,\cdots ,1$ denotes the cycle number of the element $\pi$ in $\mathcal{D}_{k,i}$. Here the $\sum_{i=m/2}^1{\left( dk^3 \right) ^{k/2-i}\frac{1}{\left( D^2d \right) ^{m/2+\left( m/2+1-i \right)}}}$ corresponds to the part with $\#\pi <m/2+1$,  and it is easy to verify that for $\#\pi <m/2+1$, 
\begin{equation}
 \underset{\pi \in \mathcal{D} _{k,i},\delta \in \mathcal{F} _m}{\max}\left[ \frac{\left( D^2d \right) ^{\#\delta \pi ^{-1}}}{\left( D^2d \right) ^k} \right] =\frac{\left( D^2d \right) ^{i-1+k-m}}{\left( D^2d \right) ^k}=\frac{1}{\left( D^2d \right) ^{m+1-i}},
\end{equation}
which occurs when the  $\#\pi -1$ disjoint pair part in $\pi$ multiplied by the fixed point part in $\delta$ and and the rest are multiplied to produce the maximum number of cycles; the former contributes $i -1$, while the latter contributes at most $k-m$.

For odd fixed point number $m$:
\begin{equation}
    \underset{\delta \in \mathcal{F} _m,\pi \in \mathcal{E} _k}{\max}\left( \frac{D^{\#\delta \pi ^{-1}}}{D^k} \right) =\frac{D^{\left( m-1 \right) /2+k-m}}{D^k},
\end{equation}
The case where $\#\delta \pi ^{-1}$ is largest occurs when the $m-1$ fixed point part in $\delta$ is multiplied by the disjoint pair in $\pi$, and the other parts of the two forming the maximum cycle number $k-m$. In this case, $\#\pi \geqslant (m-1)/2+1$. Then
\begin{align}
&\sum_{\delta ,\alpha \in \mathcal{F} _m}{\sum_{\pi \in \mathcal{D} _k}{\frac{1}{\left( Dd \right) ^{k-\#\delta \pi ^{-1}}}}}\frac{1}{D^{k-\#\pi \alpha ^{-1}}}\frac{\mathrm{tr}\left( \mathrm{P}_{\pi}O \right)}{d^{\#\pi}}\leqslant \left| \mathcal{F} _m \right|^2\sum_{i=k/2}^1{\left| \mathcal{D} _{k,i} \right|}\underset{\pi \in \mathcal{D} _{k,i},\delta \in \mathcal{F} _m}{\max}\left[ \frac{\left( D^2d \right) ^{\#\delta \pi ^{-1}}}{\left( D^2d \right) ^k} \frac{d^{k/2}}{d^{\#\pi}}\right] 
\\
&\leqslant \left| \mathcal{F} _m \right|^2\left| \mathcal{D} _{k,k/2} \right|\left[ \frac{1}{\left( D^2d \right) ^{m/2+1/2}}\sum_{i=k/2}^{m/2+1/2}{\left( dk^3 \right) ^{k/2-i}}+\sum_{i=m/2-1/2}^1{\left( dk^3 \right) ^{k/2-i}\frac{1}{\left( D^2d \right) ^{m/2+1/2+\left( m/2+1/2-i \right)}}} \right] 
\\
&\leqslant 4\left( k^2 \right) ^{k-m}\frac{k!}{\left( k/2 \right) !2^{k/2}}\frac{1}{\left( D^2d \right) ^{m/2}}\mathcal{O} \left[ \left( dk^3 \right) ^{k/2-m/2} \right]  \left[ 1+\mathcal{O} \left( \frac{k^3}{D^2} \right) \right]  
\\
&=\mathcal{O} \left[ \left( d^{1/2}k^{7/2} \right) ^{k-m} \right] \frac{k!}{\left( k/2 \right) !2^{k/2}}\frac{1}{\left( D^2d \right) ^{m/2}}.\label{S285}
\end{align}

Substituting Eq.~(\ref{S285}) and Eq.~(\ref{S279}) into Eq.~(\ref{S268}), we can yield for $\forall m\in \left\{ k-2,\cdots ,0 \right\} $:
\begin{align}
   \langle v_{m}^{\left( i \right)}|G_{dD}X^{-1}X_OG_D|v_{m}^{\left( i \right)}\rangle& =\left( D^2d \right) ^k\sum_{\delta ,\alpha \in \mathcal{F} _m}{\sum_{\pi \in \mathcal{E} _k}{v_{\delta}v_{\alpha}\frac{1}{\left( Dd \right) ^{k-\delta \pi ^{-1}}}}\frac{1}{D^{k-\pi \alpha ^{-1}}}}\frac{\mathrm{tr}\left( P_{\pi}O \right)}{d^{\#\pi}}
\\
&\leqslant \left( D^2d \right) ^k\mathcal{O} \left[ \left( d^{1/2}k^{9/2} \right) ^{k-m} \right] \frac{k!}{\left( k/2 \right) !2^{k/2}}\frac{1}{\left( D^2d \right) ^{m/2}}.
\end{align}
Substituting the above expression into Eq.~(\ref{S267}) and summing over $i$ (different eigenvalues within the same matrix block) yields:
\begin{align}
\sum_i{{\lambda _{f}^{\left( i \right)}}^n\langle v_{m}^{\left( i \right)}|G_{dD}X^{-1}X_OG_D|v_{m}^{\left( i \right)}\rangle}&\leqslant \left| \mathcal{F} _m \right|\frac{d^{n\left[ m+\left( k-m \right) /2 \right]}}{d^{kn}}\left[ 1+\mathcal{O} \left( \frac{nk^2}{D} \right) \right] \left( D^2d \right) ^k\mathcal{O} \left[ \left( d^{1/2}k^{9/2} \right) ^{k-m} \right] \frac{k!}{\left( k/2 \right) !2^{k/2}}\frac{1}{\left( D^2d \right) ^{m/2}}
\\
&=\mathcal{O} \left[ \left( d^{1/2}k^{11/2} \right) ^{k-m} \right] \left[ 1+\mathcal{O} \left( \frac{nk^2}{D} \right) \right] \left( D^2d \right) ^k\frac{1}{d^{n\left( k/2-m/2 \right)}}\frac{k!}{\left( k/2 \right) !2^{k/2}}\frac{1}{\left( D^2d \right) ^{m/2}}
\\
&=\left( D^2d \right) ^{k/2}\mathcal{O} \left[ \left( \frac{k^{11}D^2d^2}{d^n} \right) ^{k/2-m/2} \right] \frac{k!}{\left( k/2 \right) !2^{k/2}},
\end{align}
valid under the condition $nk^2/D\ll1$.
Here $\frac{d^{m+( k-m ) /2}}{d^k}[ 1+\mathcal{O}( \frac{k^2}{D} ) ] $ is the upper bound of the largest eigenvalue in the block matrix corresponding to $m$ fixed points. This follows by applying the Gerschgorin Circle Theorem to the Te matrix in the large-$D$ regime.

\begin{lemma}[Gerschgorin Circle Theorem]
Let $A = (a_{ij}) \in \mathbb{C}^{n \times n}$ be an $n \times n$ matrix. For each $i \in \{1, 2, \dots, n\}$, let $R_i$ denote the sum of the absolute values of the off-diagonal entries in the $i$-th row:
\begin{equation}
    R_i = \sum_{j \neq i} |a_{ij}|.
\end{equation}
Let $D(a_{ii}, R_i) \subseteq \mathbb{C}$ be a closed disk centered at $a_{ii}$ with radius $R_i$, defined as:
\begin{equation}
    D(a_{ii}, R_i) = \{ z \in \mathbb{C} : |z - a_{ii}| \leq R_i \}.
\end{equation}
Then, every eigenvalue $\lambda$ of $A$ lies within at least one of these Gerschgorin disks:
\begin{equation}
    \lambda \in \bigcup_{i=1}^n D(a_{ii}, R_i).
\end{equation}
\end{lemma} 

Evaluating the upper bound for the row sums of the $\mathrm{Te}$ matrix, we yield:
\begin{align}
    \sum_{\alpha \in S_k}{\mathrm{Te}_{\delta ,\alpha}}&=\sum_{\alpha \in S_k}{\sum_{\sigma \in S_k}{\mathrm{Wg}\left( \delta \sigma ^{-1},dD \right) d^{\#\sigma}D^{\#(\sigma \alpha ^{-1})}}}
\\
&\leqslant \sum_{\sigma \in S_k}{\sum_{\alpha \in S_k}{\omega \frac{\left| \mathrm{Moeb}\left( \delta \sigma ^{-1} \right) \right|}{\left( dD \right) ^{k-\#\delta \sigma ^{-1}}}\frac{1}{d^{k-\#\sigma}}\frac{1}{D^{k-\#\sigma \alpha ^{-1}}}}}
\\
&=\omega \frac{\overline{D}}{D^k}\sum_{\sigma \in S_k}{\frac{\left| \mathrm{Moeb}\left( \delta \sigma ^{-1} \right) \right|}{\left( dD \right) ^{k-\#\delta \sigma ^{-1}}}\frac{1}{d^{k-\#\sigma}}}
\\
&=\omega \frac{\overline{D}}{D^k}\sum_{\sigma \in S_k}{\frac{\left| \mathrm{Moeb}\left( \sigma \right) \right|}{\left( dD \right) ^{k-\#\sigma}}\frac{1}{d^{k-\#\sigma \delta}}}
\\
&\leqslant \omega \frac{\overline{D}}{D^k}\sum_{\sigma \in S_k}{\frac{\left| \mathrm{Moeb}\left( \sigma \right) \right|}{\left( dD \right) ^{k-\#\sigma}}\frac{d^{k-\#\sigma +\#\delta}}{d^k}}
\\
&\leqslant \omega \frac{\overline{D}}{D^k}\frac{d^{\#\delta}}{d^k}\sum_{l=0}^{k-1}{\frac{1}{l!}\left( \frac{2k^2}{D} \right)}^l
\\
&=\frac{d^{\#\delta}}{d^k}\left( 1+\mathcal{O} \left( \frac{k^2}{D} \right) \right).\label{Terowsummationupperbound}
\end{align}
The maximal element in each row lies on the diagonal, given by:
\begin{align}
    \mathrm{Te}_{\delta ,\delta}&=\sum_{\sigma \in S_k}{\mathrm{Wg}\left( \delta \sigma ^{-1},dD \right) d^{\#\sigma}D^{\#(\sigma \delta ^{-1})}}
\\
&\leqslant \omega \sum_{\sigma \in S_k}{\frac{\left| \mathrm{Moeb}\left( \delta \sigma ^{-1} \right) \right|}{\left( dD^2 \right) ^{k-\#\delta \sigma ^{-1}}}\frac{1}{d^{k-\#\sigma}}}
\\
&\leqslant \omega \frac{d^{\#\delta}}{d^k}\sum_{\sigma \in S_k}{\frac{\left| \mathrm{Moeb}\left( \delta \sigma ^{-1} \right) \right|}{\left( D^2 \right) ^{k-\#\delta \sigma ^{-1}}}}
\\
&\leqslant \omega \frac{d^{\#\delta}}{d^k}\sum_{l=0}^{k-1}{\frac{1}{l!}\left( \frac{2k^2}{D^2} \right)}^l
\\
&=\frac{d^{\#\delta}}{d^k}\left( 1+\mathcal{O} \left( \frac{k^2}{D^2} \right) \right). \label{S353}
\end{align}

Since $\max(\#\delta) = m + (k-m)/2$ for permutations $\delta$ with $N_f(\delta) = m$ fixed points, Eq.~(\ref{Terowsummationupperbound}) and (\ref{S353}) together imply that, in the large-$D$ regime, the eigenvalues of the block matrix with $m$ fixed points are upper bounded by $\frac{d^{( k+m ) /2}}{d^k}[ 1+\mathcal{O} ( \frac{k^2}{D} ) ]$, valid under the condition $k^2/D\ll 1$.

Finally, returning to Eq.~(\ref{S292}), the remainder terms in $\mathbb{E}[C_l(\hat{O})^k]$ are given by

\begin{align}
    \sum_{\lambda \ne 1}{\left. \mathrm{tr}\!\left( \mathrm{ATe}^n \right) \right|_{v_{f}^{\left( i \right)}}}&=\sum_m{\sum_i{\frac{{\lambda _{m}^{\left( i \right)}}^n\langle v_{m}^{\left( i \right)}|G_{dD}\mathcal{I} _{\mathcal{E} _k}G_D|v_{m}^{\left( i \right)}\rangle}{\langle v_{m}^{\left( i \right)}|G_{dD}X^{-1}G_{dD}|v_{m}^{\left( i \right)}\rangle}}}
\\
&=\frac{k!}{\left( k/2 \right) !2^{k/2}}\frac{1}{\left( D^2d \right) ^{k/2}}\mathcal{O} \left( \frac{k^{11}D^2}{d^n} \right) \left[ 1+\mathcal{O} \left( \frac{k^2d^{1/2}}{D} \right) \right]  .
\end{align}

Combining this with the dominant term in Eq.~(\ref{S101}), we have:
\begin{equation}
    \mathbb{E} \left( C_l\left( \hat{O} \right) ^k \right)  =\frac{k!}{\left( k/2 \right) !2^{k/2}}\frac{1}{\left( D^2d \right) ^{k/2}}\left[ 1+\mathcal{O} \left( \frac{k^3}{D^2d} \right) +\mathcal{O} \left( \frac{\mathrm{poly}\left( k \right) D^2}{d^n} \right) \right] .
\end{equation}
where $k^3$ arises from the ratio in sequence $\{ \mathcal{D} _{k,i}\}_{i=0}^{k/2}$. 

Combining the conditions established in the preceding derivations, we find that in the asymptotic regime where $\frac{nk^2}{D} \ll 1$ and $\frac{\mathrm{poly}(k) D^2}{d^n} \ll 1$, the $k$th-order moment of the distribution of $C(\hat{O}_l)$ deviates from that of the Gaussian distribution $\mathcal{N}\left( 0, \frac{1}{D^2d} \right)$ by at most a residual correction of $\mathcal{O}\left( \frac{k^3}{D^2d} \right) + \mathcal{O}\left( \frac{\mathrm{poly}(k) D^2}{d^n} \right)$.

Similarly, if we consider operators in the space of Clifford generators, we can obtain a more compact finite correction: in the asymptotic regime where $ \frac{nk^2}{D}\ll 1$ and $\frac{\mathrm{poly}\left( k \right) D^2}{d^n} \ll 1$, the $k$th-order moment of the distribution $C(\hat{O_l})$ deviates from that of the pure Gaussian distribution $\mathcal{N}\left( 0,\frac{1}{D^2d} \right) $ by at most a residual correction of $\mathcal{O} \left( \frac{k^2}{D^2d} \right) +\mathcal{O} \left( \frac{\mathrm{poly}\left( k \right) D^2}{d^n} \right) $. This is consistent with the results presented in our main text.

\subsection{Residual Terms for $C( \hat{O}_g )$ in the Large-$(D,n)$ Regime}\label{reslargeDNglobal}

In this subsection, we derive rigorous upper bounds for the residual terms of $\mathbb{E}(C(\hat{O}_g)^k)$ in the large-$(D,n)$ regime for every operator $\hat{O}_g$ within the space of Clifford generator, confirming that their ratio to the leading-order contribution vanishes as $(D,n) \to \infty$.

We calculate the $k$th-order moment of $C( \hat{O}_g ) $
which can be expanded as (throughout the remainder of this subsection, we suppress the subscript $g$ for notational simplicity.):
\begin{equation}\label{eq:O_k_moment}
\mathbb{E}
\left(
C(\hat{O})^k
\right)
=
\frac{1}{c^k}
\sum_{(j_1,\ldots,j_k)\in\mathcal{J}^k}
c_{j_1}\cdots c_{j_k}
\,
\mathbb{E}
\left[
\prod_{a=1}^{k}
C(\hat{Q}_{j_a})
\right].
\end{equation}
where $\mathcal J$ is the operator bases index set $\left\{ 1,2,\cdots ,m \right\} $. The mixed moment of the operator bases can be expressed as
\begin{equation}
\mathbb{E}
\left[
\prod_{a=1}^{k}
C(\hat{Q}_{j_a})
\right]
=
\mathrm{tr}
\left(
\prod_{s=1}^{n}\mathrm{A}^{(s)}
\right),
\label{eq:mixed_moment_matrix}
\end{equation}
where $\mathrm{A}^{(s)}=\sum_{\pi \in S_k}{\mathrm{Wg}\left( \delta \pi ^{-1},dD \right) \mathrm{tr}\left( \mathrm{P}_{\pi}Q_s \right) D^{\#\left( \pi \tau ^{-1} \right)}}$
which is determined by $Q_s
=
\hat{Q}_{j_1}^{(s)}
\otimes
\hat{Q}_{j_2}^{(s)}
\otimes\cdots\otimes
\hat{Q}_{j_k}^{(s)}$. In large-$D$ regime, evaluating the upper bound for the row sum of the $\mathrm{A}^{(s)}$ matrix, we yield:
\begin{align}
\sum_{\tau \in S_k}{{\mathrm{A}^{(s)}}_{\delta ,\tau}}&=\sum_{\tau \in S_k}{\sum_{\pi \in S_k}{\mathrm{Wg}\left( \delta \pi ^{-1},dD \right) \mathrm{tr}\left( \mathrm{P}_{\pi}Q_s \right) D^{\#\left( \pi \tau ^{-1} \right)}}}
\\
&=\frac{\mathrm{tr}\left( \mathrm{P}_{\delta}Q_s \right)}{d^k}\left[ 1+\mathcal{O} \left( \frac{\,\,k^2}{D^2} \right) \right] +\sum_{\tau ,\pi \in S_k,\left( \tau ,\pi \right) \ne \left( \delta ,\delta \right)}{\mathrm{Wg}\left( \delta \pi ^{-1},dD \right) \mathrm{tr}\left( \mathrm{P}_{\pi}Q_s \right) D^{\#\left( \pi \tau ^{-1} \right)}}
\\
&\leqslant \frac{\mathrm{tr}\left( \mathrm{P}_{\delta}Q_s \right)}{d^k}\left[ 1+\mathcal{O} \left( \frac{\,\,k^2}{D^2} \right) \right] +\frac{\overline{D}}{D^k}\sum_{\pi \in S_k\backslash \{\delta \}}{\frac{\left| \mathrm{Moeb}\left( \delta \pi ^{-1} \right) \right|}{\left( Dd \right) ^{k-\#\delta \pi ^{-1}}}}\frac{\mathrm{tr}\left( \mathrm{P}_{\pi}Q_s \right)}{d^k}
\\
&\leqslant \frac{\mathrm{tr}\left( \mathrm{P}_{\delta}Q_s \right)}{d^k}\left[ 1+\mathcal{O} \left( \frac{\,\,k^2}{D^2} \right) \right] +\frac{\overline{D}}{D^k}\frac{1}{d^{k/2}}\sum_{\pi \in S_k\backslash \{\delta \}}{\frac{\left| \mathrm{Moeb}\left( \delta \pi ^{-1} \right) \right|}{\left( Dd \right) ^{k-\#\delta \pi ^{-1}}}}
\\
&\leqslant \frac{\mathrm{tr}\left( \mathrm{P}_{\delta}Q_s \right)}{d^k}\left[ 1+\mathcal{O} \left( \frac{\,\,k^2}{D^2} \right) \right] +\frac{\overline{D}}{D^k}\frac{1}{d^{k/2}}\sum_{l=1}^{k-1}{\frac{1}{l!}\left( \frac{2k^2}{Dd} \right) ^l}
\\
&=\frac{\mathrm{tr}\left( \mathrm{P}_{\delta}Q_s \right)}{d^k}+\frac{1}{d^{k/2}}\mathcal{O} \left( \frac{k^2}{D} \right). \label{380}
\end{align}
In the fourth line of the expression above, since every local Pauli operator is traceless $\mathrm{tr}
(
\hat{Q}_{j}^{(s)}
)=0$,
any permutation $\delta$ containing a fixed point gives $\mathrm{tr}(\mathrm{P}_{\delta}Q_s)=0$. For permutations without fixed points, the trace is bounded by $\left|
\mathrm{tr}(\mathrm{P}_{\delta}Q_s)
\right|
\leq
d^{k/2}$.

Equation~(\ref{380}) shows that in the large-$D$ regime, the matrix $\mathrm{A}^{(s)}$ can be decomposed into a diagonal matrix and a small correction matrix whose row sums are bounded in terms of $D$. We decompose each matrix as
\begin{equation}
\mathrm{A}^{(s)}
=
\mathrm{A}^{(s)}_{\mathrm{d}}
+
\mathrm{A}^{(s)}_{\mathrm{od}},
\end{equation}
where $\mathrm{A}_{\mathrm{d}}^{(s)}=\mathrm{diag}\left( \mathrm{tr}\left( \mathrm{P}_{\delta}Q_s \right) /d^k \right) $. In Sec.~\ref{sectionglobal}, we have already given the contribution of the diagonal part of these mixed moments to Eq.~(\ref{eq:O_k_moment}): 
\begin{equation}
M_{\mathrm{d}}=\frac{k!}{\left( k/2 \right) !2^{k/2}}\frac{1}{d^{kn/2}}\left[ 1+\mathcal{O} \left( \frac{k^2}{d^n} \right) \right].
\label{eq:diagonal_correction}
\end{equation}

We next estimate the contribution from the off-diagonal matrix
elements. The off-diagonal contribution is defined as
\begin{equation}
\Delta_{\mathrm{od}}
=
\mathrm{tr}
\left(
\prod_{s=1}^{n}
(A^{(s)}_{\mathrm{d}}
+
A^{(s)}_{\mathrm{od}})
\right)
-
\mathrm{tr}
\left(
\prod_{s=1}^{n}
A^{(s)}_{\mathrm{d}}
\right).
\end{equation}
Using the submultiplicativity of the infinity norm ($\parallel \mathrm{A}\parallel _{\infty}=\underset{i}{\max}(\sum_j{\left| \left. \mathrm{A}_{ij} \right| \right.})$) $\|AB\|_{\infty} \le \|A\|_{\infty} \|B\|_\infty$, we yield: 
\begin{align}
    \mathrm{tr}\left( \prod_{s=1}^i{\mathrm{A}_{\mathrm{d}}^{(s)}}\prod_{i+1}^n{\mathrm{A}_{\mathrm{od}}^{(s)}} \right) &\leqslant k!\left( \left\| \mathrm{A}_{\mathrm{d}}^{(s)} \right\| _{\infty} \right) ^i\left( \left\| \mathrm{A}_{\mathrm{od}}^{(s)} \right\| _{\infty} \right) ^{n-i}
\\
&=k!\frac{1}{d^{kn/2}}\left[ \mathcal{O} \left( \frac{k^2}{D} \right) \right] ^{n-i},
\end{align}
where we used $\|
\mathrm{A}^{(s)}_{\mathrm{d}}
\|_{\infty}
\leq
d^{-k/2}$ and $\|
\mathrm{A}^{(s)}_{\mathrm{od}}
\|_{\infty}
\leq
d^{-k/2}
\mathcal{O}
\left(
\frac{k^2}{D}
\right)$.
Summing over all positions of $\mathrm{A}_{\mathrm{od}}^{(s)}$, in the large-$D$ regime satisfying $nk^2/D\ll 1$, yields
\begin{equation}
\left| \Delta _{\mathrm{od}} \right|=k!\frac{1}{d^{kn/2}}\mathcal{O} \left( \frac{nk^2}{D} \right) .
\label{eq:off_diagonal_correction}
\end{equation}

Combining the diagonal and off-diagonal contributions, we yield:
\begin{equation}
\mathbb{E} \left( C(\hat{O})^k \right) =M_{\mathrm{d}}+\frac{1}{c^k}\sum_{(j_1,...,j_k)\in \mathcal{J} ^k}{c_{j_1}}\cdots c_{j_k}\left| \Delta _{\mathrm{od}} \right|,
\label{eq:mixed_final}
\end{equation}
Given that the normalization factor $c^2=\sum_{j=1}^m{c_{j}^{2}}$, we yield
\begin{equation}
 \frac{1}{c^k}\sum_{(j_1,...,j_k)\in \mathcal{J} ^k}{c_{j_1}}\cdots c_{j_k}\left| \Delta _{\mathrm{od}} \right|\leqslant k!\frac{1}{d^{kn/2}}\mathcal{O} \left( \frac{m^{k/2}nk^2}{D} \right)  ,
\end{equation}
where $m=\left. \left| \mathcal{J} \right. \right|$. Finally, we can obtain
\begin{equation}
    \mathbb{E} \left( C(\hat{O})^k \right) =\frac{k!}{\left( k/2 \right) !2^{k/2}}\frac{1}{d^{kn/2}}\left[ 1+\mathcal{O} \left( \frac{k^2}{d^n} \right) +\mathcal{O} \left( \frac{nm^{k/2}k!}{D} \right) \right]. 
\end{equation}

Consequently, we have demonstrated that in the large-$(D,n)$ limit, $C(\hat{O}_g)$ exhibits asymptotic Gaussianity for operators of the form $\hat{O} = \frac{1}{c}\sum_{j=1}^m c_j \hat{Q}_j$, thereby establishing that the sequence $\mathscr{C}_g = \{ C_g(\hat{O}_i) \}_{i=1}^p$ constitutes a GP. We further show that in the asymptotic regime where $nm^{k/2}k!/D\ll1$ and $k^2/d^n\ll1$, the $k$th-order moment of the distribution $C( \hat{O}_g ) $ deviates
from that of the pure Gaussian distribution $\mathcal{N} \left( 0,1/d^n \right) $ by at most a residual correction of $\mathcal{O}(k^2/d^n)+\mathcal{O}(nm^{k/2}k!/D)$.

\end{document}